\documentclass{amsart}
\usepackage{fdsymbol} 
\usepackage{bbm}
\usepackage[all,cmtip]{xy}
\usepackage{stmaryrd} 

\usepackage{fonttable}

\makeindex \topmargin=0in \headheight=8pt \headsep=0.5in \textheight=8.9in
\usepackage{etoolbox}
\patchcmd{\section}{\scshape}{\large\bfseries}{}{}
\makeatletter
\renewcommand{\@secnumfont}{\bfseries}
\makeatother
    
\usepackage[maxnames=99]{biblatex}
\usepackage{booktabs}
\bibliography{references}

\numberwithin{equation}{section}

\newtheorem{thma}{Theorem}
\newtheorem{thmas}{Theorem}[section]
\newtheorem{lemma}{Lemma}[section]
\newtheorem{cor}{Corollary}[section]
\newtheorem{prop}{Proposition}[section]

\newtheorem{defin}{Definition}[section]
\theoremstyle{remark}
\newtheorem{rem}[lemma]{Remark}

\usepackage{xcolor}
\newcommand{\rd}[1]{\leavevmode{\color{red}{#1}}}

\usepackage{eucal}
\usepackage{enumitem}

\usepackage{hyperref}
\hypersetup{
    colorlinks=true,
    linkcolor=blue,
    filecolor=magenta,      
    urlcolor=cyan,
    citecolor=blue,
}

\usepackage[normalem]{ulem}

\newcommand{\nc}{\newcommand}

\def\eps{\varepsilon}   
\def\vphi{\varphi}
\def\smm{\smallsetminus}

\nc{\Aa}{{\mathcal{A}}}
\nc{\Bb}{{\mathcal{B}}}
\nc{\Cc}{{\mathcal{C}}}
\nc{\Dd}{{\mathcal{D}}}

\nc{\Ff}{{\mathcal{F}}}
\nc{\Gg}{{\mathcal{G}}}
\nc{\Hh}{{\mathcal{H}}}
\nc{\Ii}{{\mathcal{I}}}
\nc{\Jj}{{\mathcal{J}}}
\nc{\Kk}{{\mathcal{K}}}
\nc{\Ll}{{\mathcal{L}}}
\nc{\Mm}{{\mathcal{M}}}
\nc{\Nn}{{\mathcal{N}}}
\nc{\Oo}{{\mathcal{O}}}
\nc{\Pp}{{\mathcal{P}}}
\nc{\Qq}{{\mathcal{Q}}}
\nc{\Rr}{{\mathcal{R}}}
\nc{\Ss}{{\mathcal{S}}}
\nc{\Tt}{{\mathcal{T}}}
\nc{\Uu}{{\mathcal{U}}}
\nc{\Vv}{{\mathcal{V}}}
\nc{\Ww}{{\mathcal{W}}}
\nc{\Xx}{{\mathcal{X}}}
\nc{\Yy}{{\mathcal{Y}}}
\nc{\Zz}{{\mathcal{Z}}}

\nc{\mA}{{\mathrm{A}}}
\nc{\mB}{{\mathrm{B}}}
\nc{\mC}{{\mathrm{C}}}
\nc{\mD}{{\mathrm{D}}}
\nc{\mE}{{\mathrm{E}}}
\nc{\mF}{{\mathrm{F}}}
\nc{\mG}{{\mathrm{G}}}
\nc{\mH}{{\mathrm{H}}}
\nc{\mI}{{\mathrm{I}}}
\nc{\mJ}{{\mathrm{J}}}
\nc{\mK}{{\mathrm{K}}}
\nc{\mL}{{\mathrm{L}}}
\nc{\mM}{{\mathrm{M}}}
\nc{\mN}{{\mathrm{N}}}
\nc{\mO}{{\mathrm{O}}}
\nc{\mP}{{\mathrm{P}}}
\nc{\mQ}{{\mathrm{Q}}}
\nc{\mR}{{\mathrm{R}}}
\nc{\mS}{{\mathrm{S}}}
\nc{\mT}{{\mathrm{T}}}
\nc{\mU}{{\mathrm{U}}}
\nc{\mV}{{\mathrm{V}}}
\nc{\mW}{{\mathrm{W}}}
\nc{\mX}{{\mathrm{X}}}
\nc{\mY}{{\mathrm{Y}}}
\nc{\mZ}{{\mathrm{Z}}}

\nc{\BB}{{\mathbb{B}}}
\nc{\CC}{{\mathbb{C}}}
\nc{\DD}{{\mathbb{D}}}
\nc{\EE}{{\mathbb{E}}}
\nc{\FF}{{\mathbb{F}}}
\nc{\GG}{{\mathbb{G}}}
\nc{\HH}{{\mathbb{H}}}
\nc{\II}{{\mathbb{I}}}
\nc{\JJ}{{\mathbb{J}}}
\nc{\KK}{{\mathbb{K}}}
\nc{\LL}{{\mathbb{L}}}
\nc{\MM}{{\mathbb{M}}}
\nc{\NN}{{\mathbb{N}}}
\nc{\OO}{{\mathbb{O}}}
\nc{\PP}{{\mathbb{P}}}
\nc{\QQ}{{\mathbb{Q}}}
\nc{\RR}{{\mathbb{R}}}

\nc{\TT}{{\mathbb{T}}}
\nc{\UU}{{\mathbb{U}}}
\nc{\VV}{{\mathbb{V}}}
\nc{\WW}{{\mathbb{W}}}
\nc{\XX}{{\mathbb{X}}}
\nc{\YY}{{\mathbb{Y}}}
\nc{\ZZ}{{\mathbb{Z}}}

\DeclareMathOperator{\Tr}{\mathrm{Tr}}

\DeclareMathOperator{\SL}{\mathrm{SL}}

\DeclareMathOperator{\Id}{\mathrm{Id}}

\DeclareMathOperator{\supp}{\mathrm{supp}}

\DeclareMathOperator{\length}{\mathrm{length}}

\def\dist{\mathrm{dist}}

\def\diam{\mathrm{diam}}
\DeclareMathOperator{\Var}{\mathrm{Var}}

\def\loc{\mathrm{loc}}

\let\Re\relax
\let\Im\relax
\DeclareMathOperator{\Re}{\mathrm{Re}}
\DeclareMathOperator{\Im}{\mathrm{Im}}

\nc{\Su}{\overline{M}}
\nc{\oSuc}{\oSu^\delta}
\nc{\oSu}{M}
\nc{\Sur}{\Sigma}

\nc{\tiling}{\Diamond}
\nc{\G}{G_\delta}
\nc\K{K^\delta}
\nc\uS{\underline{S}^\delta}
\nc{\dG}{G^\dagger_\delta}

\nc{\w}{\mathrm{w}}

\nc{\fl}{\omega}
\nc{\flow}{\mathrm{flow}}

\nc{\dif}{\mathbf{w}}
\nc{\sqdif}{\mathbf{\varpi}}

\nc{\Div}{D}

\nc{\Bp}{\mathrm{B}}

\nc{\pf}{\mathrm{E}}

\nc{\spinor}{\varsigma}

\nc\inst{l}
\nc\insth[1]{\inst_{h_{#1}}}
\nc\iinst[1]{\inst^{\mathrm{int}}_{h_{#1}}}
\nc{\instm}{\inst_{\dif}}

\nc\class[1]{[#1]}

\nc\chr[2]{\mbox{\small$\left[ #1 \atop #2  \right]$}}

\def\cst{\mathrm{cst}}

\nc{\Fun}{\mathrm{Fun}}

\nc{\indic}{1\!\!1}
\nc{\sign}{\epsilon}

\nc{\CLE}{\mathrm{CLE}}
\nc{\Llor}{\Ll^{\shortrightarrow}}

\nc{\Kinv[1]}{\mathcal{D}_{#1}}
\nc{\Distr[1]}{\mathcal{D}'_\loc\left(#1\right)}
\nc{\Cbw}{\mathcal{C}}
\nc{\Aring}{\mathcal{R}}
\nc{\sss}{p(s)}
\nc{\Cpd}{\mathbb{C}^+_\delta}
\nc{\fOd}{\delta}
\nc{\cyclic}{\mathcal{C}}
\nc{\Kterm[1]}{\mathcal{K}_{#1}^{-1}}
\nc{\NOd}{N_\delta}
\nc{\NCd}{N_\delta}

\nc{\Adj}{\mathrm{Adj}}

\DeclareMathOperator{\CR}{\mathrm{CR}}

\begin{document}

\title[Nesting of double-dimer loops: local fluctuations and convergence of the field]
{Nesting of double-dimer loops: local fluctuations and convergence to the nesting field of $\mathrm{CLE}(4)$}

\author[Mikhail Basok]{Mikhail Basok$^\mathrm{a}$}
\author[Konstantin Izyurov]{Konstantin Izyurov$^\mathrm{a}$}

\thanks{\textsc{${}^\mathrm{A}$ University of Helsinki.}}

\keywords{dimer model, nesting field, conformal loop ensemble, discrete multivalued functions}

\begin{abstract} We consider the double-dimer model in the upper-half plane discretized by the square lattice with mesh size $\delta$. For each point $x$ in the upper half-plane, we consider the random variable $N_\delta(x)$ given by the number of the double-dimer loops surrounding this point. We prove that the normalized fluctuations of $N_\delta(x)$ for a fixed $x$ are asymptotically Gaussian as $\delta\to 0+$. Further, we prove that the double-dimer nesting field $N_\delta(\cdot) - \EE N_\delta(\cdot)$, viewed as a random distribution in the upper half-plane, converges as $\delta\to 0+$ to the nesting field of $\CLE(4)$ constructed by Miller, Watson and Wilson~\cite{miller2015conformal}.
\end{abstract}

\maketitle

\tableofcontents

\section{Introduction}
\label{sec:introduction}

The dimer and double-dimer models are among the most studied models in planar statistical mechanics in recent years. To name a few, it has connections to random matrices and interacting particle systems \cite{johansson2002non}, algebraic geometry \cite{kenyon2006dimers,kenyon2007limit}, calculus of variations and PDE \cite{cohn2001variational, astala2020dimer}, integrable systems \cite{goncharov2013dimers}. A facet of particular interest for us is that the dimer model was one of the first models where conformal invariance of the scaling limit and its relation to the two-dimensional conformal field theory were rigorously established, when Kenyon proved the convergence of the height function to the Gaussian free field \cite{KenyonConfInvOfDominoTilings, KenyonGFF} . 

Since its introduction by Schramm \cite{schramm2000scaling}, Stochastic Loewner evolutions (SLE) and its relatives, Conformal loop ensembles (CLE) \cite{sheffield2012conformal} have taken a central place in the study of critical planar systems. With regard to the double-dimer model, a conjecture attributed to Kenyon in \cite{rohde2011basic} states that interfaces and loop ensembles in the double-dimer model converge to SLE${}_4$ and CLE${}_4$ respectively. 
This conjecture is widely expected to be true, and reasonable progress toward proving it was made during the last decade~\cite{KenyonConfInvOfDominoTilings,DubedatDoubleDimers,BasokChelkak,bai2023crossing}, culminating in a proof of convergence of probabilities of cylindrical events for the double-dimer loop laminations in the upper half-plane to the corresponding $\CLE_4$ quantities. However, in its full strength, the conjecture remains open, mainly because unlike many other models, the double-dimer model lacks an appropriate Russo--Seymour--Welsh theory, and, consequently, \emph{a priori} precompactness results for curves and loops.

A loop ensemble in a discrete domain $\Omega_\delta$ can be encoded by its nesting function $N_\delta(\cdot)$, where $\NOd(x),$ $x\in\Omega_\delta$, counts the number of loops in an ensemble surrounding $x$. While in the discrete, the two descriptions are equivalent, the question of the asymptotic behavior of $N_\delta$ as $\delta\to 0$ is different from the question of convergence of the loop ensemble; one may say that it concerns a different mode of convergence. This question is of particular interest for the double-dimer model given that the convergence of curves is not yet fully established.

In this paper, we state our main results for the double-dimer model in the upper half-plane $\CC^+_\delta:=\CC^+\cap\delta\ZZ^2$, where $\CC^+=\{z:\Im z>0\}$ and $\delta>0$ is the mesh size. This model is defined as a suitable infinite volume limit, see Section~\ref{sec:Height_function_loop_statistics_and_monodromy}. For extensions to Temperleyan domains, see the discussion at the end of the present introduction. Our first result concerns the behavior of $N_\delta$ at a single bulk point $v$:

\begin{thma}
    \label{thma:main1}
    Consider the double-dimer model in $\CC^+_\delta$, let $v$ be a fixed point in the upper half-plane, and denote by $\NCd(v)$ the number of double-dimer loops encircling $v$. Then we have \begin{enumerate}
        \item The average number of encircling loops has the asymptotics $\mu_\delta \coloneqq \EE \NCd(v) = -\frac{1}{\pi^2}\log \delta + O(1)$,
        \item The variance of the number of loops has the asymptotics $\sigma_\delta^2 \coloneqq \Var  \NCd(v) = -\frac{2}{3\pi^2}\log \delta + O(1)$,
        \item And we have for the normalized fluctuations:
        \[
        \frac{\NOd(v) - \mu_\delta}{\sigma_\delta} \xrightarrow{\text{in distribution}}\Nn(0,1),\quad \delta\to 0+.
        \]
    \end{enumerate}
\end{thma}
These results, including the constants $\frac{1}{\pi^2}$ and $\frac{2}{3\pi^2}$, are (expected to be) universal, and, in fact, they are easily seen to hold for $\CLE_4$, when one replaces $\NCd(v)$ with the number of $\CLE_4$ loops encircling $x$ whose (conformal) radius seen from $x$ is at least $\delta$. This is a consequence of the fact that $-\log \mathrm{CR}(\gamma_n,x)$, where $\gamma_n$ is the $n$-th outermost loop surrounding $x$ and $\mathrm{CR}$ denotes the conformal radius, is a sum of i.i.d. random variables with known distribution, see~\cite{schramm2009conformal}. Note, however, that since Theorem~\ref{thma:main1} concerns loops on all scales up to the lattice scale, this type of results do not follow easily from convergence to CLE and CLE computations.

Our second aim is to study $\NCd(x)$ as $x$ varies. In this case, we consider the fluctuations of the  \emph{double-dimer nesting field} $\vphi_\delta(x) = \NCd(x) - \EE \NCd(x)$ without further normalization. Obviously, $\vphi_\delta$ does not have a point-wise limit, but we can consider $\vphi_\delta$ as a field, more precisely, as a random generalized functions acting on test functions by integration. A natural candidate for the scaling limit of $\vphi_\delta$ is the \emph{$\CLE{}_4$ nesting field}, which we denote by $\varphi$. In~\cite{miller2015conformal} Miller, Watson and Wilson introduced and studied the so-called nesting fields for $\CLE{}_\kappa$ for $\kappa\in (8/3,8)$. They constructed these fields as the limit $\varphi(\cdot):=\lim_{\eps\to 0} (N^{B\eps}(\cdot)-\EE N^{B\eps}(\cdot))$, where $N^{B\eps}(x)$ is the number of loops in CLE surrounding a disc of radius $\eps$ centered at $x$. They proved that the limit exists almost surely in the topology of any Sobolev space $H^{-2-\nu}_\loc$ ($\nu>0$), and is conformally invariant; in fact, as we explain in Section \ref{subsec:nesting_fields_for_CLE(4)}, their arguments imply that the convergence in probability holds with respect to a stronger topology of $H^{-1-\nu}_\loc$. Our second result is as follows:

\begin{thma}
    \label{thma:main2}
  Let $\vphi_\delta = \NCd - \EE \NCd$ denote the double-dimer nesting field sampled in the discrete upper half-plane $\CC^+_\delta = \delta \ZZ^2\cap \CC^+$. Then for any $\nu>0$, the fields $\vphi_\delta$ converge to the CLE${}_4$ nesting field $\vphi$ in distribution with respect to the topology of the local Sobolev space $H^{-1-\nu}_\loc(\CC^+)$.
\end{thma}

The convergence in distribution in $H^{-1-\nu}_\loc(\CC^+)$ is equivalent to the convergence in distribution in $H^{-1-\nu}_\loc(K)$ for each open, relative compact $K\subset \CC^+$, see the proof of Theorem 2.

Let us discuss briefly our approach to the problem. As it is apparent nowadays, the dimer and double-dimer models are very hard to approach via soft probabilistic methods. The main obstacle is the lack of a nice domain Markov property which is caused by an extreme sensitivity of the model to the boundary conditions. The situation is better if the dimer model is sampled in a Temperleyan domain; in this case dimer configurations appear to be in a correspondence with spanning trees taken on the initial graph $\Gamma$, which allows to reduce the study of the dimer model to the UST model. Among other, this approach was employed in the series of papers~\cite{berestycki2020dimers},~\cite{BerestyckiLaslierRayI},~\cite{BerestyckiLaslierRayII},~\cite{benoit2022does},~\cite{laslier2021central} leading to many strong results concerning the dimer height function. Unfortunately, this machinery does not yet permit to analyze the behavior of double-dimer loops precisely enough to study their convergence: the double-dimer loops do not seem to be natural combinatorial objects arising directly from a pair of UST, while the correspondence between values of the height function and the loops is too subtle.

On the other hand, the dimer model in any planar domain is known to have a very rich algebraic structure due to the Kasteleyn theorem. In the case of graphs that are regular enough, such as lattices, this leads to a connection between the dimer model and discrete complex analysis (see~\cite{CLR1} for more details on this subject). This framework has been used broadly to study the planar dimer model starting from the landmark works of Kenyon~\cite{KenyonConfInvOfDominoTilings,KenyonGFF}, as well as numerous others some of which are~\cite{russkikh2018dimers, russkikh2020dominos, KLRR, CLR1, CLR2, chelkak2020fluctuations, berggren2024perfect, berggren2024perfectLozenge}. Due to the nature of this approach, some quantities that possess algebraic expressions in terms of the Kasteleyn matrix, such as point-wise height correlations or probabilities of topological events for double-dimer loops, are analyzed surprisingly precisely, while other very natural probabilistic aspects, such as crossing estimates for double-dimer loops, are still lacking satisfactory understanding, because it is not natural to approach them algebraically. As it turns out, nesting fields belong to the first type of objects. In particular, the starting point of our proof of Theorem~\ref{thma:main1} is the combinatorial relation
\[
    \EE \exp(i th_\delta(x)) = \EE (\cos t)^{\NOd(x)}
\]
between the Laplace transform of $\NOd(x)$ and the Fourier transform of the double-dimer height function at $x$. It was shown by Pinson~\cite{pinson2004rotational} that the `electric correlator' $\EE \exp(i t h_\delta(x))$ can be expressed via the determinant of the Kasteleyn matrix acting on the space of multivalued functions with monodromy $e^{it}$ around the face containing $x$ (see Section~\ref{sec:Height_function_loop_statistics_and_monodromy} for details). This observation was further developed in the work~\cite{DubedatFamiliesOfCR} by Dub\'edat to study correlations of the form $\EE \exp(i t_1h_\delta(x_1))\ldots\exp(i t_kh_\delta(x_k))$. The key point of Dub\'edat's work is computation of the near-diagonal asymptotics of the inverse Kasteleyn opeartor with monodromy: this asymptotics is used to control logarithmic variation of the aforementioned correlation when one of the points $x_i$ is moved.

In our proof we use a similar approach to the one of Dub\'edat, but in our case the point $x$ is fixed while the parameter $t$ is the one that varies. We use the same combinatorial idea to link the near-diagonal asymptotics of the inverse Kasteleyn operator with monodromy and the logarithmic variation of the Laplace transform $\EE (\cos t)^{\NOd(x)}$ with respect to $t$. Improving Dub\'edat's results on the asymptotics, we control the Laplace transform precisely enough to obtain the convergence result.

The main input for our proof of Theorem~\ref{thma:main2} is the known results~\cite{kenyon2014conformal, DubedatDoubleDimers, BasokChelkak, bai2023crossing} on convergence of probabilities of cylindrical events for the double-dimer laminations to the corresponding CLE${}_4$ probabilities, see Theorem~\ref{thmas:on_convergence_of_cyl_prob}. The first step is to modify the regularization procedure of \cite{miller2015conformal} used to define the nesting field, replacing $N^{B\eps}(x)$, i.e., the number of loops surrounding a disc to radius $\eps$ around $x$, with $N(x,x+\eps)$, the number of loops surrounding both $x$ and $x+\eps$. We show that this modification leads to the same field in the limit. The advantage is that correlations of its double-dimer counterpart, $N_\delta(x,x+\eps)$, can be now expressed in terms of cylindrical events, and thus converge as $\delta\to 0$ for a fixed $\eps$. We stress however that this alone is not sufficient to derive the convergence of nesting field, as there is a non-trivial exchange of limits involved. Put differently, since $\NOd(v)-\EE \NOd(v)$ counts all the loops surrounding $v$ up to the lattice mesh size, it does not appear to be possible to derive its convergence from convergence of \emph{macroscopic} loops to CLE${}_4$ alone. 

Our way to deal with this problem is highly model-specific, namely, we isolate the singular (as $\eps\to 0$) part of correlations of $\NOd(\cdot)$ in terms of the correlations of the \emph{square} of the height function, see Section~\ref{subsec:combinatorial_corresp_betwen_nesting_and_height}. The regularization procedure for the latter is similar to Wick normal ordering for the Gaussian free field, and we show that is can be controlled, uniformly with respect to $\delta$, using an approximate Wick's identity for the height function which may be of independent interest, see Lemma~\ref{lemma:Wick_rule_dimer_height}. 

It is natural to ask whether Theorems~\ref{thma:main1} and~\ref{thma:main2} extend to nice enough approximations $\Omega_\delta$ to arbitrary (say, simply connected) domains $\Omega$ with boundary instead of $\CC^+$. As mentioned above, the dimer model is notoriously sensitive to the structure of the boundary of the approximating domains $\Omega_\delta$. In particular, the conformal invariance results, such as convergence of the dimer height function to the Gaussian free field, do not hold in general in their literal form; instead, the domain may be divided into \emph{liquid}, \emph{gaseous} and \emph{frozen} regions, and convergence of the height function to the GFF is expected only in liquid region where the GFF is sampled with respect to a conformal structure whose definition might involve a non-trivial change of metric~\cite{kenyon2006dimers,kenyon2007limit}. This conjecture, supported by several partial results, remains widely open in general.

Thus, it is reasonable to restrict the class of approximations even further, and consider Temperleyan approximations. In our proof of Theorem 1, the ingredient that is only available in the upper half-plane is the construction and estimation of the inverse Kasteleyn matrix with monodromy, specifically Corollary \ref{cor:nea-diag_estimate_on_Krhoinv}. In fact, rather than carrying out this construction, it is easier to extend Theorem \ref{thma:main1} to the case of Temperleyan domains by as yet unpublished coupling arguments of the first author and B.~Laslier \cite{basoklaslier}. For Theorem \ref{thma:main2}, we stipulate that the only missing input is the extension of Theorem \ref{thmas:on_convergence_of_cyl_prob} to the case of Temperleyan domains; where again the main issue is the construction and estimation of the inverse Kasteleyn matrices with SL${}_2$($\RR$) monodromies. We believe that such extension is possible, but some technical details must be filled.

With that in mind, and in particular to stress that Theorem \ref{thma:main2} is in fact proved for Temperleyan domains conditionally on the extension of Theorem~\ref{thmas:on_convergence_of_cyl_prob}, for a large part of the paper we work in a domain $\Omega_\delta$ which is allowed to be either a bounded Temperleyan approximation to a bounded simply connected domain $\Omega$, with a horizontal or vertical boundary segment, or the upper half-plane $\CC^+_\delta$.

\subsection*{Organization of the paper} We begin with the technical Section~\ref{sec:Height_function_loop_statistics_and_monodromy} where we recall the classical discrete complex analysis framework and use it to obtain several estimates on height correlations and loop statistics for the double-dimer model. The section is concluded with Section~\ref{subsec:combinatorial_corresp_betwen_nesting_and_height} where we establish a combinatorial relation between two-point correlations of the nesting fields and certain observables expressed in terms of height function and macroscopic loops.

We continue with Section~\ref{sec:proof_of_main1} where we prove Theorem~\ref{thma:main1}. It begins with two subsections where we study the asymptotic of the inverse Kasteleyn operator with scalar monodromy around a puncture and two punctures, respectively. Theorem~\ref{thma:main1} is proven in Section~\ref{subsec:Proof_main1}. 

Finally, Section~\ref{sec:CLE_nesting} is devoted to the proof of Theorem~\ref{thma:main2}. We begin, in Section~\ref{subsec:def_of_mL_loc}, by preparatory lemmas concerning relevant spaces of random fields. In Section \ref{subsec:nesting_fields_for_CLE(4)} we review the construction of $\CLE_\kappa$ nesting fields developed by Miller, Watson, and Wilson. Our main technical objective there is to prove that in their construction, one can replace the number of loops encircling a ball of small radius about a point with the number of loops surrounding two points close to each other. After some more preparatory lemmas in Sections \ref{sec:FieldConv} and \ref{subsec:double-dimer_nesting_fields}, we complete a proof of Theorem~\ref{thma:main2} in Section~\ref{subsec:proof_main2}.

\subsection*{Acknowledgments} The work of both authors was supported by Academy of Finland through Academy projects Critical phenomena in dimension two: analytic and probabilistic methods (333932) and Lattice models and conformal field theory (363549). Part of this research was performed while M.B. was visiting the Institute for Pure and Applied Mathematics (IPAM), which is supported by the National Science Foundation (Grant No. DMS-1925919). We are grateful to Gaultier Lambert for pointing out the work \cite{pinson2004rotational} to us.
\section{Height function, loop statistics and monodromy}
\label{sec:Height_function_loop_statistics_and_monodromy}

We start this section by recalling the definition of Temperleyan domains. Recall that a domain $\Omega_\delta\subset \delta \ZZ^2$ is called \emph{Temperleyan} if it is obtained by the following procedure. First, consider a finite simply connected polygon $\gamma \subset 2\delta\ZZ^2$. Let $\widetilde{\Omega}_\delta$ consist of vertices of $\delta\ZZ^2$ that are inside or on $\gamma$; these are comprised of black vertices that are vertices of $2\delta\ZZ^2$ or its dual $(2\delta\ZZ+\delta)^2$, and white vertices that correspond to edges of $2\delta\ZZ^2$. The Euler formula implies that the number of black vertices in $\widetilde{\Omega}_\delta$ is bigger than the number of white vertices by one. We therefore remove from $\widetilde{\Omega}_\delta$ one black vertex lying on $\gamma$ to obtain a Temperleyan domain $\Omega_\delta$. We call the removed black vertex a ``root'' of the domain.

Let $\Omega\subset \CC$ be a proper simply-connected domain with a distinguished straight line horizontal segment on the boundary. We say that a sequence of Tempreley domains $\Omega_\delta$ approximates $\Omega$ nicely, if $\Omega_\delta$ converge to $\Omega$ in Carath\'eodory topology, and in addition there is a sequence of straight boundary segments of $\widetilde{\Omega}_\delta$ converging to the boundary segment of $\Omega$, and the roots of $\Omega_\delta$ converge to the middle point of the segment.

Throughout this section, $\Omega$ and $\Omega_\delta$ are either a bounded simply connected domain and its nice approximation by Temperleyan domains, or $\Omega=\CC$ and $\Omega_\delta=\CC^+_\delta$. In the former case, the primary object of study is the \emph{double dimer model}, which consists of two dimer covers (perferct matchings) of $\Omega_\delta$ chosen uniformly at random and independently of each other. When superimposed upon each other, two dimer covers form a collection of simple lattice loops, which is the primary object of our interest here.
In the the case $\Omega_\delta = \CC^+_\delta$ we define the corresponding probability measure by taking the limit over a sequence of finite Temperley domains exhausting $\CC^+_\delta$ with roots tending to infinity; as usual, this measure is considered with respect to the sigma-algebra generated by states of single edges. The dimer model constructed in this way inherit nice properties of the double-dimer models on finite Temperley domains, in a sense that local statistics can be still written via determinantal identities involving the unique inverse Kasteleyn operator vanishing at infinity, see~\cite[Section~5.4]{DubedatDoubleDimers}. The same property is classically known for the dimer Gibbs measure in the full-plane lattice $\ZZ^2$ with maximal entropy~\cite{de2007quadri},~\cite{kenyon1997local}; note that this measure can also be obtained by taking a limit along a sequence of finite Temperley domains. Another important property of the dimer model in $\CC^+_\delta$ is that there is almost surely no infinite path in the double-dimer configuration (see~\cite[Lemma~26]{DubedatDoubleDimers}).

\subsection{Combinatorial relations between nesting and height function}
\label{subsec:combinatorial_corresp_betwen_nesting_and_height}

Let us recall the definition of the height function of the dimer and the double dimer models; see~\cite{kenyon2009lectures} for more details. Given a collection $\Llor$ of oriented simple disjoint lattice loops in $\Omega_\delta$, we define $h_\delta^{\Llor}$ to be the function which is constant on faces of $\delta\ZZ^2$, zero outside $\Omega_\delta$, and whose values increase by $+1$ if one crosses a loop from left to right. Given two dimer configurations $D_1,D_2$ in $\Omega_\delta$, we superimpose them to get a collection of loops and double edges. We orient each loop so that all dimers in $D_1$ are oriented from black to white and denote by $h_\delta^{(D_1,D_2)}$ the corresponding height function. (If $\Omega_\delta = \CC^+_\delta$, then we assume additionally that there is no infinite path in the superposition.) If $(D_1,D_2)$ is sampled according to the law of the double-dimer model, then we get the double-dimer height function which we denote simply by $h_\delta$. If $D_1$ is sampled according to the law of the dimer model, while $D_2$ is taken arbitrary and fixed, then $h_{\delta,1} = h_\delta^{(D_1,D_2)} - \EE h_\delta^{(D_1,D_2)}$ is the centered dimer height function. Since changing $D_2$ amounts to shifting $h_\delta^{(D_1,D_2)}$ by a deterministic function, $h_{\delta,1}$ does not depend on $D_2$, and $h_\delta = h_{\delta,1} - h_{\delta,2}$, where $h_{\delta,2}$ has the same distribution as $h_{\delta,1}$ and is independent of it.

Let $\Ll$ denote the set of \emph{unoriented} loops in the random double dimer configuration $(D_1,D_2)$. Conditionally on $\Ll$,
orientations of loops in $\Llor$ are independent and each orientation is equally probable. It follows that the value of the double-dimer height function $h_\delta$ at a face $x$ of $\Omega_\delta$ can be computed as
\begin{equation}
    \label{eq:h_delta_sum_bernoulli}
    h_\delta(x) = \sum_{\Ll\ni \gamma\text{ surrounds }x} s_\gamma,
\end{equation}
where, conditionally on $\Ll$, $\{s_\gamma\}_{\gamma\in \Ll}$ are independent $\frac12$-Bernoulli variables taking values $\pm 1$.

Given faces $x_1,\dots, x_n\in\Omega$, denote by $\NOd(x_1,\dots, x_n)$ the number of loops in $\Ll$ surrounding these points. The relation \eqref{eq:h_delta_sum_bernoulli}, in particular, implies that for all $t\in \mathbb{R},$
\begin{equation}
    \label{eq:char_fct_of_h_delta_and_Laplce_transform_of_N}
    \EE\left[ \exp(i t h_\delta(x)) \right] = \EE \left[\left(\cos t\right)^{\NOd(x)}\right],
\end{equation}
and that we have the following relations between $h_\delta$ and $N$ for any faces $x,y$ of $\Omega$ (not necessarily distinct):
\begin{equation}
\label{eq:Nxy_from_h}
\EE \NOd(x,y)=\EE [h_\delta(x)h_\delta(y)].
\end{equation}
More generally, the relation between loops statistics and moments of $h_\delta$ han be summarised by the following lemma:
\begin{lemma}
    Given faces $x_1,\dots,x_n$ of $\Omega_\delta$ and $I\subset \{x_1,\dots, x_n\}$, denote by $A_I$ the number of loops in $\Ll$ which encircle every point $x_i$ for $i\in I$, but none of the points $x_i$ for $i\notin I$. Then we have
\begin{equation}
        \label{eq:nvh1}
        \EE \prod_{i = 1}^n h_
        \delta(x_i)^{d_i} = (-i)^{\sum_{i = 1}^n d_i}\left.\frac{\partial^{d_1 + \ldots + d_n}}{\partial t_1^{d_1}\dots \partial t_n^{d_n}}\right|_{t_i=0} \EE \prod_{I\subset \{x_1,\dots, x_n\}} \cos\left(\sum_{i \in I} t_i\right)^{A_I}.
    \end{equation}
\end{lemma}
\begin{proof}
We can write, using \eqref{eq:h_delta_sum_bernoulli} and computing the expectation conditionally on $\Ll$ first,
\begin{multline}
        \label{eq:nvh2}
        \EE \left(\left.\exp\left(i\sum_{i = 1}^n t_ih(x_i)\right)\right|\mathcal{L}\right) = \EE \exp\left(i\sum_{i = 1}^n t_i\sum_{\gamma\text{ surrounds }x_i}s_\gamma\right)\\=\EE \left.\exp\left(i\sum_{\gamma\in \mathcal{L}} s_\gamma \sum_{i:\gamma\text{ surrounds }x_i}t_i\right)\right)=\EE\prod_{\gamma\in\mathcal{L}}\cos\left(\sum_{i:\gamma\text{ surrounds }x_i}t_i\right)=\prod_{I\subset \{x_1,\dots, x_n\}} \cos\left(\sum_{i \in I} t_i\right)^{A_I},
    \end{multline}
    and the result follows by taking expectations and differentiating.
\end{proof}

We will need a few particular cases, one of which is $n=1,d_1=4$:
\begin{equation}
\label{eq:N2viah4}
\EE h^4_\delta(x)=\frac{d^4}{dt^4}\EE (\cos t)^{\NOd(x)}=3\EE N^2_\delta(x)-2\EE \NOd(x).
\end{equation}

For others, we introduce the following fields:

\begin{enumerate}
    \item Define the~\emph{nesting field} and its \emph{two-point approximation} by
    \[
        \vphi_\delta(x) = \NOd(x) - \EE \NOd(x),\qquad \vphi^\eps_\delta(x) = \NOd(x,x + \eps) - \EE \NOd(x,x+\eps),\quad \eps>0.
    \]

    \item Define, for $h=h_\delta$, the \emph{normal-ordered square of $h$} and its two-point approximation by
    \[
        \psi_\delta(x) = h_\delta(x)^2 - \EE h_\delta(x)^2,\qquad \psi^\eps_\delta(x) = h_\delta(x)h_\delta(x+\eps) - \EE h_\delta(x)h_\delta(x+\eps),\quad \eps>0.
    \]

\end{enumerate}

We will need the following relation between moments of these fields:

\begin{cor}
    \label{lemma:nesting_via_height}
    Define
    \begin{align*}
        P_\delta(x,y) =& \,2\NOd(x,y)^2-2\NOd(x,y),\\
        P_{\delta,\eps}(x,y) =& \,2\NOd(x,y)\NOd(x,y+\eps) - 2\NOd(x,y,y+\eps)\\
        P_{\delta,\eps,\eps}(x,y) =& \,\NOd(x,y)\NOd(x+\eps,y+\eps) + \NOd(x,y+\eps)\NOd(x+\eps,y) -2\NOd(x,x+\eps,y,y+\eps).
    \end{align*}
    Then we have
    \begin{align*}
        &\EE\psi_\delta(x)\psi_\delta(y) = \EE \vphi_\delta(x)\vphi_\delta(y) + \EE P_\delta(x,y),\\
        &\EE\psi_\delta(x)\psi_\delta^\eps(y) = \EE\vphi_\delta(x)\vphi_\delta^\eps(y) + \EE P_{\delta,\eps}(x,y),\\
        &\EE\psi_\delta^\eps(x)\psi_\delta^\eps(y) = \EE\vphi_\delta^\eps(x)\vphi_\delta^\eps(y) + \EE P_{\delta,\eps,\eps}(x,y)
    \end{align*}
\end{cor}

\begin{proof}

The corollary follows by applying \eqref{eq:nvh1} on the left-hand side and expressing each $A_I$ in terms of $\NOd(\cdot)$. We start with finding an expression for $\EE [h_\delta(x_1)h_\delta(x_2)h_\delta(x_3)h_\delta(x_4)]$ for some arbitrary $x_1,x_2,x_3,x_4$. By \eqref{eq:nvh1}, we have
   \begin{equation*}
   \EE [h_\delta(x_1)h_\delta(x_2)h_\delta(x_3)h_\delta(x_4)]
   = \left.\frac{\partial^4}{\partial t_1\partial t_2 \partial t_3 \partial t_4}\right|_{t=0} \EE \prod_{I\subset \{1,2,3,4\}} \cos\left(\sum_{i \in I} t_i\right)^{A_I}
   \end{equation*}
We need to sum the contributions from $\frac{\partial}{\partial t_i}$ acting on various terms of the form $\cos\left(\sum_{i \in I} t_i\right)^{A_I}$, in the product, or on terms of the form $\sin\left(\sum_{i \in I} t_i\right)$ arising from a previous differentiation. The only non-zero contributions are those for which all the sines produced are subsequently acted upon. Pairing $i$ with $j$ if $\frac{\partial}{\partial t_i}$ created a sine term which $\frac{\partial}{\partial t_j}$ has acted upon, and taking into account that $N_\delta(\{x_i\}_{i\in I})=\sum_{J\supset I}A_J$, we get
    \begin{multline*}
   \left.\frac{\partial^4}{\partial t_1\partial t_2 \partial t_3 \partial t_4}\right|_{t=0}\prod_{I\subset \{1,\dots, n\}} \cos\left(\sum_{i \in I} t_i\right)^{A_I}\\=\sum_{\substack{\{p,\hat{p}\}\text{ pairing}\\\text{of } 1,2,3,4}}\sum_{p\subset I\neq I'\supset \hat{p}}A_IA_{I'}+A_{\{1,2,3,4\}}(3A_{\{1,2,3,4\}}-2)\\=\sum_{\{p,\hat{p}\}\text{ pairing}}N_\delta(\{x_i\}_{i\in p})N_\delta(\{x_i\}_{i\in\hat{p}})-2N_\delta(x_1,x_2,x_3,x_4)
   \end{multline*}
   Taking the expectation we finally get
   \begin{multline}
       \label{eq:nvh10}
       \EE [h_\delta(x_1)h_\delta(x_2)h_\delta(x_3)h_\delta(x_4)] =\\
       = \EE\Big[ N_\delta(x_1,x_2)N_\delta(x_3,x_4) + N_\delta(x_1,x_3)N_\delta(x_2,x_4) + N_\delta(x_1,x_4)N_\delta(x_2,x_3) - 2N_\delta(x_1,x_2,x_3,x_4) \Big].
   \end{multline}
   We are now in the position to derive the relations required in the corollary. Expanding by definition we get
   \begin{multline*}
       \EE\psi_\delta^\eps(x)\psi_\delta^\eps(y) - \EE\vphi_\delta^\eps(x)\vphi_\delta^\eps(y) =\\
       = \EE [h_\delta(x)h_\delta(x+\eps)h_\delta(y)h_\delta(y+\eps)] - \EE[h_\delta(x)h_\delta(x+\eps)]\EE[h_\delta(y)h_\delta(y+\eps)] -\\
       - \EE [N_\delta(x,x+\eps)N_\delta(y,y+\eps)] + \EE N_\delta(x,x+\eps)\EE N_\delta(y,y+\eps).
   \end{multline*}
   Using~\eqref{eq:nvh1} with $x_1 =x,x_2=x+\eps,x_3=y,x_4=y+\eps$ to evaluate the first term and~\eqref{eq:Nxy_from_h} to evaluate the second term on the right-hand side of the expression above we get
   \[
    \EE\psi_\delta^\eps(x)\psi_\delta^\eps(y) - \EE\vphi_\delta^\eps(x)\vphi_\delta^\eps(y) = \EE P_{\delta,\eps,\eps}(x,y)
   \]
   as required. The other two expressions can be derived similarly.

\end{proof}
\begin{rem}\label{rem:P_well_defined_for_CLE}
Note that $P_\delta$, $P_{\delta,\eps}$, $P_{\delta,\eps,\eps}$ are second degree polynomials in $\NOd(\cdot)$, where $\cdot$ stands for sub-collections of $\{x,x+\eps,y,y+\eps\}$ containing at least one of $x,x+\eps$ and at least one of $y,y+\eps$. In particular, their analogs are well defined for infinite loops ensembles like CLE${}_4$, and when computed for double dimers, they remain stochastically bounded as the mesh size $\delta\to 0$ and/or $\eps\to 0$. Moreover, $P_\delta=P_{\delta,\eps}=P_{\delta,\eps,\eps}$ on the event that there is no large loop separating $x$ from $x+\eps$ or $y$ from $y+\eps$, which has high probability for small $\eps$.
\end{rem}

\subsection{Kasteleyn matrix, dimers and the Gaussian free field}
In this section we collect several basic facts and technical lemmas about the dimer height function and certain statistics of double-dimer loops. Along the way we will establish the asymptotic relations for $\mu_\delta$ and $\sigma_\delta$ stated in Theorem~\ref{thma:main1}.

We begin by introducing the (standard) discrete complex analysis setup we will stick to.
The Kasteleyn weights are fixed as follows. We assume that vertices of $\ZZ^2$ are bicolored, let $0\in \ZZ^2$ be black for definiteness. In what follows we will usually be using letters $b$ and $w$ to denote black and white vertices respectively. We set
\[
  K(w,w+e) = e,\qquad e\in \{ \pm1, \pm i \}
\]
where $w$ is an arbitrary white vertex.
For $u\in \ZZ^2$ set
\begin{equation}
    \label{eq:def_of_eta}
    \eta_u = \exp\left[ -\frac{\pi i}{2}(\Im u \mod 2) \right] = \begin{cases}
        1,\quad \Im u\in 2\ZZ,\\
        -i,\quad \Im u \in 2\ZZ +1.
    \end{cases}
\end{equation}
Note that $K(w,b) = \pm\bar{\eta}_w\bar{\eta}_b$ (we keep the conjugation intentionally to keep up with~\cite{CLR1}, where $\eta$ stands for the origami square root function). We recall that there is a unique full-plane inverse $K^{-1}$ satisfying
\begin{equation}
  \label{eq:K-1_asymp}
  K^{-1}(b,w) = \Pr\left[ \frac{1}{\pi(b-w)},\eta_b\eta_w\RR \right] + O\left( \frac{1}{|b-w|^2} \right)=\frac{1}{2\pi(b-w)}+\frac{\eta^2_w\eta^2_b}{2\pi(\bar{b}-\bar{w})}+ O\left( \frac{1}{|b-w|^2} \right).
\end{equation}
see~\cite{KenyonCriticalPlanarGraphs}.

Given $\delta>0$ we set
\[
    K_\delta(\delta w,\delta b) = \delta K(w,b),\qquad K_\delta^{-1}(\delta b,\delta w) = \delta^{-1} K^{-1}(b,w)
\]
When $u\in \delta\ZZ^2$, we write $\eta_u$ in place of $\eta_{\delta^{-1}u}$, abusing the notation slightly. Given $\Omega_\delta\subset \delta \ZZ^2$ we denote by $K_{\Omega_\delta}$ the restriction of $K_\delta$ to the vertices of $\Omega_\delta$. In the case when $\Omega_\delta = \CC^+_\delta = \delta\ZZ^2\cap \CC^+$ we define
\begin{equation}
    \label{eq:def_of_half-plane_K-1}
    K_{\CC^+_\delta}^{-1}(b,w) = K_\delta^{-1}(b,w) - \eta_b^2K_\delta^{-1}(\bar{b},w).
\end{equation}
It is easy to see that $K_{\CC^+_\delta}^{-1}(b,w)$ is both left and right inverse to $K_{\CC^+_\delta}$ as matrices. The following lemma is a classical result obtained by Kenyon~\cite{KenyonConfInvOfDominoTilings,KenyonGFF}.

\begin{lemma}
    \label{lemma:asymp_of_Kinv}
    Assume that $\Omega,\Omega_\delta$ are as in the beginning of Section~\ref{sec:Height_function_loop_statistics_and_monodromy}. Let $I$ denote the open horizontal interval on the boundary of $\Omega$, and let us assume that $I\subset\RR$. There exist functions
    $F^{[++]}, F^{[--]}:(\Omega\cup I)\times(\Omega\cup I)\smm\mathrm{diag}\to \CC$ and $F^{[+-]}, F^{-+}:(\Omega\cup I)\times(\Omega\cup I)\smm\mathrm{diag}(I\times I)\to \CC$ such that the following holds:
    \begin{enumerate}
        \item We have $F^{[++]} = -\overline{F^{[--]}}$ and $F^{[+-]} = -\overline{F^{-+}}$. The function $F^{[++]}$ is holomorphic in both variables, and $F^{[+-]}$ is holomorphic in the first variable and anti-holomorphic in the second one, and these statements hold up to $I\times I\smm\mathrm{diag}$.
        \item We have $F^{[++]}(x,y) = \frac{2}{\pi (x-y)} + \mathrm{reg}$, where $\mathrm{reg}$ is a function continuous on $(\Omega\cup I)\times (\Omega\cup I)$.
        \item We have
        \begin{multline}
        \label{eq:asymp_of_Kinv}
            K_{\Omega_\delta}^{-1}(b,w) = \frac{1}{4}\left( F^{[++]}(b, w) - \eta_w^2F^{[+-]}(b, w) + \eta_b^2F^{[-+]}(b, w) - (\eta_b\eta_w)^2F^{[--]}(b, w) \right) \\
            + o(1) + O\left(\frac{\delta}{(b-w)^2}\right) + O\left(\frac{\delta}{(b-\bar{w})^2}\right)
        \end{multline}
       as $\delta\to 0$ uniformly in $(b,w)$ from any compact in $(\Omega\cup I)\times(\Omega\cup I)$.
        \item Let $h$ be the GFF in $\Omega$ with Dirichlet boundary conditions and $x_1,\dots, x_n\in \Omega$ be some distinct points. Let $l_1,\dots, l_n$ be disjoint simple paths connecting these points with $I$ inside $\Omega$ and oriented towards $I$. Then we have
        \begin{multline}
        \label{eq:GFF_determ}
            \pi^{-\frac{n}{2}}
            \EE(h(x_1)\cdot\ldots\cdot h(x_n)) \\= (4i)^{-n}\sum_{s\in \{+,-\}^n}\int_{l_1}\dots\int_{l_n}\det\left[ F^{[s_is_j]}(z_i,z_j)\indic[i\neq j] \right]\,dz_1^{[s_1]}\dots dz_n^{[s_n]}
        \end{multline}
        where $dz^{[+]} = dz$ and $dz^{[-]} = d\bar{z}$.
    \end{enumerate}
\end{lemma}
\begin{proof}
    If $\Omega$ is bounded, then the first three claims is a summary of the analysis performed in~\cite{KenyonConfInvOfDominoTilings}. Specifically, our mesh size $\delta$ and inverse Kasteleyn operator $K_{\Omega_\delta}^{-1}(b,w)$ are denoted in ~\cite{KenyonConfInvOfDominoTilings} by $\epsilon$, $\frac{1}{\epsilon}C(w,b)$, while the functions $F^{[++]}(b,w)$, $F^{[+-]}(b,w)$, $F^{[-+]}(b,w)$ and $F^{[--]}(b,w)$ correspond to the functions $F_{1,1}(w,b)$, $-F_{1,-1}(w,b)$, $\overline{F_{1,-1}(w,b)}$ and $-\overline{F_{1,1}(w,b)}$ respectively. The third claim, with uniformity over compacts in $\Omega\times\Omega$, follows from \cite[Theorem 13]{KenyonConfInvOfDominoTilings} after substituting the definitions of $F^{[\pm\pm]}$ and using the asymptotics \eqref{eq:K-1_asymp} to replace $C_0$ with its continuous counterpart. (The $\eta$ factors implement the convention described after \cite[Theorem 13]{KenyonConfInvOfDominoTilings}.) To extend the results to compacts in $(\Omega\cup I)\times (\Omega\cup I)$, we can follow the proof of~\cite[Theorem 14]{KenyonConfInvOfDominoTilings} where the asymptotics of $C(v_1,\cdot)$ is derived in the regime when $v_1$ is on the distance $O(\epsilon)$ from the straight boundary segment. The proof of this theorem consists of Schwartz reflecting $C(v_1,\cdot)$ and then repeating the arguments from the proof of~\cite[Theorem 13]{KenyonConfInvOfDominoTilings} in the doubled domain. To obtain the asymptotics that we need we observe that these arguments work verbatim if we take $v_1$ from any compact subset of $\Omega\cup I$. Note that after being extended to the doubled domain, the function $C(v_1,\cdot)$ acquires the second pole sitting in the reflected image of $v_1$; this explains the fact that the term $C_0(v_1,\cdot)$ from the asymptotics in~\cite[Theorem 13]{KenyonConfInvOfDominoTilings} is replaced with more precise term $C_H(v_1,\cdot)$ (which corresponds to $K_{\CC^+_\delta}^{-1}(b,w)$ in our notation).

    When $\Omega = \CC^+$, we have $F^{[++]}(x,y)=\frac{2}{\pi (x-y)}$ and $F^{[+-]}(x,y)=\frac{2}{\pi (x-\bar{y})}$, and the asymptotic relations between them and $K_{\CC^+_\delta}^{-1}$ follow directly from~\eqref{eq:K-1_asymp} and~\eqref{eq:def_of_half-plane_K-1}.

    The fourth claim is the content of~\cite{KenyonGFF}; by conformal invariance of both sides of the equation \cite[Proposition 15]{KenyonConfInvOfDominoTilings}, it does not matter whether one is working in a bounded domain or in $\CC^+_\delta$.
\end{proof}

We will need the following corollary:

\begin{lemma}
\label{lem:vanishings}
Let $S^{(n)}_{\geq 3}$ denote the set of permutations of $\{1,\dots,n\}$ that don't have any fixed points or $2$-cycles. Then,
\begin{equation}
\label{eq:vanishing}
\sum_{\sigma\in S^{(n)}_{\geq 3}}\sum_{s\in \{+,-\}^n}\int_{l_1}\dots\int_{l_n}(-1)^{\mathrm{sign}(\sigma)}\prod_{i=1}^n F^{[s_is_{\sigma(i)}]}(z_i,z_{\sigma(i)}) \,dz_1^{[s_1]}\dots dz_n^{[s_n]}=0.
\end{equation}
\end{lemma}
\begin{rem}
\label{rem:vanishing}
Since \eqref{eq:vanishing} holds for any end-points $x_1,\dots,x_n$ of $l_1,\dots,l_n$, it is in fact the case that the sum of integrands is identically zero.
\end{rem}
\begin{proof}
We first observe that the statement is true if we sum instead over permutations that have no fixed points and that do not correspong to pairings. To wit, expand the determinant in the right-hand side of \eqref{eq:GFF_determ} into a sum over permutations, and observe that permutations which have fixed points do not contribute because of the $\indic[i\neq j]$ factor, and the permutations corresponding to pairings combine together to $\sum_{p\text{ pairing }}\prod_{\{i,j\}\in p}\frac{4}{\sqrt{\pi}}\EE [h(x_{i})h(x_{j})]$, which is equal to the left-hand side of \eqref{eq:GFF_determ} by Wick's formula. Now we can prove the lemma by induction: the sum over permutations without fixed points, not corresponding to a pairing, and with a given non-empty set of two-cycles, vanishes by induction hypothesis. Hence the sum over $S_{\geq 3}^{(n)}$ also vanishes.
\end{proof}

\subsection{Approximate Wick's rule and the asymptotic of \texorpdfstring{$\mu_\delta$}{mdelta} and \texorpdfstring{$\sigma_\delta$}{sigmadelta}}
\label{subsec:asymp_of_mu_delta_and_sigma_delta}

In~\cite{KenyonConfInvOfDominoTilings, KenyonGFF} Kenyon has proved that $h_{\delta,1}$ converges to the Gaussian free field in $\Omega$ with Dirichlet boundary conditions, rescaled by $\frac{1}{\sqrt{\pi}}$. The proof is based on an expression of multi-point correlations of $h_{\delta,1}$ in terms of the inverse Kasteleyn matrix $K_{\Omega_\delta}^{-1}$. In the next lemma we aim to specify how this expression results in an asymptotic Wick's rule for $h_{\delta,1}$.

Given an oriented dual lattice path $l$ and a dual edge $(bw)^\ast\in l$ oriented alongside $l$, define
\begin{equation}
  \label{eq:def_of_dedge}
  d(bw)^\ast = \begin{cases}
      K_\delta(w,b),\quad b\text{ is on the left of }l,\\
      -K_\delta(w,b),\quad \text{else}.
  \end{cases}
\end{equation}
Note that
\begin{equation}
    \label{eq:dedge=int}
    d(bw)^\ast = i\int_{(bw)^\ast} dz,\qquad (\eta_b\eta_w)^2\,d(bw)^\ast = -i\int_{(bw)^\ast} d\bar{z}.
\end{equation}

The following lemma is an approximate Wick's rule for the height function:

\begin{lemma}
    \label{lemma:Wick_rule_dimer_height}
    For each $n\geq 3$ there exists a function $F_\delta$ defined on $n$-tuples of faces of $\Omega_\delta$ such that the following holds. Let $x_1,\dots, x_n$ be some (not necessary distinct) faces of $\Omega_\delta$ identified with their centers, let $r = \min\{\diam \{x_i\}_{i\in I}\ \mid\ I\subset \{1,\dots, n\},\ |I| = 3\}$. Then we have
    \begin{equation}
   \label{eq:Wick_rule_dimer_height}
        \EE[h_{\delta,1}(x_1)\cdot\dots\cdot h_{\delta,1}(x_n)] =\sum_{p \text{ - pairing}}\prod_{\{i,j\}\in p}\EE \left[h_{\delta,1}(x_{i}) h_{\delta,1}(x_{j})\right]
        + F_\delta(x_1,\dots, x_n)
    \end{equation}
    and for any compact $K\subset \Omega\cup I$
    there exists a constant $C>0$ such that whenever $x_1,\dots, x_n\in K$ we have
    \begin{equation}
    \label{eq:Wick_rule_dimer_height_error}
        |F_\delta(x_1,\dots, x_n)| \leq C|\log \delta|^{\lfloor \frac{n-3}{2} \rfloor} \cdot \left(\min\left( \frac{\delta|\log\delta|^{n-1}}{r}, 1 \right)+ o(1)\right),
    \end{equation}
    where $o(1)$ refers to the same error term as in Lemma~\ref{lemma:asymp_of_Kinv}; in particular, it is uniform given $n,K$.

    The same statement is true if $h_{\delta,1}$ is replaced by $h_\delta$.
\end{lemma}
We remark that the lemma is true both for even and odd $n$; in the latter case the sum over pairings is empty and we are left just with the error term. We also remark that for the last factor, we will only use that it tends to zero uniformly as $\delta\to 0$ and $r\gg \delta\log \delta$ (say, $r$ fixed).
We will need the following lemma:

\begin{lemma}
\label{lem:estimate_integrals}
For every $n>0$, there exists a constant $C>0$ such that for any $0<\delta<1/2$, we have the the following bounds:
\begin{eqnarray}
    \int_0^1\dots \int_0^1 \frac{1}{x_1+x_2 + \delta} \cdot\ldots\cdot \frac{1}{x_{n-1} + x_{n} + \delta}\,dx_1\dots dx_{n}\leq &C; \label{eq:int_bound_const}\\
        \int_0^1\dots \int_0^1 \frac{1}{x_1+x_2 + \delta} \cdot\ldots\cdot \frac{1}{x_{n} + x_{1} + \delta}\,dx_1\dots dx_{n}\leq &C |\log \delta|;\label{eq:int_bound_log}\\
        \int_0^1\dots \int_0^1 \frac{1}{(x_1+x_2 + \delta)^2} \frac{1}{x_2+x_3 + \delta}\cdot\ldots\cdot \frac{1}{x_{n-1} + x_{n} + \delta}\,dx_1\dots dx_{n}\leq &C |\log\delta|^{n-1};\label{eq:int_bound_log_bis}\\
        \int_0^1\dots \int_0^1 \frac{1}{(x_1+x_2 + \delta)^2} \frac{1}{x_2+x_3 + \delta}\cdot\ldots\cdot \frac{1}{x_{n} + x_{1} + \delta}\,dx_1\dots dx_{n}\leq &C \delta^{-1}.\label{eq:int_bound_delta_minus1}
\end{eqnarray}
\end{lemma}
\begin{proof}
Note that we can actually integrate to $\frac14$ instead of $1$ as the difference will be $O(1)$. We break each of the integrals into $n!$ parts corresponding to orderings of $x_1,\dots,x_n$, and in each term, drop the smaller variable for the upper bound; we also shift all variables by $\delta$. This yields a bound of the form
\[
\dots\leq \sum_{\sigma}\int_{\delta<x_{\sigma(1)}<\dots<x_{\sigma(n)}<\frac34}\frac{dx_{\sigma(1)}\dots dx_{\sigma(n)}}{x^{d_1}_{\sigma(1)}\dots x^{d_n}_{\sigma(n)}}+O(1),
\]
where the degrees $d_i=d_i(\sigma)$ are the incoming degrees of vertices in the oriented multigraph with vertex set $\{1,\dots,n\}$ and one edge connecting $i,j$ for each term of the form $(x_{\sigma(i)}+x_{\sigma(j)}+\delta)^{-1}$ in the integrand, oriented from $\min(i,j)$ to $\max(i,j)$. Ignoring orientations, this multigraph is a simple $n$-loop, a simple $n$-path, a simple $n$-path with one double edge, and a simple $n$-loop with one double edge, in the cases \eqref{eq:int_bound_const}, \eqref{eq:int_bound_log}, \eqref{eq:int_bound_log_bis}, \eqref{eq:int_bound_delta_minus1} respectively.

We now estimate the integrals by integrating over $x_{\sigma(n)},\dots,x_{\sigma(1)},$ in this order, taking into account that for $0<x<\frac34$ and $m\geq 1$ we have
\begin{equation}
\label{eq:int_log_estimate}
\int_{x}^\frac12 \frac {|\log y|^k\,dy}{y^m}\leq c(k,m)\cdot \frac{|\log x|^{k+\indic_{m=1}}}{x^{m-1}}.
\end{equation}
Denote $p_j:=d_j+\dots+d_n-(n-j).$ Each excursion of the $n$-path or the $n$-loop into the set $\{j,\dots,n\}$ that visits $k$ vertices contributes at least $k$ to $d_j+\dots+d_n$, except in the case \eqref{eq:int_bound_const} and $j=1$. Hence, with that exception, we have $p_j\geq 1,$  and inductive application of \eqref{eq:int_log_estimate} gives
\begin{multline}
\int_{\delta<x_{\sigma(1)}<\dots<x_{\sigma(n)}<\frac34}\frac{dx_{\sigma(1)}\dots dx_{\sigma(n)}}{x^{d_1}_{\sigma(1)}\dots x^{d_n}_{\sigma(n)}}\\ \leq \dots\leq C_j\cdot \int_{\delta<x_{\sigma(1)}<\dots<x_{\sigma(j)}<\frac34}\frac{|\log x_{\sigma(j)}|^{q_j}dx_{\sigma(1)}\dots dx_{\sigma(j)}}{x^{d_1}_{\sigma(1)}\dots x^{d_{j-1}}_{\sigma(j-1)}x^{p_j}_{\sigma(j)}}
\leq C_0|\log \delta|^{q_0}\delta^{n-d_1-\dots-d_n}.
\end{multline}
where $q_i=|\{k>i:p_k=1\}|.$ To prove \eqref{eq:int_bound_const}, observe that the same induction is valid up to $j=1$, and we have $p_1=0$, hence the final integral has an integrable singularity. In the cases  \eqref{eq:int_bound_log}, \eqref{eq:int_bound_log_bis}, \eqref{eq:int_bound_delta_minus1}, the exponent $n-d_1-\dots-d_n$ is equal to $0,0,$ and $-1$ respectively. Observe furthermore that for $j>1$, in the case of loops, the above bound $p_j\geq 1$ is never attained, since in fact each excursion visiting $k$ vertices of $\{j,\dots,n\}$ will contribute at least $k+1$ to $d_j+\dots+d_n$. This means that in the cases \eqref{eq:int_bound_log},  \eqref{eq:int_bound_delta_minus1}, we have $q_i=0$ for $i\geq 1$. In the former case, we have $p_1=1$, hence $q_0=1$, and in the latter case we have $p_1=2$, hence $q_0=0$. Finally, for \eqref{eq:int_bound_log_bis}, we improve the trivial bound $q_0\leq n$ to $q_0\leq n-1$ by noticing that if $j>0$ and the double edge is incident to one of the vertices of $\{j,\dots,n\}$, then once again the bound $p_j\geq 1$ improves to $p_j>1$.
\end{proof}

\begin{proof}[Proof of Lemma \ref{lemma:Wick_rule_dimer_height}]
    We first assume that there can be chosen $n$ paths $l_1,\dots, l_n$ on the dual graph $\Omega_\delta^\ast$ such that each $l_i$ connects $x_i$ with the straight segment on the boundary of $\Omega_\delta$, and no two paths are incident to the same vertex of $\Omega_\delta$. Given a compact $K$ as in the statement of the lemma containing $x_1,\dots, x_n\in K$ one can further choose paths $l_1,\dots, l_n$ such that the distance from $l_i$ to the boundary of $\Omega_\delta$ with the straight segment removed is bounded from below uniformly in $\delta$, and such that for any $i\neq j$ and $x\in l_i, y\in l_j$ we have
    \begin{equation}
    \label{eq:distance_bound}
        |x-y| \gtrsim \length(l_i(x,x_i)) + |x_i-x_j| + \length(l_j(y,x_j)),
    \end{equation}
    with some constant independent on $\delta$, where $\length(l(a,b))$ is the Euclidean length of the path segment of $l$ between $a$ and $b$. We can moreover require that each $l_i$ is a union of segments of even length except maybe for one (in the units of $\delta$), parallel to the coordinate axis, with the number of segments bounded by a constant depending only on $n$ and $K$.

   The following formula was obtained in the course of the proof of  ~\cite[Proposition~20]{KenyonConfInvOfDominoTilings}:
    \begin{equation}
        \label{eq:Wdh1}
        \EE[h_{\delta,1}(x_1)\cdot\dots\cdot h_{\delta,1}(x_n)] = \sum_{(b_1w_1)^*\in l_1}\dots \sum_{(b_{n}w_{n})^*\in l_{n}} \det \left[K^{-1}_{\Omega_\delta}(b_i,w_j)\indic[i\neq j]\right]\,\prod_{i = 1}^{n} d(b_iw_i)^\ast.
    \end{equation}
Indeed, for $(b_i w_i)^\star\in l_i$, denote by $dh_{(b_i w_i)^\star}$ the change of the single-dimer height function along $(b_i w_i)^\star$, that is, $dh_{(b_i w_i)^\star}=\indic[(b_iw_i)\in D_1]
(-1)^{b_i \text{ is on the left of }  l_i}$. Then,
\[
h_{\delta,1}(x_i)=
\sum_{(b_i w_i)^\star\in l_i}dh_{(b_i w_i)^\star} -\EE dh_{(b_i w_i)^\star}=\sum_{(b_i w_i)^\star\in l_i}\left(\indic[(b_iw_i)\in D_1]-\EE \indic[(b_iw_i)\in D_1]\right)\frac{d(b_i w_i)^\star}{K_\delta(w_i,b_i)},\]
and \eqref{eq:Wdh1} follows from ~\cite[Lemma~21]{KenyonConfInvOfDominoTilings}.

    Given a permutation $\sigma$, denote by $\cyclic_{\sigma}=\{I^\sigma_1,\dots,I^\sigma_{c(\sigma)}\},$ where $I^\sigma_1\sqcup\dots\sqcup I^\sigma_{c(\sigma)}=\{1,\dots,n\}$ and $|I^\sigma_1|\leq|I^\sigma_2|\leq\dots\leq|I^\sigma_{c(\sigma)}|$, its cyclic decomposition. We start with ~\eqref{eq:Wdh1}, expand the determinant into a sum of permutations, and exchange the order of summation. The permutations which have a fixed point do not contribute because of the $\indic[i\neq j]$ factor. Therefore, we have
    \begin{multline*}
    \EE[h_{\delta,1}(x_1)\cdot\dots\cdot h_{\delta,1}(x_n)]=\sum_{\sigma: |I^\sigma_{1}|\geq 2}\sum_{(b_1w_1)^*\in l_1}\dots \sum_{(b_{n}w_{n})^*\in l_{n}}(-1)^{\mathrm{sign}(\sigma)}\prod_{j=1}^{c(\sigma)}\left[\prod_{i\in I^\sigma_{j}}K^{-1}_{\Omega_\delta}(b_i,w_{\sigma(i)})d(b_iw_i)^\ast\right]\\
    =\sum_{\sigma: |I^\sigma_{1}|\geq 2}(-1)^{\mathrm{sign}(\sigma)}\prod_{j=1}^{c(\sigma)}\left(\sum \prod_{i\in I^\sigma_{j}}K^{-1}_{\Omega_\delta}(b_i,w_{\sigma(i)})d(b_iw_i)^\ast\right),
    \end{multline*}
    where the inner sum is over the paths $l_{i_1},\dots,l_{i_{|I^\sigma_j|}}$ with $\{i_1,\dots, i_{|I^\sigma_j|}\}=I^\sigma_j$. Assume that $n$ is even and $\sigma$ corresponds to a pairing, i.e., $p = \{\{i_1,i_2\},\dots,\{i_{n-1},i_n\}\}$. Since $\mathrm{sign}(\sigma)=c(\sigma)=n/2$ and \[\det\left[K^{-1}_{\Omega_\delta}(b_l,w_k)\indic[l\neq k]\right]_{k,l\in\{i,\sigma(i)\}}=-K^{-1}_{\Omega_\delta}(b_i,w_{\sigma(i)})K^{-1}_{\Omega_\delta}(b_{\sigma(i)},w_{i}),\] the term corresponding to $\sigma$ in the above sum simplifies to $\prod_{\{i,j\}\in p}\EE \left[h_{\delta,1}(x_{i}) h_{\delta,1}(x_{j})\right].$  We infer that the error term $F_\delta(x_1,\dots,x_n)$ in \eqref{eq:Wick_rule_dimer_height} is equal to
    \begin{equation}
        \label{eq:Wdh2}
F_\delta(x_1,\dots,x_n)=   \sum_{\sigma: |I^\sigma_1|\geq 2,\,|I^\sigma_{c(\sigma)}|\geq 3}(-1)^{\mathrm{sign}(\sigma)}\prod_{j=1}^{c(\sigma)}\left(\sum \prod_{i\in I^\sigma_{j}}K^{-1}_{\Omega_\delta}(b_i,w_{\sigma(i)})d(b_iw_i)^\ast\right).
    \end{equation}

To give the reader an idea of the following proof, let us see what happens if we estimate each term by substituting the bound $ K^{-1}_{\Omega_\delta}(b,w) = O\left(\frac{1}{|b-w|}\right)$, coming from Lemma~\ref{lemma:asymp_of_Kinv}. Bounding the sums by the corresponding integrals and using \eqref{eq:distance_bound} and Lemma \ref{lem:estimate_integrals} (specifically, eq.~\eqref{eq:int_bound_log}), we get
\[
|F_\delta(x_1,\dots,x_n)|\leq C\sum_{\sigma: |I^\sigma_1|\geq 2,\,|I^\sigma_{c(\sigma)}|\geq 3}|\log r||\log \delta|^{c(\sigma)-1}\leq C|\log r||\log \delta|^{\lfloor\frac{n-3}{2}\rfloor}.
\] This is not quite as good as we need, in particular, for $n=4$ and $r\approx \delta$, this yields a bound of $O(|\log \delta|)$ instead of the $O(1)$ needed in the proof of Lemma \ref{lemma:asymp_of_mu_delta_sigma_delta} below. The improvement comes from cancellations captured in Lemma \ref{lem:vanishings}: the main term of the asymptotics of $F_\delta$ corresponds to its continuous counterpart, which vanishes because the continuous GFF satisfies the Wick rule exactly.

 Given a collection $\mathcal{I}=\{I_1,\dots,I_{r}\}$, $r=r(\mathcal{I})$, of disjoint two-element subsets of $\{1,\dots,n\}$, let $\mathfrak{S}_{\mathcal{I}}$ denote the set of permutations of $\{1,\dots,n\}$ with $I_1^{\sigma}=I_1,\dots,I_r^{\sigma}=I_r$, and $|I_{r+1}^\sigma|\geq 3$. Denote also $I_{\mathcal{I}}:=\{1,\dots,n\}\setminus \cup_{i=1}^r I_i$.
 We split the sum in \eqref{eq:Wdh2} into sums over $\mathfrak{S}_{\mathcal{I}}$ with a given $\mathcal{I}$, take the factors corresponding to the two-cycles out of the sum, and bound each of them as $O(|\log \delta|)$ using the bound $ K^{-1}_{\Omega_\delta}(b,w) = O\left(\frac{1}{|b-w|}\right),$ \eqref{eq:distance_bound} and~\eqref{eq:int_bound_log} as explained above. This yields
 \begin{equation}
        \label{eq:Wdh_resummed}
|F_\delta(x_1,\dots,x_n)|  \leq C\cdot \sum_{\mathcal{I}}|\log \delta|^{r(\mathcal{I})}\left|\sum_{\sigma\in\mathfrak{S}_{\mathcal{I}}}(-1)^{\mathrm{sign}(\sigma)}\prod_{j=r(\mathcal{I})+1}^{c(\sigma)}\left(\sum \prod_{i\in I^\sigma_{j}}K^{-1}_{\Omega_\delta}(b_i,w_{\sigma(i)})d(b_iw_i)^\ast\right)\right|.
    \end{equation}
For $m=1,\dots,7$, let $\Kterm[m](b,w)$ denote the $m$-th term in the expansion \eqref{eq:asymp_of_Kinv} of $K^{-1}_{\Omega_\delta}(b,w).$ Substituting \eqref{eq:asymp_of_Kinv} for $K^{-1}_{\Omega_\delta}(b_i,w_{\sigma(i)})$, the sum over $\sigma$ in the right-hand side of \eqref{eq:Wdh_resummed} becomes
\begin{equation}
\label{eq:sum_sigma_expanded}
\sum_{\sigma\in\mathfrak{S}_{\mathcal{I}}}(\dots)=\sum_{m:\in \{1,\dots,7\}^{I_{\mathcal{I}}}} \sum_{\sigma\in\mathfrak{S}_{\mathcal{I}}}(-1)^{\mathrm{sign}(\sigma)}\prod_{j=r(\mathcal{I})+1}^{c(\sigma)}\left(\sum \prod_{i\in I^\sigma_{j}}\Kterm[m_i](b_i,w_{\sigma(i)})d(b_iw_i)^\ast\right).
\end{equation}
We group the terms in this sum as follows.

\emph{Group 1.} The terms where $m_i=5$ for some $i$, i.e., $\Kterm[m_i](b_i,w_{\sigma(i)})=o(1)$. Substituting the bound $O\left(\frac{1}{|b_s-w_{\sigma(s)}|}\right)$ for all other $\Kterm[m_s]$ and using \eqref{eq:int_bound_const}, the term corresponding to $I^\sigma_j\ni i$ in \eqref{eq:sum_sigma_expanded} can be bounded by $o(1)\cdot O(1)=o(1),$ and all the other terms can be bounded using  \eqref{eq:int_bound_log} by $O(|\log \delta|)$. This gives a contribution to $|F_\delta|$ of
$
\sum_{\mathcal{I}}\sum_{\sigma\in\mathfrak{S}_{\mathcal{I}}}o(1)\cdot |\log \delta|^{c(\sigma)-1}\leq o(1)\cdot |\log \delta|^{\lfloor\frac{n-3}{2}\rfloor}.
$

\emph{Group 2.} The terms where $m_i=6$ for some $i$, i.e., $\Kterm[m_i](b_i,w_{\sigma(i)})=O\left(\frac{\delta}{|b_i-w_{\sigma(i)}|^2}\right)$. Let $j$ be such that  $I^\sigma_j\ni i$; bounding all other $\Kterm[m_s]$ by $O\left(\frac{1}{|b_s-w_{\sigma(s)}|}\right)$ and using \eqref{eq:int_bound_delta_minus1}, we get a bound of $O(1)$ for the term corresponding to $I^\sigma_j$. Note that since $|I^\sigma_j|\geq 3$, we can find $k\in I^\sigma_j$ such that $|b_k-w_{\sigma(k)}|\geq r/2$. If $k=i$, we bound $\Kterm[m_i](b_i,w_{\sigma(i)})$ by $O\left(\frac{\delta}{r|b_i-w_{\sigma(i)}|}\right)$ and all other $\Kterm[m_s]$ by $O\left(\frac{1}{|b_s-w_{\sigma(s)}|}\right)$; using \eqref{eq:int_bound_log} gives the bound of $O\left(\frac{\delta\log \delta}{r}\right).$ If $k\neq i$, we bound  $\Kterm[m_k](b_k,w_{\sigma(k)})$ by $O(1/r)$ and use \eqref{eq:int_bound_log_bis}; this gives a bound of $O\left(\frac{\delta|\log \delta|^{n-1}}{r}\right).$ Picking the smaller of the two bounds, and bounding other terms in the product $\prod_{j=r(\mathcal{I})+1}^{c(\sigma)}$ in \eqref{eq:Wdh_resummed} by $O(|\log \delta|)$, we get an overall bound of the form  $C\cdot \min\left(\frac{\delta\log \delta}{r},1\right)\cdot |\log \delta|^{\lfloor\frac{n-3}{2}\rfloor}.$

\emph{Group 3.} The terms where $m_i=7$ for some $i$. This is even easier than the previous case, since we can simply bound $O\left(\frac{\delta}{|b_i-\overline{w_{\sigma(i)}}|^2}\right)=O(\delta)$.

\emph{Group 4.} The terms where $m_i\in\{1,\dots,4\}$ for all $i$, and moreover for some choice of the signs $s_i,$ $i\in I_{\mathcal{I}}$, we have, for every $j>r(\mathcal{I})$,
\[\prod_{i\in I^\sigma_j}K_{m_i}^{-1}(b_i,w_{\sigma(i)})=\prod_{i\in I^\sigma_j}\frac{1}{4}F^{[s_is_{\sigma(i)}]}(b_i,w_{\sigma(i)})\int_{(b_iw_i)^\ast}dz^{[s_i]}.
\]
(Note that by \eqref{eq:dedge=int}, this simply means that for each $i$, the first sign in the superscript of $K_{m_i}^{-1}=(\dots)F^{[\cdot,\cdot]}$ agrees with the second sign in the superscript of $K_{m_{\sigma^{-1}(i)}}^{-1}=(\dots)F^{[\cdot,\cdot]}$.)
In this case, the innermost sums are Riemann sum approximations to the integrals in the right-hand side of \eqref{eq:vanishing}, which sum up to zero. Hence, this group equals the sum of several terms of of the form $\prod_{j=r(\mathcal{I})+1}^{c(\sigma)}Q_j$, where each $Q_j$ is either an integral of $\prod_{i\in I^\sigma_j}F^{[s_i,s_{\sigma(i)}]}(b_i,w_{\sigma(i)})dz^{[s_i]}$, or is an error of the Riemann sum approximation, with at least one $Q_j$ of the latter type. The former are of order $O(|\log \delta|)$. The error of the Riemann sum approximation can be bounded by $\delta$ times the integral of the norm of the gradient of $\prod_{i\in I^\sigma_j}F^{[s_i,s_{\sigma(i)}]}(b_i,w_{\sigma(i)})dz^{[s_i]}$, which in its turn can be bounded by
\[
C\cdot \sum_{k\in I^\sigma_j}\frac{1}{|b_k-w_{\sigma(k)}|}\prod_{i\in I^\sigma_j}\frac{1}{|b_{i}-w_{\sigma(i)}|}.
\] This is a bound of the same type as in the case of Group 2 above, and hence we get the same overall bound $C\cdot \min\left(\frac{\delta\log \delta}{r},1\right)\cdot |\log \delta|^{\lfloor\frac{n-3}{2}\rfloor}.$

\emph{Group 5.} The terms where $m_i\in\{1,\dots,4\}$ for all $i$, and for some $i$, the first sign in the superscript of $K_{m_i}^{-1}=(\dots)F^{[\cdot,\cdot]}$ disagrees with the second sign in the superscript of $K_{m_{\sigma^{-1}(i)}}^{-1}=(\dots)F^{[\cdot,\cdot]}$ In this case, instead of being a Riemann sum approximation to an integral, the sum over $l_i$ is approximately telescoping (cf. \cite[Proof of Proposition 20]{KenyonConfInvOfDominoTilings}): for two consecutive $(b_i,w_i)^\star$, $(\hat{b}_i\hat{w}_i)^\star$ on the same straight line segment of $l_i$, we have $\eta^2_{b_i}=-\eta^2_{\hat{b}_i}$ and $\eta^2_{w_i}=-\eta^2_{\hat{w}_i}$, hence
\begin{multline*}|K_{m_{\sigma^{-1}(i)}}^{-1}(b_{\sigma^{-1}(i)},w_i)K_{m_{i}}^{-1}(b_i,w_{\sigma(i)})+K_{m_{\sigma^{-1}(i)}}^{-1}(b_{\sigma^{-1}(i)},\hat{w}_i)K_{m_{i}}^{-1}(\hat{b}_i,w_{\sigma(i)})|\\
=\frac14\left| F^{[\cdot,s]}(b_{\sigma^{-1}(i)},w_i)F^{[\tilde{s},\cdot]}(b_i,w_{\sigma(i)})-F^{[\cdot,s]}(b_{\sigma^{-1}(i)},\hat{w}_i)F^{[\tilde{s},\cdot]}(\hat{b}_i,w_{\sigma(i)})\right|\\
\leq \frac{C\delta}{|b_{\sigma^{-1}(i)}-w_i|^2|b_i-w_{\sigma(i)}|}+ \frac{C\delta}{|b_{\sigma^{-1}(i)}-w_i||b_i-w_{\sigma(i)}|^2}.
\end{multline*}
Bounding other factors in the product $ \prod_{s\in I^\sigma_{j}}\Kterm[m_s](b_s,w_{\sigma(s)})$, where $I^\sigma_j\ni i$, by $O\left(\frac{1}{|b_s-w_{\sigma(s)}|}\right)$, we once again get the bound of the same type as in Groups 2 and 5 above, yielding the same overall bound.

Putting the bounds for all five groups together, we get the required result.

In the case that there are no paths $l_1,\dots, l_n$ with the required properties, we necessarily have $r=O(\delta)$, therefore, the last factor in \eqref{eq:Wick_rule_dimer_height_error} is $O(1)$. We can find points $\hat{x}_1,\dots,\hat{x}_n$ such that $\hat{x}_i$ can be connected to the straight segment on the boundary by paths $l_i$ with the above properties, and $x_i$ can be connected to $\hat{x}_i$ by a path $\hat{l}_i$ of length $O(\delta)$, such that moreover $l_i$ and $\hat{l}_j$ are never incident to the same vertex of $\Omega_\delta$ for $i\neq j$. For $S\subset \{1,\dots,n\}$, put $Q^S_i=h_{\delta,1}(\hat{x}_i)$ if $i\notin S$ and $Q^S_i=h_{\delta,1}(x_i)-h_{\delta,1}(\hat{x}_i)$ else. We can write
\begin{multline}
\label{eq:points_are_close_what_to_do}
\EE[h_{\delta,1}(x_1)\cdot\dots\cdot h_{\delta,1}(x_n)]=\EE[h_{\delta,1}(\hat{x}_1)\cdot\dots\cdot h_{\delta,1}(\hat{x}_n)]+\sum_{\emptyset\neq S\subset \{1,\dots,n\}}\EE\prod_{i=1}^n Q^S_i\\
=\sum_{p \text{ - pairing}}\prod_{i=1}^{n/2}\EE \left[h_{\delta,1}(x_{p(2i-1)}) h_{\delta,1}(x_{p(2i)})\right]+F_\delta(\hat{x}_1,\dots,\hat{x}_n)\\+\sum_{\emptyset\neq S\subset \{1,\dots,n\}}\left(\EE\prod_{i=1}^n Q^S_i-\sum_{p \text{ - pairing}}\prod_{i=1}^{n/2}\EE \left[Q^S_{p(2i-1)}Q^S_{p(2i)})\right]\right).
\end{multline}
We proceed to estimating the last term. Denote $E_{(bw)}:=\mathbb{I}[(bw)\in D_1]-\EE\mathbb{I}[(bw)\in D_1]$. Writing $h_{\delta,1}(x_i)-h_{\delta,1}(\hat{x}_i)$ in terms of $E_{(bw)},$ $(bw)^\ast\in \hat{l}_i$, we see that $\prod_{i\in S} Q_{i}^S$ is a homogeneous polynomial of degree $|S|$ in $E_{(bw)}$, $(bw)^\ast\in \hat{l}_1\cup\dots\cup \hat{l}_n$. In fact, it can be rewritten as a linear combination of monomials in $E_{(bw)}$, where each monomial only has factors corresponding to disjoint edges; this is because for non-disjoint edges $(bw),(\hat{b}\hat{w})$, we have $\mathbb{I}[(bw)\in D_1]\mathbb{I}[(\hat{b}\hat{w})\in D_1]\equiv 1$ or $\mathbb{I}[(bw)\in D_1]\mathbb{I}[(\hat{b}\hat{w})\in D_1]\equiv 0$. Let us denote this linear combintaion as $\sum_{U\in\mathcal{E}_S}\alpha_{U}\prod_{(bw)\in U}E_{bw}$, where $\mathcal{E}_S$ stands for the set of subsets of $\cup_{i\in S} \hat{l}_i$ of cardinality $\leq |S|$ containing only pairwise disjoint edges. We can now write, similarly to \eqref{eq:Wdh1} and the subsequent expansion,
\begin{multline}
\EE\prod_{i=1}^n Q^S_i=\sum_{U\in\mathcal{E}_S}\tilde{\alpha}_U\sum_{\prod_{i\in S^c}l_i}\det \left[K^{-1}_{\Omega_\delta}(b_i,w_j)\indic[i\neq j]\right]_{i,j\in S^c\cup U}\,\prod_{i\in S^c\cup U} d(b_iw_i)^\ast\\
=\sum_{U\in\mathcal{E}_S}\tilde{\alpha}_U\sum_{\sigma:|I^\sigma_1|\geq 2}(-1)^{\mathrm{sign}(\sigma)}\prod_{j=1}^{c(\sigma)}\sum \prod_{i\in I^\sigma_{j}}K^{-1}_{\Omega_\delta}(b_i,w_{\sigma(i)})d(b_iw_i)^\ast;
\end{multline}
where we are summing over permutations of $S^c\cup U$. Note that in this normalization, the coefficients $\tilde{\alpha}_U$ do not depend on $\delta$. As above, the contribution of each cycle not passing through $U$ can be bounded by $O(|\log \delta|),$ while for the cycles that do pass through $U$, we can use the bound $|K^{-1}_{\Omega_\delta}(b_i,w_{\sigma(i)})d(b_iw_i)^\ast|=O(1)$ for $i\in U$ and \eqref{eq:int_bound_const}, which gives a contribution of $O(1).$ We consider the following contributions to the last sum in \eqref{eq:points_are_close_what_to_do}, depending on $S$ and the permutation $\sigma$:

\emph{Case 1:} $|S|\geq 3$. In this case, the total number of cycles that do not pass through $U$ is at most $\lfloor\frac{n-3}{2}\rfloor$, thus we get the desired bound
$\EE\prod_{i=1}^n Q^S_i=O(|\log \delta|^{\lfloor\frac{n-3}{2}\rfloor}).$ Note also that for any pairing, at least two terms in the product $\prod_{i=1}^{n/2}\EE \left[Q^S_{p(2i-1)}Q^S_{p(2i)})\right]$ have $\{p(2i-1),p(2i)\}\cap S \neq \emptyset$, and thus can be bounded as $O(1)$, while the other terms can be bounded by $O(|\log \delta|)$. Thus, the sum over pairings is also bounded by $O(|\log \delta|^{\frac{n}{2}-2})=O(|\log \delta|^{\lfloor\frac{n-3}{2}\rfloor})$.

\emph{Case 2:} $|S|=1$. In this case, we in fact have $\tilde{\alpha}_\emptyset =0$, since the pairwise disjointness condition is void; in this case, $\EE\prod_{i=1}^n Q^S_i$ is given exacly by \eqref{eq:Wdh1} with summation over $l_i$ replaced by summation over $\hat{l}_i$ where $S=\{i\}$. If $\sigma$ has at least one cycle of length three, it has at most $\lfloor\frac{n-1}{2}\rfloor$ cycles, one of which passes through $U$, leading to a bound of $O(|\log \delta|^{\lfloor\frac{n-3}{2}\rfloor})$. If $\sigma$ corresponds to a pairing $p$, its contribution cancels out exactly with $\prod_{\{i,j\}\in p}\EE \left[Q^S_{i}Q^S_{j})\right].$

\emph{Case 3:} $|S|=2$ and $\sigma$ has at least one cycle of length three. In this case, either $U$ is empty and the total number of cycles in $\sigma$ is at most $\lfloor\frac{n-3}{2}\rfloor$, or $U$ is non-empty, the total number of cycles is at most $\lfloor\frac{n-1}{2}\rfloor$, but at least one of them passes through $U$. In both cases, we get a contribution of $O(|\log \delta|^{\lfloor\frac{n-3}{2}\rfloor})$.

\emph{Case 4:} $|S|=2$ and $n$ is even. The permutations unaccounted for are those correponding to pairings, which we group as follows: we fix $S$ and group together the terms with $U=\emptyset$ and a pairing $p$ of $\{1,\dots,n\}\setminus S$, and the terms with $U=S$ and the pairing $p\cup S$. The factors corresponding to the pairs of $p$ separate as $\prod_{\{i,j\}\in p}\EE[Q_i^S Q^S_j]=\prod_{\{i,j\}\in p}\EE [h_{1,\delta}(\hat{x}_i)h_{1,\delta}(\hat{x}_j)]$, and the remaining factors combine to $\EE[Q_{i_1}^S Q^S_{i_2}]$, where $\{i_1,i_2\}=S$. (Take into account that first-degree monomials in the expansion of $Q_{i_1}^S Q^S_{i_2}$ do not contrbute to the expectation since $\EE E_{(bw)}=0$.) We conclude that this group cancels out with $\prod_{\{i,j\}\in p\cup S}\EE[Q_i^S Q^S_j]$. The remaining terms are those where $U=S$ and $\sigma$ has two two-cycles passing through $S$; the contribution of those terms, as well as the corresponding products  $\prod_{\{i,j\}\in p}\EE[Q_i^S Q^S_j]$, is $O(|\log \delta|^{\frac{n}{2}-2})$.

\emph{Case 5:} $|S|=2$ and $n$ is odd. The terms unaccounted for are those with $|U|=1$ and premutations corresponding to pairings. These permutations has $\frac{n-1}{2}$ two-cycles, one of which passes through $U$. This yields a contribution of $O(|\log \delta|^{\frac{n-3}{2}})=O(|\log \delta|^{\lfloor\frac{n-3}{2}\rfloor}).$

    Finally, to derive the same statement for $h_\delta$, recall that $h_\delta=h_{\delta,1}-h_{\delta,2}$, where $h_{\delta,2}$ is an independent copy of
    $h_{\delta,1}$; hence $h_\delta$ would satisfy the exact Wick's rule if  $h_{\delta,1}$ and $h_{\delta,2}$ did. We can therefore write
    \begin{multline}
    \EE[h_{\delta}(x_1)\cdot\dots\cdot h_{\delta}(x_n)]=\sum_{S\subset \{1,\dots,n\}}(-1)^{|S|}\EE\left[\prod_{i\in S} h_{\delta,2}(x_i)\prod_{i\in S^c} h_{\delta,1}(x_i)\right]\\
    = \sum_{S\subset \{1,\dots,n\}}(-1)^{|S|}\sum_{p,\hat{p}\text{ - pairings of } S,S^c}\prod_{\{i,j\}\in p}\EE\left[h_{\delta,2}(x_i)h_{\delta,2}(x_j)\right]\prod_{\{i,j\}\in \hat{p}}\EE\left[h_{\delta,1}(x_i)h_{\delta,1}(x_j)\right] \\+ \tilde{F}_\delta(x_1,\dots,x_n)\\
    = \sum_{p\text{ - pairing of } \{1,\dots,n\}}\prod_{\{i,j\}\in p}\EE\left[h_{\delta}(x_i)h_\delta(x_j)\right]+\tilde{F}_\delta(x_1,\dots,x_n),
    \end{multline}
    where the error term can be written as
    \[
    \tilde{F}_\delta(x_1,\dots,x_n)=\sum_{S=\{i_1,\dots,i_{|S|}\}}(-1)^{|S|}\left(F_\delta(x_S)\EE\prod_{i\in S^c}h_{\delta,1}(x_i)+F_\delta(x_{S^c})\EE\prod_{i\in S}h_{\delta,2}(x_i)-F_\delta(x_{S})F_\delta(x_{S^c})\right),
    \]
where for $u\subset \{1,\dots,n\}$, $F_{\delta}(x_U)$ stands for $F_{\delta}(x_{i_1},\dots,x_{i_{|U|}})$ as in \eqref{eq:Wick_rule_dimer_height_error}. By \eqref{eq:Wick_rule_dimer_height} and the $O(|\log \delta|)$ bound for the two-point function, we have the bound
\[
\EE\prod_{i\in U}h_{\delta,1}(x_i)\leq C\cdot \begin{cases}|\log \delta|^{\frac{|U|}{2}}&|U|\text{ even,}\\
|\log \delta|^{\frac{|U|-3}{2}}&|U|\text{ odd.}
\end{cases}
\]
Since $r$ is decreasing under the inclusion order on the collections of marked points, we get the desired bound on $|\tilde{F}_\delta(x_1,\dots,x_n)|$ by plugging in the above estimate and \eqref{eq:Wick_rule_dimer_height_error}.
\end{proof}

\begin{lemma}
  \label{lemma:asymptotics_of_h(0)2}
Given a compact subset $K$ of $\Omega$, we have for any $x,y\in K$
  \[
    \EE h_{\delta}(x)h_{\delta}(y) = -\frac{1}{\pi^2} \log\max(|x-y|,\delta) + O(1),
  \]
 uniformly in $x,y\in K$ and in $\delta>0$.
\end{lemma}
\begin{proof}
  Choose two dual lattice edge disjoint paths $l_1,l_2$ connecting $x,y$ with the straight boundary segment of $\Omega_\delta$ as it was done in the proof of Lemma~\ref{lemma:Wick_rule_dimer_height}; in this case we can moreover always choose them to satisfy the properties as in the beginning of the proof of that Lemma.  We start with analysing the correlation of the single-dimer height function. Using ~\eqref{eq:Wdh1}, plugging in the expansion \eqref{eq:asymp_of_Kinv} and grouping the terms as in the proof of Lemma \ref{lemma:Wick_rule_dimer_height}, we can write
   \begin{multline}
    \label{eq:Eh2_via_K}
    \EE h_{\delta, 1}(x)h_{\delta, 1}(x) = \\
    = \sum_{\substack{(b_1w_1)^*\in l_1,\\ (b_2w_2)^*\in l_2}} -K_{\Omega_\delta}(b_1,w_2)K_{\Omega_\delta}(b_2,w_1)d(b_1w_1)^\ast d(b_2w_2)^\ast
    \\
    =-4^{-2}\sum_{s_1,s_2\in \{+,-\}}\int_{l_1}\int_{l_2\setminus{B_\delta(x)}}F^{[s_1s_2]}(z_1,z_2)F^{[s_2s_1]}(z_2,z_1)\,dz_1^{[s_1]}dz_2^{[s_2]}+O(1);
  \end{multline}
observe that in this case, all the groups except for Group 4 will give a contribution of $O(1)$. Furthermore, since $F^{[+-]}$ and $F^{[-+]}$ are bounded up to diagonal, the corresponding integrals also contribute of $O(1)$. Plugging in the expansions $F^{[++]}(z_1,z_2)=\frac{2}{\pi i(z_1-z_2)}+R(z_1,z_2)$ and $F^{[--]}(z_1,z_2)=-\frac{2}{\pi i(\bar{z}_1-\bar{z}_2)}+\overline{R(z_1,z_2)}$, where $R(z_1,z_2)=O(1)$, and using \eqref{eq:int_bound_const} to bound the contributions from terms like $F^{[++]}(z_1,z_2)R(z_2,z_1)$ etc., we get
\begin{multline*}
    \EE h_{\delta, 1}(x)h_{\delta, 1}(x) = \\
    =-\frac{1}{4\pi^2}\left(\int_{l_1}\int_{l_2\setminus B_\delta(x)}\frac{dz_1dz_2}{(z_1-z_2)^2}+\int_{l_1}\int_{l_2\setminus B_\delta(x)}\frac{d\bar{z}_1d\bar{z}_2}{(\bar{z}_1-\bar{z}_2)^2}\right)+O(1)\\
    =-\frac{1}{2\pi^2}\log\max(|x-y|,\delta)+O(1).
  \end{multline*}
The result for $h_\delta$ follows since $h_\delta=h_{\delta,1}-h_{\delta,2}$ with $h_{\delta,1}$ and $h_{\delta,2}$ independent.

\end{proof}

We are ready to deduce the asymptotics of $\mu_\delta$ and $\sigma_\delta$ in Theorem 1:
\begin{lemma}
    \label{lemma:asymp_of_mu_delta_sigma_delta}
    We have
    \begin{equation}
      \label{eq:asymp_of_mu}
      \mu_\delta = -\frac{1}{\pi^2}\log\delta + O(1),\qquad \sigma^2_\delta = -\frac{2}{3\pi^2}\log\delta + O(1).
    \end{equation}
    as $\delta\to 0+$.
\end{lemma}
\begin{proof}
The asymptotics of $\mu_\delta$ is immediate from \eqref{eq:Nxy_from_h} and Lemma \ref{lemma:asymptotics_of_h(0)2} by putting $x=y$. For $\sigma_\delta$, we use \eqref{eq:N2viah4}, \eqref{eq:Nxy_from_h} and Lemma \ref{lemma:Wick_rule_dimer_height} to get
\begin{multline*}
\sigma_\delta^2=\EE N^2_{\Omega_\delta}(x)-\mu^2_\delta=\frac{1}{3}\EE h^4_\delta(x) +\frac{2}{3}\mu_\delta - \mu_\delta^2 =\left(\EE h^2_\delta(x)\right)^2+\frac23 \mu_\delta -\mu_\delta^2 +O(1)=\frac{2}{3}\mu_\delta+O(1).
\end{multline*}
\end{proof}

A as a peculiar byproduct of this results, we get a central limit theorem for the \emph{single dimer} height function; a similar result can be obtained from the computations in \cite{pinson2004rotational}:
\begin{cor}
    For a fixed $x\in\Omega$, we have the following convergence in distribution for the centered single dimer height function:
    \[
    \sqrt{-\frac{2\pi^2}{\log \delta}}h_{\delta,1}(x)\stackrel{\delta\to 0}{\longrightarrow} \Nn(0,1).
    \]
\end{cor}
\begin{proof}
Since the left-hand side has variance $1$, it is tight as $\delta\to 0$, hence we only need to prove that the only possible subsequential limit is a standard Gaussian. From the relation \eqref{eq:h_delta_sum_bernoulli}, we know that $h_{\delta,1}(x)-h_{\delta,2}(x)$ is a sum of $\NOd(x)$ i.i.d. Bernoulli random variables. Let $E_R$ denote the event $|\NOd+\frac{1}{\pi^2}\log\delta|\leq R\sqrt{-\log \delta}$; from Lemma \ref{lemma:asymp_of_mu_delta_sigma_delta}, we have $\PP[E_R]=1-O\left(R^{-2}\right)$ and also by Central limit theorem \begin{multline*}\PP\left[\left.\sqrt{-\frac{\pi^2}{\log\delta}}(h_{\delta,1}(x)-h_{\delta,2}(x))\leq \lambda \right|E_R\right]=\PP\left[\left.\frac{1}{\sqrt{\NOd(x)}}(h_{\delta,1}(x)-h_{\delta,2}(x))\leq \lambda+o(1)\right|E_R\right]\\ =\PP[\Nn(0,1)\leq \lambda]+o(1).
\end{multline*} Hence, $\sqrt{-\frac{\pi^2}{\log\delta}}\left(h_{\delta,1}-h_{\delta,2}\right)$ converges in distribution to the standard Gaussian, that is, any subsequential limit $h_1$ of $h_{\delta,1}$ has the property that $\frac{1}{\sqrt{2}}(h_{1}-h_2)$ is a standard Gaussian, where $h_2$ is an independent copy of $h_1$. We claim that this implies that $h_1$ itself is a Gaussian. Let $M$ be the median of $h_1$, then for a real non-negative $a$, we have \[\exp(a^2)=\mathbb{E}e^{a(h_1-h_2)}\geq \mathbb{E}\left[e^{a(h_1-M)}\mathbf{1}[h_2\leq M]\right]\geq\frac12e^{-aM}\mathbb{E}e^{a h_1},\] and similarly for a negative real $a$. It follows that the characteristic function $f(z)=\mathbb{E}\exp(i z h_1)$ is defined for all $z\in \CC$, and it is an entire function of order $2$. Also, we have $f(z)f(-z)=\exp(-z^2),$ hence $f$ has no zeros. Therefore by Hadamard factorization, $f(z)=\exp(az^2+bz+c)$ for some $a,b,c\in \CC$; the condition  $f(z)f(-z)=\exp(-z^2)$ further restricts it to $f(z)=\exp(-z^2/2+bz),$ and $b=i\EE[h_1]=0.$
\end{proof}

\subsection{Loop statistics and Kasteleyn operator with an \texorpdfstring{$\SL(2,\CC)$}{SL(2,C)} monodromy}
\label{subsec:Loop statistics and Kasteleyn operator with an SL2 monodromy}

Following~\cite{kenyon2014conformal, DubedatDoubleDimers} we define a Kasteleyn operator with inserted $\SL(2,\CC)$ monodromy as follows. Let $\lambda_1,\dots, \lambda_n$ be a collection of distinct faces of $\Omega_\delta$ and assume that there exist a collection of edge-disjoint dual lattice paths $l_1,\dots, l_n$ connecting these faces to the boundary of $\Omega_\delta$. Let $\gamma_1,\dots, \gamma_n$ be loops representing generators of $\pi_1(\Omega_\delta\smm\{\lambda_1,\dots, \lambda_n\})$ (the base point is chosen arbitrary) chosen such that $\gamma_i$ makes one counterclockwise turn around $\lambda_i$ and does not intersect $l_j$ for $j\neq i$. This choice provides a bijection between $n$-tuples $(\rho_1,\dots, \rho_n)$ of $\SL(2,\CC)$ matrices and representations $\rho: \pi_1(\Omega_\delta\smm\{\lambda_1,\dots, \lambda_n\})\to \SL(\CC,2)$.

Given a representation $\rho$ (and the corresponding tuple of matrices) we modify the operator $K_{\Omega_\delta}$ as follows. Orient each $l_i$ towards the boundary.
\begin{equation}
    \label{eq:def_of_KOmega_delta_rho}
    K_{\Omega_\delta, \rho}(w,b) = \begin{cases}
        K_{\Omega_\delta}(w,b)\cdot \Id_{2\times 2}, \qquad wb\text{ does not cross }l_1,\dots, l_n,\\
        K_{\Omega_\delta}(w,b) \cdot \rho_i,\qquad wb\text{ crosses $l_i$ and $b$ is on the left},\\
        K_{\Omega_\delta}(w,b) \cdot \rho_i^{-1},\qquad wb\text{ crosses $l_i$ and $b$ is on the right}.
    \end{cases}
\end{equation}
By this, we simply mean a matrix of twice the size of $K_{\Omega_\delta}$, i.e., with columns (respectively, rows) indexed by two copies of the set of black (respectively, white) vertices of $\Omega^\delta$, where each entry of $K_{\Omega_\delta}$ is replaced by a corresponding $2\times 2$ block. Note that $K_{\Omega_\delta, \rho}$ is intertwined with $K_{\Omega_\delta}$ acting on the space of multivalued functions with monordomy $\rho$. We have
\begin{lemma}
    \label{lemma:det_Krho}
    Denote by $\Ll$ the double-dimer loop ensemble (the collection of unoriented loops corresponding to the random double-dimer configuration). Then we have
    \begin{equation}
        \det (K_{\Omega_\delta, \rho}K_{\Omega_\delta, \Id}^{-1}) = \EE \prod_{\gamma\in \Ll}\frac{\Tr\rho(\gamma)}{2}.
    \label{eq:det_to_loop_stat}
    \end{equation}
\end{lemma}

\begin{proof}
    In the case when $\Omega_\delta$ is finite this formula was established by Kenyon~\cite[Theorem~2]{kenyon2014conformal} (see also the remark after the proof of this theorem which justifies our use of the usual determinant in place of the Q-determinant used by Kenyon). When $\Omega_\delta = \CC^+_\delta$,  this formula follows by taking a limit over a sequence of finite Temperley domains~\cite[Section~5.4]{DubedatDoubleDimers}.
\end{proof}

\begin{rem}
\label{ref: infinite_determinants}
In the case  $\Omega_\delta = \CC^+_\delta$, the operator $L_\rho:=K_{\CC^+_\delta, \rho}K_{\CC^+_\delta, \Id}^{-1}$ acts on an infinite-dimensional space. We recall from in~\cite[Section~5.4]{DubedatDoubleDimers} the following interpretation of the determinant: we have that $L_\rho(u,w)=\delta_{u,w}\cdot \Id_{2\times w}$ unless $u$ is adjacent to one of $l_i$. In particular, if $V$ denotes the set of vectors which are zero outside of white vertices adjacent to the cuts, then $V$ is an invariant subspace for $L_{\rho}$, and so is any $V'\supset V$, and $\det L_{\rho}|_{V'}=\det L_{\rho}|_{V},$ which is taken as the definition of the determinant in \eqref{eq:det_to_loop_stat}.

For $w$ adjacent to a cut, we only have a bound $L_\rho(w,u)=O(|u|^{-1}),$ in particular, $L_{\rho}$ does not act on $L^2(\CC^+_\delta).$ However, it does act in a weighted $L^2(\CC^+_\delta),$ say with the weight $(|x|+1)^{-1}$; which puts the above interpretation into the standard framework of Fredholm determinants, see e.g. \cite{gohberg1978introduction}.
\end{rem}

We use Lemma~\ref{lemma:det_Krho} to get the following estimate.
\begin{equation}
    \label{eq:def_of_N(lambda)}
    \NOd(x_1,\dots, x_n) = \#\text{ of double-dimer loops surrounding }\{x_1,\dots, x_n\}.
\end{equation}

\begin{lemma}
    \label{lemma:moments_of_Nxy}
    For any compact $K\subset \Omega$ and $j>0$ there exists a constant $C>0$ such that for any faces $x\neq y\in K$ of $\Omega_\delta$ we have
    \[
        \left|\EE \NOd(x,y)^j - \left( -\frac{1}{\pi^2}\log|x-y| \right)^j\right| \leq C \left(|\log|x-y||^{j-1} + 1\right).
    \]
\end{lemma}
\begin{proof}
    (Cf. \cite[Section 10]{kenyon2014conformal}). Choose edge disjoint dual lattice paths $l_1,l_2$ connecting $x,y$ to the straight segment on the boundary of $\Omega_\delta$ satisfying the properties as described in the beginning of the proof of Lemma \ref{lemma:Wick_rule_dimer_height}, and define $\rho_t$ by declaring the corresponding 2-tuple to be
    \[
        \rho_{t1} = \begin{pmatrix} 1 & t \\ 0 & 1 \end{pmatrix},\qquad \rho_{t2} = \begin{pmatrix} 1 & 0 \\ t & 1 \end{pmatrix}.
    \]
   Denote $N=\NOd(x,y)$. Lemma~\ref{lemma:det_Krho} implies that
    \begin{equation}
        \label{eq:moN1}
        \det (K_{\Omega_\delta, \rho}K_{\Omega_\delta, \Id}^{-1}) = \EE \left(\frac{2+t^2}{2}\right)^{N}.
    \end{equation}
   Expanding the right-hand side in a Taylor series and using induction, we get
    \begin{equation}
        \label{eq:moN2}
        \left.\frac{d^{2j}}{dt^{2j}}\EE \left(\frac{2+t^2}{2}\right)^{N}\right|_{t=0} = \frac{(2j)!}{2^j}\EE{N \choose j}=\frac{(2j)!}{2^j j!}\EE N^j+O\left(|\log(x-y)|^{j-1}+1\right).
    \end{equation}
    To evaluate the derivative of the left-hand side of~\eqref{eq:moN1} we note that we can write $K_{\Omega_\delta, \rho}K_{\Omega_\delta, \Id}^{-1}=\Id+tA$, where, recalling \eqref{eq:def_of_dedge},
    \begin{equation}
    A(w,u)=\sum_{\substack{b\sim w\\ (bw)\cap l_1\neq \emptyset}}K^{-1}_{\Omega_\delta}(b,u)\begin{pmatrix} 0 & d(bw)^\star \\ 0 & 0 \end{pmatrix}+\sum_{\substack{b\sim w\\ (bw)\cap l_2\neq \emptyset}}K^{-1}_{\Omega_\delta}(b,u)\begin{pmatrix} 0 & 0 \\ d(bw)^\star & 0 \end{pmatrix}.
    \label{eq: frperturbation}
    \end{equation}
    By Fredholm expansion, we have
   \[
\left.\frac{d^{2j}}{dt^{2j}}\det(\Id+tA)\right|_{t=0}=(2j)!\Tr\left[\Lambda^{2j}(A)\right]=\sum_{\substack{w^{[s_1]}_1,\dots,w^{[s_{2j}]}_{2j}\\ \text{distinct}}}\sum_{\sigma}(-1)^{\mathrm{sign}(\sigma)}\prod_{i=1}^{2j}A(w^{[s_i]}_i,w^{[s_{\sigma(i)}]}_{\sigma(i)}),
    \]
where for $x\in \Omega^\delta$ we denote by $x^{[1]}$ and $x^{[2]}$ the two copies of $x$ used in the indexing of columns and rows of $A$. We observe from \eqref{eq: frperturbation} that the last product is zero unless, for every $i$, we have either $s_i=1$, $s_{\sigma(i)}=2$, and $w_i$ adjacent to $l_1$, or $s_i=2$, $s_{\sigma(i)}=1$, and $w_i$ adjacent to $l_2$. Moreover, changing the order of summation, plugging in \eqref{eq: frperturbation} and expanding, we can rewrite this as
\[
\left.\frac{d^{2j}}{dt^{2j}}\det(\Id+tA)\right|_{t=0}=\sum_{\sigma,s}(-1)^{\mathrm{sign}(\sigma)}\sum_{\substack{(b_1w_1)^\star \in l_{s_1},\dots,(b_{2j}w_{2j})^\star\in l_{s_{2j}}\\ \text{distinct}}}\prod_{i=1}^{2j}K^{-1}_{\Omega_\delta}(b_i,w_{\sigma(i)})d(b_iw_i)^\star,
\]
where we sum over $s$ alternating along the cycles of $\sigma$; in particular, only $\sigma$ with all cycles of even length contribute. Using the bound $K^{-1}_{\Omega_\delta}(b,w)=O(|b-w|)$ and \eqref{eq:int_bound_log}, we see that the contribution of a permutation $\sigma$ is bounded by $O(|\log |x-y|+1|^{c(\sigma)})$, which is $O(|\log |x-y|+1|^{j-1})$ unless $\sigma$ has all cycles of length $2$. Similarly, using that $K^{-1}(b_i,w_{\sigma(i)})d(b_iw_i)^\ast = O(1)$ and \eqref{eq:int_bound_const}, one observes that dropping the restriction that $(b_iw_i)^\star$ are distinct introduces an error of order $O(|\log |x-y|+1|^{j-1})$. Putting these observations together and re-labeling the terms in the product, we get
\begin{multline*}
\left.\frac{d^{2j}}{dt^{2j}}\det(\Id+tA)\right|_{t=0}\\
=\sum_{p\ \text{pairing}}\sum_{\substack{s:s_i\neq s_k \\ \text{ for }\{i,k\}\in p}}\prod_{\{i,k\}\in p}\left(\sum_{(b_iw_i)^\star \in l_{s_{i}},(b_{k}w_{k})^\star\in l_{s_{k}}}-K^{-1}_{\Omega_\delta}(b_i,w_{k})K^{-1}_{\Omega_\delta}(b_k,w_{i})d(b_iw_i)^\star d(b_kw_k)^\star\right)\\+O(|\log |x-y|+1|^{j-1})\\
=2^{j}(2j-1)!!\left(\sum_{(b_1w_1)^\star \in l_{1},(b_{2}w_{2})^\star\in l_{2}}-K^{-1}_{\Omega_\delta}(b_1,w_{2})K^{-1}_{\Omega_\delta}(b_2,w_{1})d(b_1w_1)^\star d(b_2w_2)^\star\right)^j+O(|\log |x-y|+1|^{j-1}).
\end{multline*}
Comparing this with \eqref{eq:moN2} and \eqref{eq:moN1}, and taking into account that $(2j-1)!!2^j j!=(2j)!$, we get
\begin{multline*}
\EE N^j=\left(\sum_{(b_1w_1)^\star \in l_{1},(b_{2}w_{2})^\star\in l_{2}}-2K^{-1}_{\Omega_\delta}(b_1,w_{2})K^{-1}_{\Omega_\delta}(b_2,w_{1})d(b_1w_1)^\star d(b_2w_2)^\star\right)^j+O(|\log |x-y|+1|^{j-1})\\
=(\EE N)^j+O(|\log |x-y|+1|^{j-1}).
\end{multline*}
We conclude by recalling that $\EE N =\EE [h_{\delta}(x) h_{\delta}(y)]= -\frac{1}{\pi^2}\log|x-y| + O(1)$ by Lemma~\ref{lemma:asymptotics_of_h(0)2}.
\end{proof}

\section{Normalized fluctuations of \texorpdfstring{$\NOd(x)$}{N(x)} and Kasteleyn operator with monodromy}
\label{sec:proof_of_main1}

The goal of this section is to prove Theorem~\ref{thma:main1}. We have already established the asymptotic of $\mu_\delta$ and $\sigma_\delta$ in Lemma~\ref{lemma:asymp_of_mu_delta_sigma_delta}, so we are left to prove that $(\NOd(v) - \mu_\delta)\sigma_\delta^{-1}$ converges to a standard normal variable in distribution. We begin by analysing the asymptotic of the inverse operator $K_{\Omega_\delta, \rho}^{-1}$ (cf.~\eqref{eq:def_of_KOmega_delta_rho}) in the case when $\Omega_\delta = \delta \ZZ^2$ and $\rho$ is of the form $\begin{pmatrix}
    e^{2\pi i s} & 0 \\ 0 & e^{-2\pi i s}
\end{pmatrix}$. In this case $K_{\Omega_\delta, \rho}$ can be interpreted as a direct sum of two Kasteleyn operators with scalar monodromies inserted, so we can focus on analyzing such an operator.
It should be mentioned that analytic aspects of Kasteleyn operators with scalar monodromies were studied by Dub\'edat in~\cite[Section~7]{DubedatFamiliesOfCR}. Our goal is to refine the analysis performed by Dub\'edat to get precise enough near-diagonal asymptotic of the inverse operator in the regime when $s$ tends to zero with the rate $\left(-\log \delta\right)^{-1/4}$.

\subsection{Full-plane operator with monodromy around one point}
\label{subsec:full-plane-Kinv}

We begin by considering the case of one puncture. Note that by scale covariance it is enough to assume that $\delta = 1$. We also shift the lattice to make the origin to be the center of a face for notational convenience; thus, in this section we will consider the Kasteleyn operator on $\ZZ^2 + \frac{1+i}{2}$. We keep the previous definition of the function $\eta$, that is, in the new notation we have
\[
    \eta_u = \begin{cases}
        1,\qquad \Im u \in 2\ZZ + \frac{1}{2},\\
        -i,\qquad \Im u \in 2\ZZ - \frac{1}{2}.
    \end{cases}
\]
Let $l$ be an infinite simple dual lattice oriented path emerging from the origin; we may take $l$ to go straight down for definiteness. We need the following auxiliary notation. Given two vertices (maybe of the same color) $u_1$ and $u_2$ of $\ZZ^2 + \frac{1+i}{2}$, let us say that $u_1$ is on the right and $u_2$ is on the left of $l$ is the straight line segment from $u_1$ to $u_2$ intersects $l$, and the algebraic intersection of $l$ and this segment is positive. For a given $s\in \RR$ we put
  \begin{equation}
    \label{eq:def_of_chi_extended}
    \chi_{s,l}(u_1,u_2) = \begin{cases}
      e^{2\pi i s},\quad u_1\text{ is on the right of $l$ and $u_2$ is on the left,}\\
      e^{-2\pi i s},\quad u_2\text{ is on the left from $l$ and $u_1$ is on the left,}\\
      1,\quad \text{else}.
    \end{cases}
  \end{equation}
Recall that $K$ denotes the full-plane Kasteleyn operator, see Section~\ref{sec:Height_function_loop_statistics_and_monodromy}. Define
\begin{equation}
  \label{eq:def_of_Ks}
    K_s(w,b) = \chi_{s,l}(w,b)K(w,b).
\end{equation}
This notation is temporary for the current section; it should not be mixed with the notation $K_\delta$ for the rescaled version of $K$. Below, we will use functions like $(b/w)^s$, these are to be understood as a branch in $\mathbb{C}\setminus l$ equal to $1$ at $b=w$. We note that $K_s$ can be understood as a discretization of the Cauchy-Riemann operator acting on multi-valued function with multilicative monodromy  $e^{2\pi i s}$ at the origin; thus we will refer to such functions $f$ satisfying $K_sf=0$ as \emph{discrete holomorphic}. Similarly to the case of single-valued discrete holomorphic functions, a multivalued discrete holomorphic function on a square lattice give rise to a pair of discrete harmonic functions with the same monodromy on two sublattices which, however, fail to be harmonic at the origin. This in particular implies a version of a maximal principle for discrete holomorphic functions which we formulate now.

Let $\Gamma = (2\ZZ + 1/2)^2,\Gamma^\dagger = (2\ZZ + 3/2)^2$ be the two black sublattices. Following the definition~\eqref{eq:def_of_Ks} of $K_s$ we introduce the corresponding Laplacians: for $f:\Gamma\to\CC$, we put
\begin{equation}
  \label{eq:def_of_Deltas}
    (\Delta_s f)(b) = \sum_{\tilde{b}\sim b} (\chi_{l,s}(b,\tilde{b})f(\tilde{b}) - f(b)),\\
\end{equation}
and for $f:\Gamma^\dagger\to \CC$ we define $\Delta_s^\dagger f$ in the same way. As above, $\Delta_s,\Delta_s^\dagger$ can be understood as Laplacians acting on functions with multiplicative monodromy $e^{2\pi i s}.$

\begin{lemma}
    \label{lemma:maximum_principle_for_Delta_s}
    Assume that $h$ is a (complex-valued) function on $\Gamma$ (resp.~$\Gamma^\dagger$) and $\Omega$ is a subset of vertices such that $\Delta_s h = 0$ (resp. $\Delta_s^\dagger h = 0$) holds at any vertex of $\Omega$. Then
    \[
        \max_{v} |h(v)| = \max_{v\in \partial\Omega} |h(v)|,
    \]
    where $\partial\Omega$ is the set of vertices of $\Gamma\setminus\Omega$ (resp. $\Gamma^\dagger\setminus \Omega$) adjacent to $\Omega$.
\end{lemma}
\begin{proof}
    Follows from~\eqref{eq:def_of_Deltas}.
\end{proof}
Let $b_0=\frac12+\frac{i}{2}\in \Gamma,b_0^\dagger=-\frac12-\frac{i}{2}\in \Gamma^\dagger$ be the black vertices incident to the face containing $0$.

\begin{lemma}
  \label{lemma:Laplacian_with_monodromy}
 If $b\in \Gamma\cup \Gamma^\dagger$ and a function $f$ is a discrete holomorphic with monodromy $e^{2\pi is}$ (i.e., satisfies $(K_sf)(w)=0$) at all white neighbors of $b$, then we have
  \begin{equation}
    \label{eq:Laplace_at_zero}
    \begin{split}
      &\Delta_s f(b) = \begin{cases}
      0,\quad b \neq b_0,\\
      i (e^{-2\pi i s}-1)f(b_0^\dagger),\quad b = b_0,
    \end{cases}\qquad b\in \Gamma, \\
      &\Delta_s^\dagger f(b) = \begin{cases}
      0,\quad b \neq b^\dagger_0,\\
      i (1-e^{2\pi i s})f(b_0),\quad b = b_0^\dagger,
    \end{cases}\qquad b\in \Gamma^\dagger
    \end{split}
  \end{equation}
\end{lemma}
\begin{proof}
    Note that for each discrete holomorphic function $f$ with monodromy $e^{2\pi is}$ we have $K_s^\ast K_s f = 0$. On the other hand, it is a straightforward computation to show that for each $b\in \Gamma\cup \Gamma^\dagger$ such that $b\notin\{b_0,b_0^\dagger\}$ and for an arbitrary $f$ we have $(K_s^\ast K_sf)(b) = (\Delta_s f)(b)$ if $b\in \Gamma$ and $(K_s^\ast K_sf)(b) = (\Delta_s f)(b)$ if $b\in \Gamma^\dagger$. It remains to analyze the situation near the origin. Straightforward computation shows that
    \begin{multline*}
        0 = (K_s^\ast K_s f)(b_0) =\\
        = (\Delta_s f)(b_0) + f(b_0^\dagger)(\overline{K(b_0-1,b_0)}K(b_0-1,b_0^\dagger) + \overline{K(b_0-i, b_0)}K(b_0-i, b_0^\dagger)e^{-2\pi i s} =\\
        = (\Delta_s f)(b_0) + i(1 - e^{-2\pi i s})f(b_0^\dagger)
    \end{multline*}
    which implies the first claim. The case of $b_0^\dagger$ is similar.
\end{proof}

\begin{cor}
(Maximum principle) If $R>0$ and $f$ satisfies $(K_sf)(w)\equiv 0$ whenever $|w|\leq R$, then
\begin{equation}
\label{eq:maximum_principle}
|f(b)|\leq \max_{R-2\leq |b'|\leq R}|f(b')|+2\max\{|f(b_0)|,|f(b^\dagger_0)|\}, \quad |b|\leq R.
\end{equation}
\end{cor}
\begin{proof}
Let $h, h^\dagger$ be restrictions of $f$ to $\Gamma$ and $\Gamma^\dagger$ respectively. Put $\Omega = \Gamma\cap \bar{B}(0,R)$ and $\Omega^\dagger = \Gamma^\dagger \cap \bar{B}(0,R)$. By Lemma~\ref{lemma:Laplacian_with_monodromy} $h$ (resp. $h^\dagger$) is a multivalued harmonic function on $\Omega \smm\{b_0\}$ (resp. $\Omega^\dagger\smm\{b_0^\dagger\}$). Eq.~\eqref{eq:maximum_principle} now follows from Lemma~\ref{lemma:maximum_principle_for_Delta_s}.
\end{proof}

We begin by collecting some facts proven in~\cite{DubedatFamiliesOfCR}.

\begin{lemma}
    \label{lemma:multivalued_Green_function}
    There exists a function $G_s(b,b_1)$ on $\Gamma\times \Gamma$ (resp. $\Gamma^\dagger\times \Gamma^\dagger$) and a constant $C>0$ such that the following holds:
    \begin{enumerate}
        \item For each $b_1$ we have $\Delta_s G_s(\cdot, b_1) = \delta_{b_1}(\cdot)$ (resp. $\Delta^\dagger_s G_s(\cdot, b_1) = \delta_{b_1}(\cdot)$).
        \item If $b,b_1$ satisfy $|b_1|/2\leq |b|\leq 2|b_1|$, and $b_1'\sim b_1$ (i.e. there is an edge between $b_1'$ and $b_1$ in $\Gamma$ (resp. $\Gamma^\dagger$)) and $w$ is the white vertex lying on the edge $b_1b_1'$, then
        \[
            |G_s(b,b_1') - G_s(b,b_1)| \leq C|b-w|^{-1}.
        \]
    \end{enumerate}
\end{lemma}
\begin{proof}
    The lemma is a compilation of~\cite[Lemma~9]{DubedatFamiliesOfCR}. Note that in the notation~\cite{DubedatFamiliesOfCR}, $M_V$ denote the set of vertices of the graph $\Gamma$, while $(\CC^{M_V})_\chi$ denote the set of functions which act on the pullback of $M_V$ to the universal cover of $\CC\smm\{0\}$ and are multiplied by $\chi$ under the action of deck transformations. To construct our function $G_s$ out of~\cite[Lemma~9]{DubedatFamiliesOfCR} we should put $\chi = e^{2\pi i s}$ and restrict the function $G_\chi$ from~\cite[Lemma~9]{DubedatFamiliesOfCR} to a fundamental domain for the action of $\pi_1(\CC\smm\{0\})$ corresponding to the path $l$. The second property now follows from the statement (6) in~\cite[Lemma~9]{DubedatFamiliesOfCR}.
\end{proof}

\begin{lemma}
  \label{lemma:existense_of_Ksinv}
  For any $s\in (0,1/2)$ and any white vertex $w_0$, there exists a unique function $K^{-1}_s(\cdot,w_0)$ satisfying $(K_sK_s^{-1})(w)=\delta_{w_0}(w)$ and $K^{-1}_s(b,w_0)=O(|b|^{s-1})$ for $|b|>2|w_0|$. Moreover, there is a constant $C_s>0$ such that for all $b,w$,
  \begin{equation}
    \label{eq:estimate_on_K_sinv}
    K_s^{-1}(b,w) \leq \frac{C_s}{|b-w|} \left( \left|\frac{b}{w}\right|^s\vee \left|\frac{w}{b}\right|^s\right).
  \end{equation}

\end{lemma}
\begin{proof}
    The existence and uniqueness are shown in~\cite[Lemma~13]{DubedatFamiliesOfCR}, along with the estimate \eqref{eq:estimate_on_K_sinv} when $|b|\geq 2|w|$ and $|b|\leq \frac{|w|}{2}$. The case when $\frac{|w|}{2} \leq |b|\leq 2|w|$ is not explicitly stated in that lemma, so let us briefly explain how it can be deduced from the other estimates. Fix a $w$ and consider the function $F(b) = K_s^{-1}(b,w)$ defined on the region $\Omega$ consisting of $b\in \Gamma$ such that $\frac{|w|}{2} \leq |b|\leq 2|w|$. Let $b_1,b_1'\in \Gamma$ be the two neighbors of $w$. Without loss of generality we can assume that $b_1 = w + 1$. Let $G_s$ be the function from Lemma~\ref{lemma:multivalued_Green_function}. Using Lemma~\ref{lemma:Laplacian_with_monodromy} it is easy to deduce that the function $F$ can be decomposed as
    \[
        F(b) = G_s(b,b_1) - G_s(b,b_1') + h(b)
    \]
    where $h$ is some function which is harmonic in $\Omega$. Note that
    \[
        \max_{b\in \partial \Omega} |h(b)|\leq (2C + C_s)|w|^{-1}
    \]
    where $C$ is the constant from Lemma~\ref{lemma:multivalued_Green_function}. Thus, using Lemma~\ref{lemma:maximum_principle_for_Delta_s} to estimate $h$ and Lemma~\ref{lemma:multivalued_Green_function} to estimate $G_s(b,b_1) - G_s(b,b_1')$ we get
    \[
        |F(b)|\leq C|b-w|^{-1} + (2C + C_s)|w|^{-1}
    \]
    for all $b\in \Omega$. The case when $b\in \Gamma^\dagger$ can be treated similarly.

\end{proof}

Note that we can set
\begin{equation}
  \label{eq:K-s}
  K_{-s}^{-1}(b,w) = \eta_b^2\eta_w^2\overline{K_s^{-1}(b,w)}
\end{equation}
which provides us with the kernel $K_s^{-1}$ for all $s\in (-1/2,1/2)$ (we can take $K^{-1}$ from Section~\ref{sec:Height_function_loop_statistics_and_monodromy} for $s=0$). Nevertheless, we would like to stress that the proof of~\cite[Lemma~13]{DubedatFamiliesOfCR} does not produce the estimate~\eqref{eq:estimate_on_K_sinv} uniform in $s\to 0$. We will show below how this uniformity follows from precompactness arguments.

The next two lemmas follow from the analysis of suitable discrete exponentials.

\begin{lemma}
  \label{lemma:holomorphic_functions}
  For any $s\in (-1/2,1/2)$ there exist unique discrete holomorphic functions $f_s, g_s$ both with monodromies $e^{2\pi i s}$ around the origin and having the following asymptotics:
  \begin{equation}
    \label{eq:asymp_of_fs}
    \begin{split}
      &f_s(b) = \eta_b^2\bar{b}^{-s} + O(sb^{|s|-1}),\\
      &g_s(b) = b^s + O(sb^{|s|-1}).
    \end{split}
  \end{equation}
  uniformly in $s\to 0+$. Moreover, when $s>0$ we have
  \begin{equation}
      \label{eq:asympt_of_fs_precise}
      f_s(b) = \eta_b^2\bar{b}^{-s} + 2^s\frac{e^{\pi i/4}}{\sqrt{2}}\frac{\Gamma(1-s)}{\Gamma(s)}b^{s-1} + O(sb^{|s|-2}).
  \end{equation}
  Finally, if $b\in\{b_0,b^\dagger_0\}$, then
  \begin{equation}
    \label{eq:fs_value_at_0}
    \begin{split}
      &f_s(b) = \Gamma(1-s)\eta_{b}^2\bar{b}^{-s},\\
      &g_s(b) = \Gamma(1+s)b^s.
    \end{split}
  \end{equation}
  Analogous statements are true if we consider discrete holomorphic functions on white vertices, that is, if we consider functions satisfying $\sum_{w\sim b}f(w)K_s(w,b) = 0$ for each $b$. Namely, functions $f_s$ and $g_s$ can be extended to white vertices such that the extension is still discrete holomorphic and~\eqref{eq:asymp_of_fs},~\eqref{eq:fs_value_at_0} hold if one just replace $b$ with $w$ everywhere.
\end{lemma}
\begin{proof}
    Assertions of the lemma are essentially a subset of assertions from~\cite[Lemma~10]{DubedatFamiliesOfCR}, but for comparing these two set of assertions we need to explain how the conventions in~\cite{DubedatFamiliesOfCR} correspond to those used in this paper. Namely,
    \begin{itemize}
        \item The formulas for $f_s,g_s$ are not scale invariant. To keep up with the conventions in~\cite{DubedatFamiliesOfCR} one has to rescale the square lattice by $\frac{1}{\sqrt{2}}$, that is, to make the diagonal of a square to be equal to one.
        \item In~\cite{DubedatFamiliesOfCR} several (related by a gauge transformation) notions of the Kasteleyn operator are used. The operator $\mathsf{K}$ used in~\cite[Lemma~10]{DubedatFamiliesOfCR} is related to the operator $K$ used by us as follows:
        \[
            \mathsf{K}((\sqrt{2})^{-1}w,(\sqrt{2})^{-1}b) = \frac{1}{\sqrt{2}} \eta_w K(w,b)\eta_b.
        \]
        In particular, if $f$ is a function on black vertices lying in the kernel of $\mathsf{K}$, then $\eta_b f(b)$ lies in the kernel of $K$.
        \item The choice of operator $\mathsf{K}$ in~\cite{DubedatFamiliesOfCR} requires introducing additional signs when restricting holomorphic functions to the lattice, implemented by the operators
            \[
            \begin{split}
                &(R_B\vphi)(b) = (\bar{R}_B\vphi)(b) = \vphi(b),\qquad b\in \Gamma,\\
                &(R_B\vphi)(b) = -(\bar{R}_B\vphi)(b) = i\vphi(b),\qquad b\in \Gamma^\dagger
            \end{split}
            \]
            (see~\cite[eq.~(3.13)]{DubedatFamiliesOfCR}). Note that
            \[
                \eta_b (R_B\vphi)(b) = \vphi(b),\qquad \eta_b (\bar{R}_B\vphi)(b) = \eta_b^2 \vphi(b).
            \]
        \item The constant $\tau$ used in~\cite[Lemma~10]{DubedatFamiliesOfCR} specializes as $\tau = e^{\pi i/4}$ in our case.
    \end{itemize}
    We are now in the position to define $f_s$ and $g_s$. In both cases we use functions $f_{k,\chi}$ from~\cite[Lemma~10]{DubedatFamiliesOfCR} properly rescaled. When $s>0$ we define
    \[
        f_s(b) \qquad \text{ to be}\qquad \eta_b(\sqrt{2})^s\Gamma(s)^{-1}e^{\pi i/4}f_\chi((\sqrt{2})^{-1}b) \qquad \text{with $\chi = e^{2\pi i s}, v_0 = 0$ }
    \]
    and
    \[
        g_s(b) \qquad \text{ to be}\qquad \eta_b(\sqrt{2})^{-s}e^{-\pi i/4}\Gamma(s)^{-1}\overline{f_{1,\chi}((\sqrt{2})^{-1}b)}\qquad \text{with $\chi = e^{2\pi i (1-s)}, v_0 = 0$ }.
    \]
    (Note that if $f$ lies in the kernel of $K$, then $\eta_b^2 \overline{f(b)}$ is in the kernel as well since $K(w,b) = \pm \bar{\eta}_w\bar{\eta}_b$.) When $s<0$ we can put
    \[
        f_s(b) = \eta_b^2 \overline{g_{-s}(b)},\qquad g_s(b) = \eta_b^2 \overline{f_{-s}(b)}.
    \]

    The asymptotic relations for $f_s$ and $g_s$ declared in the lemma now follow from asymptotic relations asserted in~\cite[Lemma~10]{DubedatFamiliesOfCR}. Note that the error terms in the latter asymptotic relations are uniform in $s$, as follows from inspection of Dub\'edat's proof .

    Finally, to extend $f_s$ and $g_s$ to white vertices we can just interchange the colors and use the same construction. We leave details to the reader.
\end{proof}

\begin{lemma}
  \label{lemma:Ksinv_when_b_is_at_0}
  The kernel $K_s^{-1}(b_0,w)$ has the following asymptotics when $s\in (-1/2,1/2)$:
  \begin{equation}
    \label{eq:Ksinv_when_b_is_at_0}
    K_s^{-1}(b_0,w) = \frac{\Gamma(1+s)}{2}\frac{1}{\pi(b_0-w)}\left( \frac{b_0}{w} \right)^s + \frac{\Gamma(1-s)}{2}\frac{(\eta_{b_0}\eta_w)^2}{\pi(\bar{b}_0-\bar{w})}\left( \frac{\bar{w}}{\bar{b}_0} \right)^s + O(w^{s-2})
  \end{equation}
  uniformly in $s\to 0+$. Moreover, we have
  \begin{equation}
    \label{eq:Ksinv_when_b_is_at_0_and_sto0}
    \lim\limits_{s\to 0+}K_s^{-1}(b_0,w) = K^{-1}(b_0,w)
  \end{equation}
  for any $w$. Similarly, if $w_0$ is incident to the origin, we have
  \begin{equation}
    \label{eq:Ksinv_when_w_is_at_0}
    K_s^{-1}(b,w_0) = \frac{\Gamma(1-s)}{2}\frac{1}{\pi(b-w_0)}\left( \frac{b}{w_0} \right)^s + \frac{\Gamma(1+s)}{2}\frac{(\eta_{b}\eta_{w_0})^2}{\pi(\bar{b}-\bar{w}_0)}\left( \frac{\bar{w}_0}{\bar{b}} \right)^s + O(b^{s-2})
  \end{equation}
  uniformly in $s\to 0+$, and we have for any $b$
  \begin{equation}
    \label{eq:Ksinv_when_w_is_at_0_and_sto0}
    \lim\limits_{s\to 0+}K_s^{-1}(b,w_0) = K^{-1}(b,w_0)
  \end{equation}
\end{lemma}
\begin{proof}

    The identities~\eqref{eq:Ksinv_when_w_is_at_0} and~\eqref{eq:Ksinv_when_w_is_at_0_and_sto0}, without the uniformity claim, are in~\cite[Lemma~14(2)]{DubedatFamiliesOfCR}. The proof there, via~\cite[Lemma~11]{DubedatFamiliesOfCR}, goes by explicitly constructing $K_s^{-1}(\cdot,w_0)$ using discrete exponentials. By writing out the details, it would be possible to see that the error terms are uniform as $s\to 0+$. Instead, we will use the alternative construction of $K_s^{-1}(b,w_0)$ sketched in the remark after~\cite[Lemma~14]{DubedatFamiliesOfCR}. We first assume that $w_0 = -\frac{1}{2} + \frac{i}{2}$ and $s>0$. We try to express $K_s^{-1}(\cdot,w_0)$ as
    \begin{equation}
        \label{eq:Ksinv_ansatz}
        K_s^{-1}(b,w_0) = \alpha_s \left(\widetilde{f}_s(b) - f_s(b)\right)
    \end{equation}
    where $\widetilde{f}_s$ be the version of the function $f_s$ but constructed when the puncture is shifted to $i$ and $\alpha_s$ is a constant which we determine below. Note that both $f_s$, $\tilde{f}_s$ are discrete holomorphic, except at $w_0$ where there's a mismatch of sheets for $\tilde{f}_s$. At $w_0$ we have
    \[
        K_s\left(\alpha_s \left(\widetilde{f}_s - f_s\right)\right)(w) = \begin{cases}0,& w\neq w_0 \\ \alpha_s(1 - e^{2\pi is}) \widetilde{f}_s(b_0),& w=w_0\end{cases}
    \]
    thus a choice $\alpha_s = \left((1 - e^{2\pi is})\widetilde{f}_s(b_0)\right)^{-1}$ yields $(K_sK_s^{-1})(w,w_0) = \delta_{w_0}(w)$. Let us now derive the asymptotics of the right-hand side of~\eqref{eq:Ksinv_ansatz} with this choice of $\alpha_s$. First, we need the asymtotics of $\widetilde{f}_s$. It can be constructed in the same way $f_s$ is, but to avoid additional computations we can just put
    \[
        \widetilde{f}_s(b)\coloneqq -e^{-\pi i s}\overline{f_s(-\bar{b} -i)}.
    \]
    The constant $-e^{-\pi i s}$ in front of the right-hand side is chosen so that the main term in the asymptotics of $\widetilde{f}_s(b)$ is $\eta_b^2\bar{b}^{-s}$. (Note that $\bar{\eta}_{-\bar{b} - i}^2 = -\eta_b^2$.) To verify that $\widetilde{f}_s(b)$ is discrete holomorphic with correct monodromy, let $\widetilde{K}_s(w,b) = \chi_{s,l+[0,i]}(w,b)K(w,b)$ be defined as $K_s(w,b)$ (see~\eqref{eq:def_of_Ks}) but with the cut $l+[0,i]$  starting from $i$ instead of $0$. We note that $\widetilde{f}_s$ is in the kernel of $\overline{K_s(-\bar{w}-i, -\bar{b}-i)}$ and we have
    \[
        \overline{K_s(-\bar{w}-i, -\bar{b}-i)} = \overline{\chi_{s,l}(-\bar{w} - i, -\bar{b}  -i)}(w-b) = -\chi_{s,l+[0,i]}(w,b)K(w,b) = -\widetilde{K}_s(w,b)
    \]
    so $\widetilde{f}_s$ is in the kernel of $\widetilde{K}_s$.

    Using the asymptotics of $f_s(b)$ and the fact that $\bar{\eta}_{-\bar{b} - i}^2 = -\eta_b^2$ we get
    \begin{equation}
        \label{eq:asymp_of_tilde_fs}
        \widetilde{f}_s(b) = \eta_b^2(\bar{b} + i)^{-s} + 2^s\frac{e^{-\pi i/4}}{\sqrt{2}}\frac{\Gamma(1-s)}{\Gamma(s)}(b-i)^{s-1} + O(sb^{s-2}).
    \end{equation}
    Furthermore, note that $-\bar{b}_0 - i = b_0^\dagger = \bar{w}_0$ and $\eta_{b_0^\dagger}^2 = -1$. Using~\eqref{eq:fs_value_at_0} we get
    \[
        \widetilde{f}_s(b_0) = -e^{\pi is}\overline{f_s(b_0^\dagger)} = e^{-\pi i s}\Gamma(1-s)\bar{w}_0^{-s}.
    \]
    Using the properties of the $\Gamma$ function, we conclude that
    \begin{equation}
        \label{eq:alphas_value}
        \alpha_s = \left((1 - e^{2\pi is})\widetilde{f}_s(b_0)\right)^{-1} = -\left( 2i\sin(\pi s) \Gamma(1-s)\bar{w}_0^{-s} \right)^{-1} = -\frac{\Gamma(s)}{2\pi i}\bar{w}_0^s.
    \end{equation}
    Substituting the asymptotics of $f_s$~\eqref{eq:asympt_of_fs_precise}, asymptotics of $\widetilde{f}_s$~\eqref{eq:asymp_of_tilde_fs} and the value of $\alpha_s$~\eqref{eq:alphas_value} into~\eqref{eq:Ksinv_ansatz} we get
    \begin{multline*}
        K_s^{-1}(b,w_0) = \alpha_s \left(\widetilde{f}_s(b) - f_s(b)\right) =\\
        = -\eta_b^2 \frac{\Gamma(s)}{2\pi i}\bar{w}_0^s\left((\bar{b} + i)^{-s} - \bar{b}^{-s}\right) - \frac{\Gamma(1-s)}{2\pi i}2^s\bar{w}_0^s\left(\frac{e^{-\pi i/4}}{\sqrt{2}}(b-i)^{s-1} - \frac{e^{\pi i/4}}{\sqrt{2}}b^{s-1}\right) + O(b^{s-2}).
    \end{multline*}
    Observing that $w_0^{-s} = 2^s\bar{w}_0^s$ we get the desired asymptotics for $K_s^{-1}(b,w_0)$.

    When $w_0 = \frac{1-i}{2}$ similar computations can be used, we omit the details. When $s<0$ we can put $K_s(w,b) = (\eta_w\eta_b)^2\overline{K_{-s}(w,b)}$ and use the previous result. Finally, the asymptotics of $K_s^{-1}(w,b_0)$ can be derived by just interchanging the roles of the sublattices.
\end{proof}

Our goal is to prove the following lemma:

\begin{lemma}
  \label{lemma:C_is_uniform}
  The inequality~\eqref{eq:estimate_on_K_sinv} in Lemma~\ref{lemma:existense_of_Ksinv} holds uniformly in $s\to 0+$, i.e., there exists a constant $C$ and $\eps>0$ such that $C_s\leq C$ for all $s\in(0,\eps)$.
\end{lemma}
\begin{rem}
    The methods of \cite{DubedatFamiliesOfCR} leading to the proof of Lemma \ref{lemma:existense_of_Ksinv} in fact give estimates that are uniform in $s$ away from zero; thus that lemma holds with a constant that does not depend on $s$ at all.
\end{rem}
\begin{proof}
  Let us put
  \[
    M_s = \sup_{|w|/2\leq |b|\leq 2|w|,\ |b-w|\geq |w|/4} |w||K^{-1}_s(b,w)|.
  \]
  According to Lemma~\ref{lemma:existense_of_Ksinv}, $M_s<\infty$ for any $s>0$. We first claim that there exist absolute constant $c>0$ such that
  \begin{equation}
  \label{eq: K_bound_all_b}
  K_s^{-1}(b,w)\leq \frac{cM_s}{|b-w|}\left( \left|\frac{b}{w}\right|^s\vee \left|\frac{w}{b}\right|^s\right),\qquad \text{for all } b,w.
  \end{equation}
First, note that by moving the cut if necessary, $K_s(\cdot,w)-K(\cdot,w)$ may be viewed as a (single-valued) discrete holomorphic function in the ball $\{b:|b-w|\leq|w|/4\}$, thus by maximum principle
  \begin{equation}
  \label{eq: neardiag_bound_by_M}
    |K^{-1}_s(b,w)|\leq \frac{M_s}{|w|}+\frac{c_1}{|b-w|}\leq \frac{4M_s+c_1}{|b-w|}, \quad |b-w|\leq|w|/4
  \end{equation}
  where $c_1>0$ is a constant depending on~\eqref{eq:K-1_asymp} only.

  To treat the region $|b|\leq |w|/2$,   let, as above, $b_0^\dagger=-\frac12-\frac{i}{2}\in \Gamma^\dagger$ denote the vertex incident to the origin. Note that by~\eqref{eq:Laplace_at_zero} the function
  \begin{equation}
      \label{eq:Ciu1}
      F_s(b) \coloneqq f_s(b_0^\dagger)K_s^{-1}(b,w) - K_s^{-1}(b_0^\dagger,w) f_s(b),\qquad b\in \Gamma\cap B(0,|w|/2),
  \end{equation}
  is $\Delta_s$-harmonic. When $b\in \partial(\Gamma\cap B(0,|w|/2))$, we get that $|f_s(b_0^\dagger)K_s^{-1}(b,w)|\leq c_2 M_s|w|^{-1}$ by definition of $M_s$ and~\eqref{eq:fs_value_at_0} and $K_s^{-1}(b_0^\dagger,w) f_s(b)\leq c_3 |w|^{-1}$ by  ~\eqref{eq:Ksinv_when_b_is_at_0} and~\eqref{eq:asymp_of_fs}. Thus, by maximal principle, $F_s(b)\leq (c_2M_s+c_3)|w|^{-1}$, and using~\eqref{eq:fs_value_at_0},~\eqref{eq:Ksinv_when_b_is_at_0} (the expressions for $f_s(b_0^\dagger)$ and $K_s^{-1}(b_0^\dagger, w)$), and~\eqref{eq:asymp_of_fs} (asymptotics of $f_s$) we conclude that
  \begin{equation}
  \label{eq: nearzero_bound_by_M}
    |K^{-1}_s(b,w)|\leq c_4M_s\left|\frac{w}{b}\right|^s\frac{1}{|w|}\leq 2c_4M_s\left|\frac{w}{b}\right|^s\frac{1}{|b-w|},\qquad |b|\leq |w|/2,\,b\in\Gamma.
  \end{equation}
  for all $s>0$ small enough, where $c_{2,3,4}>0$ are absolute constants.

  In a similar way, we get the same estimate for $b\in\Gamma^\dagger.$ Replacing $K^{-1}_s$ with $K^{-1}_{-s}$ (cf.~\eqref{eq:K-s}) and interchanging the roles of $b$ and $w$ we get the symmetric estimate when $b\geq 2|w|$. By combining that estimate with \eqref{eq: neardiag_bound_by_M} and \eqref{eq: nearzero_bound_by_M}, we get~\eqref{eq: K_bound_all_b}.

It remains to prove that $\limsup_{s\to 0+}M_s<\infty$.
Assume on the contrary that $M_s\to\infty$ (along a subsequence that we will henceforth not mention), and consider
  \[
    \widetilde{K}_s^{-1}(b,w) \coloneqq M_s^{-1} K_s^{-1}(b,w).
  \]
  For each $s$ let $b_s,w_s$ be such that $|w_s\widetilde{K}_s^{-1}(b_s,w_s)|\geq 1/2$ and $|w_s|/2\leq |b_s|\leq 2|w_s|,\ |b_s-w_s|\geq |w_s|/4$.

  We now consider two cases:

  \textit{Case 1: $w_s$ stays bounded as $s\to 0$.} The bound~\eqref{eq: K_bound_all_b} allow us to find a sequence $s_k\to 0+$ such that  $\widetilde{K}_{s_k}^{-1}(\cdot,w_{s_k})$ converges pointwise to a bounded function $\Phi(\cdot)$. Since $K_s \widetilde{K}_{s}^{-1}(w,w_s)=M_s^{-1}\delta_{w_s}(w)$ and by our assumption $M_s^{-1}\to 0$, we have $K\Phi\equiv 0$. Thus, $\Phi$ is a bounded discrete holomorphic function, i.e., it is constant on each of the lattices $\Gamma$ and $\Gamma^\dagger$. Denote these constants by $X$ and $X^\dagger$. We have $\max(|X| , |X^\dagger|)\geq 1/2$ by the construction of $\widetilde{K}_s$ and the choice of $w_s$. If we show that $X = X^\dagger = 0$, then we obtain a contradiction.

  To this end, let $\gamma$ be a simple counterclockwise loop composed of edges of the dual graph to $(\ZZ + 1/2)^2$ and encircling $0$ and $w_{s}$. Consider the following expression:
  \[
    Y_s = \sum_{(bw)^\ast\in \gamma} \widetilde{K}_s^{-1}(b,w_s)K^{-1}_{-s}(b_0,w)d(bw)^\ast.
  \]
  By a discrete version of Green's identity, one sees that $Y$ does not depend on the choice $\gamma$ provided $\gamma$ encircles all $w_{s_k}$ and the origin. Expanding $\gamma$ to infinity and using Lemmas~\ref{lemma:existense_of_Ksinv},~\ref{lemma:Ksinv_when_b_is_at_0} we conclude that $Y_{s_k} = 0$ for each $k$. On the other hand,
  \[
    \lim\limits_{k\to \infty} Y_{s_k} = \sum_{(bw)^\ast\in \gamma}\Phi(b)K^{-1}(b_0,w)d(bw)^\ast = \Phi(b_0) = X,
  \]
  where we used Lemma~\ref{lemma:Ksinv_when_b_is_at_0} to take the limit $\lim\limits_{s\to 0+} K^{-1}_{-s_k}(b_0,w) = K^{-1}(b_0,w)$.
  It follows that $X = 0$. Arguing in the same way we conclude that $X^\dagger = 0$ too.

  \textit{Case 2: $w_s$ is not bounded as $s\to 0+$.} Put $\delta_s = |w_s|^{-1}$; we may assume by passing to a subseqeunce that $\delta_s\to 0$ and that moreover $\delta_s w_s\to w$. Consider the function
  \[
    \Phi_s(b) = \delta^{-1}_s\widetilde{K}_s^{-1}(\delta^{-1}_s b, w_s).
  \]
  on the lattice $\delta_s(\ZZ + 1/2)^2$
  The bound~\eqref{eq: K_bound_all_b} ensures that $\Phi_s(\cdot)$ is uniformly bounded on compact subsets of $\CC\setminus\{0,w\}$. Standard precompactness theory for discrete holomorphic functions(see~\cite{CLR1} for a modern approach) imply that we can find a sequence $s_k\to 0+$ and two functions $\Phi, \Phi^\ast$ defined on $\CC\setminus\{0,w\}$ such that $\Phi$ is holomorphic and $\Phi^\ast$ is antiholomorphic, and  $\Phi_{s_k}(\cdot ) = \Phi(\cdot) + \eta_b^2 \Phi^\ast(\cdot) + o(1)$ as $k\to \infty$ uniformly on compact subsets of $\CC\smm\{0,w_0\}$.

  The bound~\eqref{eq: K_bound_all_b} and the first inequality in \eqref{eq: neardiag_bound_by_M} imply that $\Phi,\Phi^\ast$ are in fact bounded, hence they are constant. We also must have at least one of these constant to be non-zero by the definition of $M_s$ and $w_s$. Let us show that $\Phi(0) = \Phi^\ast(0) = 0$ and obtain a contradiction.

  Put $F_s(w) = \delta_s^{-1}K_{-s}^{-1}(b_1,\delta_s^{-1}w)$, where $b_1$ is a fixed black vertex. As above, we have
  \begin{equation}
    \label{eq:residue_is_zero}
    \sum_{(bw)^\ast\in \gamma_s} \Phi_s(b)F_s(w)d(bw)^\ast = 0
  \end{equation}
  for each $s$ and a counterclockwise simple loop $\gamma_s$ on the dual lattice to $\delta_s(\ZZ + 1/2)^2$ encircling $0,\delta_s w_s,\delta_s b_1$. Let us substitute $b_1 = b_0$, where $b_0\in \Gamma$ is incident to the face containing the origin and take $\gamma_s$ to approximate a given simple polygonal loop $\gamma$ encircling $w_0$ and the origin. By Lemma~\ref{lemma:Ksinv_when_b_is_at_0} we have $F_s(w)= -\frac{1}{2\pi w}-\frac{\eta^2_{b_0}\eta_w^2}{2\pi \bar w}+o(1)$ as $s\to 0$, uniformly near $\gamma$. This implies
  \begin{multline*}
    \sum_{(bw)^\ast\in \gamma_{s_k}} \Phi_{s_k}(b)F_{s_k}(w)d(bw)^\ast = \\
    -\sum_{(bw)^\ast}\left(\frac{\Phi}{2\pi w}+ \frac{\Phi^\star\eta^2_{b_0}(\eta_b\eta_w)^2}{2\pi \bar{w}}+\frac{\Phi\eta_{b_0}\eta^2_w}{2\pi \bar{w}}+\frac{\Phi^\ast\eta_b^2}{2\pi w}\right)d(bw)^\ast+o(1)=
    \\-i\int_\gamma \left(\Phi\,\frac{dz}{2\pi z} - \Phi^\ast \,\frac{\eta_{b_0}^2d\bar{z}}{2\pi \bar{z}}\right) + o(1),\qquad k\to +\infty,
  \end{multline*}
  where we have used \eqref{eq:dedge=int} and the observation that $\eta_b^2$ (respectively, $\eta_w^2$) are alternating $\pm 1$ along straight line segments of $\gamma$, leading to approximate telescoping for the last two terms in the second sum. We infer that $ \Phi(0) - \eta_{b_0}^2\Phi^\ast(0) = 0$, and similarly, setting $b_1 = b_0^\dagger$ we get $\Phi(0) - \eta^2_{b_0^\dagger} \Phi^\ast(0) = 0.$ Since $\eta_{b_0}^2 = -\eta_{b_0^\dagger}^2$ we conclude that $\Phi(0) = \Phi^\ast(0) = 0$, which is a desired contradiction.

\end{proof}

Using Lemma~\ref{lemma:existense_of_Ksinv} completed with Lemma~\ref{lemma:C_is_uniform} and following the approach of~\cite{DubedatFamiliesOfCR} we can obtain the following asymptotic for $K_s^{-1}$. Recall the notation~\eqref{eq:def_of_chi_extended}. We denote
\begin{align}
      \Kinv[s](b,w)&:=\frac{1}{2\pi (b-w)}\left( \frac{b}{w} \right)^s + \frac{\eta^2_w\eta^2_b}{2\pi (\bar{b} - \bar{w})} \left( \frac{\bar{w}}{\bar{b}} \right)^s, \\
      \Kinv[s,b\approx w](b,w)&:=\chi_{s,l}(w,b)^{-1}\left( K^{-1}(b,w) + \frac{s}{2\pi}\left[\frac{1}{w} - \frac{(\eta_b\eta_w)^2}{\bar{w}}\right]\right),\\
      \Kinv[s,b\approx 0](b,w)&:=-\frac{g_s(b)}{2\pi w^{1+s}} - \frac{\eta_w^2f_s(b)}{2\pi \bar{w}^{1-s}}.
\end{align}

\begin{lemma}
  \label{lemma:near-diagonal-expansion}
  The unique $K_s^{-1}$ given in Lemma \ref{lemma:existense_of_Ksinv} has the following asymptotics near the diagonal (respectively, for large $b$)
  \begin{equation}
    \label{eq:near-diagonal-expansion}
    K_s^{-1}(b,w) = \begin{cases}\Kinv[s,b\approx w](b,w) + O\left(\log|w||w|^{-5/4}\right),& |b-w|\leq |w|^{3/4}\\
    \Kinv[s](b,w)+O\left(|b|^{s-1}|w|^{-s-\frac12}\right),& |b|\geq 2|w|
    \end{cases}
  \end{equation}
  uniformly in $s\to 0+$.
\end{lemma}
\begin{proof}
  We will use a similar construction as in the proof of~\cite[Lemma~14]{DubedatFamiliesOfCR}, but with some small adjustments needed for our purpose. Hereinafter, by $A\lesssim B$ we mean $|A|\leq C\cdot |B|$ for an absolute constant $C$. Define the parametrix $S(b,w)$ as follows
 \[
    S(b,w) = \begin{cases}
      \Kinv[s](b,w), & |b|> |w|^\frac12,\ |b-w|> |w|^{\frac34},\\
      \Kinv[s,b\approx w](b,w),& |b-w|\leq |w|^{\frac34},\\
      \Kinv[s,b\approx 0](b,w),& |b|\leq |w|^\frac12.
    \end{cases}
  \]
  Let $f$ be a branch in $\mathbb{C}\setminus l$ of a holomorphic $e^{2\pi i s}$ multiplicative function. Then, since $Kz\equiv Kz^2\equiv 0$, Taylor expansion gives
\begin{equation}
  \label{eq: estimate_K_f}
  |K_s f|(u)\lesssim \sup_{|z-u|\leq 1}|\partial^3f(z)|
  \end{equation}
  where we omit the restriction to the lattice from the notation. Similarly, if $f$ is anti-holomorphic and $\hat{f}(b)=\eta^2_bf(b)$, we have
  \begin{equation}
  \label{eq: estimate_K_f_bar}
  |K_s \hat{f}|(u)\lesssim \sup_{|z-u|\leq 1}|\bar{\partial}^3f(z)|.
  \end{equation}
  Define $T = K_sS - \Id$.  We estimate $T(u,w)$ as follows:
  \begin{itemize}
  \item In the region $|u|> |w|^\frac12+1,\ |u-w|>|w|^{\frac34} + 1$, we apply (\ref{eq: estimate_K_f}, \ref{eq: estimate_K_f_bar}) to get
\begin{equation}
   \label{eq:estimate_on_T_one_puncture_1}
 T(u,w) \lesssim
      \left( \frac{1}{|u|^3} + \frac{1}{|u-w|^3} \right)\frac{1}{|u-w|} \left( \left|\frac{u}{w}\right|^s\vee \left|\frac{w}{u}\right|^s\right)
\end{equation}
\item In the discs $|u|<w^\frac12-1$ and $|u-w|<|w|^\frac34-1$, we note that $(K_sS)(u)=\delta_{u,w}$, so
\begin{equation}
   \label{eq:estimate_on_T_one_puncture_2}
 T(u,w) = 0
\end{equation}
\item In the annular strip $|u-w| \in [|w|^{3/4}-1,|w|^{3/4}+1]$, we use that $$|(K_sS)(u,w)|=|(K_s\Kinv[s,b\approx w])(u,w)+K_s(S-\Kinv[s,b\approx w]))(u,w)|\leq 4\sup_{b\sim u}|\Kinv[s](b,w)-\Kinv[s,b\approx w](b,w)|.$$ To estimate the RHS, assuming without loss of generality that $l$ does not pass through the region, we use the Taylor expansion $$\Kinv[s](b,w)=\frac{1}{2\pi (b-w)}+\frac{\eta^2_w\eta^2_b}{2\pi (\bar{b} - \bar{w})}+\frac{s}{2\pi}\left[\frac{1}{w} - \frac{(\eta_b\eta_w)^2}{\bar{w}}\right]+O\left(\frac{|b-w|}{|w|^2}\right)$$ and combining this with \eqref{eq:K-1_asymp} yields
\begin{equation}
   \label{eq:estimate_on_T_one_puncture_3}
 T(u,w) = O\left(\frac{|u-w|}{|w|^2}\right)+O\left(\frac{1}{|u-w|^2}\right)=O\left(|w|^{-\frac54}\right)
\end{equation}
\item for $|u|\in [|w|^\frac12-1,|w|^\frac12+1]$, we apply a similar argument to $\Kinv[s,b\approx w]$ and $\Kinv[s,b\approx 0]$, using the estimate
$$
\Kinv[s](b,w)=-\frac{b^s}{2\pi w^{s+1}}-\frac{\eta^2_w\eta^2_b\bar{w}^{s-1}}{2\pi \bar{b}^s}+O\left(\frac{|b|^{s+1}}{|w|^{s+2}}\right)+O\left(\frac{|w|^{s-2}}{|b|^{s-1}}\right).
$$
Combining with \eqref{eq:asymp_of_fs}, we get
\begin{equation}
   \label{eq:estimate_on_T_one_puncture_4}
 T(u,w) = O\left(\frac{|u|^{s+1}}{|w|^{s+2}}+\frac{|w|^{s-2}}{|u|^{s-1}}+\frac{|u|^{|s|-1}}{|w|^{s+1}}+\frac{|u|^{|s|-1}}{|w|^{1-s}}\right)=O\left(|w|^{-\frac{3}{2}(1-s)}\right).
\end{equation}
  \end{itemize}
Next, we estimate $K_s^{-1}T$ by writing, using Lemma~\ref{lemma:C_is_uniform},
$$
|(K^{-1}_sT)(b,w)|\leq\sum_u|K^{-1}_s(b,u)||T(u,w)|\lesssim\sum_u\frac{1}{|b-u|}\left( \left|\frac{b}{u}\right|^s\vee \left|\frac{u}{b}\right|^s\right)|T(u,w)|.
$$
We split the sum into three parts, $\mathcal{U}_1=\{|u|<\frac{2}{3}|b|\}$, $\mathcal{U}_2=\{\frac{2}{3}|b|\leq |u|<\frac{4}{3}|b|\}$ and $\mathcal{U}_3=\{\frac43|b|\leq|u|\},$ where the prefactor is of order $|b|^{s-1}|u|^{-s}$, $|b-u|^{-1}$ and $|u|^{s-1}|b|^{-s}$ respectively. We thus have:
\begin{multline}
   \label{eq:sum_over_U_123}
 |(K^{-1}_sT)(b,w)|\lesssim |b|^{s-1}\sum_{u\in \mathcal{U}_1}|u|^{-s} |T(u,w)|\\+ \sum_{u\in \mathcal{U}_2}|b-u|^{-1}|T(u,w)|+|b|^{-s}\sum_{u\in \mathcal{U}_3}|u|^{s-1}|T(u,w)|.
\end{multline}
For the last term, applying \eqref{eq:estimate_on_T_one_puncture_1}, we get the bound $\lesssim |b|^{-s}\sum_{u\in \mathcal{U}_3}|u|^{s-1}|u|^{-4+s}|w|^{-s}\lesssim |w|^{-s}|b|^{-3+s}.$ For the first and the second term, consider first the case $|b|>2|w|$. For $u\in\mathcal{U}_2$, we have $|T(u,w)|\lesssim |b|^{s-4}|w|^{-s},$  and the sum again evaluates to $\lesssim |b|^{s-3}|w|^{-s}.$ Finally, for the first term in \eqref{eq:sum_over_U_123}, applying \eqref{eq:estimate_on_T_one_puncture_3},\eqref{eq:estimate_on_T_one_puncture_4} on the corresponding annular regions yields $|b|^{s-1}|w|^{-\frac12-s}$ and $|b|^{s-1}|w|^{s-1}$, of which the former is larger for small $s$. Finally, summing \eqref{eq:estimate_on_T_one_puncture_1} over $\mathcal{U}_1$ yields the contribution $\lesssim|b|^{s-1}|w|^{-\frac{3}{2}}$ to the first term in \eqref{eq:sum_over_U_123}. Put together, we get
\begin{equation}
\label{eq: K_s-1T_bound}
|(K^{-1}_sT)(b,w)|\lesssim|b|^{s-1}(|w|^{s-1}\vee |w|^{-\frac12-s}).
\end{equation}
In the regime $|b-w|\leq |w|^{3/4},$ the estimate of the first term in \eqref{eq:sum_over_U_123} is identical, while the second term, summing \eqref{eq:estimate_on_T_one_puncture_3} on the corresponding annual region will yield, in the worst case when $|b-w|\approx |w|^{\frac34}$, the bound$\lesssim |w|^{-\frac54}\log|w|$. Other contributions are smaller.
 Now, we can consider $S-K_s^{-1}T$ and note that $K_s(S-K_s^{-1}T)=K_sS-T=\Id$. Also, because of \eqref{eq: K_s-1T_bound}, it satisfies, for a fixed $w$, the bound $(S-K_s^{-1}T)(b,w)=O(|b|^{s-1}).$ Since there are no non-trivial discrete holomorphic $e^{2\pi s}$ multivalued functions with this asymptotics (cf.~\cite[Lemma~10]{DubedatFamiliesOfCR}), we conclude that
  \begin{equation}
    \label{eq:Kinv_via_S_one_puncture}
    K_s^{-1} = S - K_s^{-1}T.
  \end{equation}
 The result now follows from the bounds on $|(K^{-1}_sT)(b,w)|$ and the definition of $S$.
\end{proof}

In what follows, it will be important to get a better approximation to $K_s^{-1}(b,w)$ when $b$ is large compared to $w$ which in its turn may be arbitrarily close the the origin:

\begin{cor}
  \label{cor:big_b_asymptotics_revisited}
  Let $K_s^{-1}$ be as in Lemma~\ref{lemma:existense_of_Ksinv}. The following asymptotical relation holds uniformly as $s\to 0+$:
  \begin{equation}
    \label{eq:big_b_asymptotics_revisited}
    K_s^{-1}(b,w) = \frac{1}{2\pi (b-w)}b^sg_{-s}(w) + \frac{\eta^2_b}{2\pi (\bar{b} - \bar{w})} \bar{b}^{-s}f_{-s}(w) + O\left(\frac{1}{|b|^{5/4-s/2}}\right),\qquad |b|\geq 2|w|.
  \end{equation}
\end{cor}
\begin{proof}
We first note that the relation holds true for $|w|\geq |b|^\frac12$. Indeed, using first \eqref{eq:asymp_of_fs} and then \eqref{eq:near-diagonal-expansion}, we can write
\[
\frac{1}{2\pi (b-w)}b^sg_{-s}(w) + \frac{\eta^2_b}{2\pi (\bar{b} - \bar{w})} \bar{b}^{-s}f_{-s}(w)=\Kinv(b,w)+O(|b|^{2|s|-2})=K_s^{-1}(b,w)+O(|b|^{2s-2}+|b|^{\frac{s}{2}-\frac{5}{4}}).
\]
On the other hand, when $|w|\leq |b|^\frac{1}{2}$, we have
\[
\frac{1}{2\pi (b-w)}b^sg_{-s}(w) + \frac{\eta^2_b}{2\pi (\bar{b} - \bar{w})} \bar{b}^{-s}f_{-s}(w)=\frac{g_{-s}(w)}{2\pi b^{1-s}} + \frac{\eta^2_bf_{-s}(w)}{2\pi \bar{b}^{1+s}}+O(b^{-\frac{3}{2}+2s}).
\]
Therefore, it is enough to show that $|F(w)| = O(|b|^{-5/4+s/2})$ for $|w|\leq |b|^\frac{1}{2}$, where the discrete holomorphic function $F$ is given by
  \[
    F(w) = K_s^{-1}(b,w) - \frac{g_{-s}(w)}{2\pi b^{1-s}} - \frac{\eta^2_bf_{-s}(w)}{2\pi \bar{b}^{1+s}}.
  \]

To this end, we apply the maximum principle \eqref{eq:maximum_principle} and note that $|F(w)| = O(|b|^{-5/4+s/2})$ for $|w|\in [|b|^\frac12-1,|b|^\frac12]$ by the above computations, while by \eqref{eq:Ksinv_when_w_is_at_0} and \eqref{eq:fs_value_at_0} we have
 \begin{multline*}
  K_s^{-1}(b,w_0) = \frac{\Gamma(1-s)w^{-s}_0}{2\pi b^{1-s}} + \Gamma(1+s)\eta^2_{w_0}\bar{w}_0^s\frac{\eta^2_b}{2\pi \bar{b}^{1+s}} +O(|b|^{-2+s})+ O(|b|^{s-2})\\
  = \frac{g_{-s}(w_0)}{2\pi b^{1-s}}+\frac{\eta^2_bf_{-s}(w_0)}{2\pi \bar{b}^{1+s}}+O(|b|^{-2+s}).
 \end{multline*}
 Therefore, $|F(w_0)|=O(|b|^{-2+s})$, and similarly $|F(w^\dagger_0)|=O(|b|^{-2+s})$, and applying \eqref{eq:maximum_principle} completes the proof.
\end{proof}

\begin{rem}
The last step of the above proof hinges on an identity between the coefficient in front of the leading terms of the asymptotics of $K_s^{-1}(\cdot,w_0)$ and the values of $g_{-s}$ and $f_{-s}$ at $w_0$. We remark here that this identity can be recovered by considering the sums $\sum_{(bw)^\ast\in \gamma}K_s^{-1}(b,w_0)g_{-s}(w)$ and  $\sum_{(bw)^\ast\in \gamma}K_s^{-1}(b,w_0)f_{-s}(w)$ as in the proof of Lemma \ref{lemma:C_is_uniform}.
\end{rem}

Given $s>0$, we define $K_s^{-1}(b,w;v)=K_s^{-1}(b-v,w-v)$ the kernel with the puncture at $v$ instead of $0$. We also set
\[
  K_{-s}^{-1}(b,w;v) = (\eta_b\eta_w)^2\overline{K_s^{-1}(b,w;v)}.
\]

\subsection{Full-plane operator with monodromy around two points}
\label{subsec:coupling function with a monodromy: two punctures}

We continue using the setup of the previous section with the following modification: we shift the grid back, so it is $\ZZ^2$ now, and we make two punctures, $v$ and $\bar{v}$, where $v\in \CC^+$.
Let $\rho: \pi_1(\widehat{\CC}\smm\{ v,\bar{v} \})\to \TT$ be the representation such that a small counterclockwise oriented loop around $v$ is sent to $e^{2\pi i s}$ by $\rho$. Let $l$ be a straight line (seen as a path on the dual lattice) connecting $v$ and $\bar{v}$ and oriented from $v$ to $\bar{v}$. We define $K_{s,-s}$ by~\eqref{eq:def_of_Ks} with this path instead of the infinite path from the origin. We want to construct and estimate $K_{s,-s}^{-1}$. We begin by introducing the expected continuous limit:
\begin{equation}
  \label{eq:def_of_S_C_rho}
  \Dd_{s,-s}(b,w) = \Dd_{s}(b,w;v,\bar{v})\coloneq \frac{1}{2\pi (b-w)} \left( \frac{b-v}{b-\bar{v}} \right)^s\left( \frac{w-\bar{v}}{w-v} \right)^s + \frac{(\eta_b\eta_w)^2}{2\pi (\bar{b} - \bar{w})} \left( \frac{\bar{b} - v}{\bar{b} - \bar{v}} \right)^s \left( \frac{\bar{w} - \bar{v}}{\bar{w} - v} \right)^s.
\end{equation}
In particular, we denote $\Kinv[0](b,w)=\frac{1}{2\pi (b-w)} + \frac{(\eta_b\eta_w)^2}{2\pi (\bar{b} - \bar{w})}.$

We proceed with some local asymptotic of $\Dd_{s,-s}(b,w)$; all the asymptotics below are unfiorm in $s$ for $s\in(0,1/10)$. When $b-w$ is small (say, $2|b-w|\leq \min(|w-v|,|w-\bar{v}|)$) we record

\begin{equation}
  \label{eq:S_C_b_close_to_w}
  \Dd_{s,-s}(b,w) = \Kinv[0](b,w)+\Cbw_{b\approx w}+\eta^2_b\cdot\Cbw^\star_{b\approx w}+ O\left( |b-w|\cdot \left( \frac{1}{|w-v|^2}+ \frac{1}{|w-\bar{v}|^2} \right) \right),
\end{equation}
where the coefficients $\Cbw_{b\approx w}=\Cbw_{b\approx w}(w,v,s)$ and $\Cbw^\star_{b\approx w}=\Cbw^\star_{b\approx w}(w,v,s)$ are given by
\[
\Cbw_{b\approx w}=\frac{si\Im v}{\pi} \left[ \frac{1}{(w-\Re v)^2 + (\Im v)^2}\right],\qquad \Cbw^\star_{b\approx w}= \frac{si\Im v}{\pi}\left[\frac{\eta_w^2}{(\bar{w} - \Re v)^2 + (\Im v)^2} \right].
\]
When $b$ is close to $v$, say, $4|b-v|\leq \min(|w-v|,|v-\bar{v}|)$), we have
\begin{multline}
  \label{eq:S_C_b_close_to_v}
  \Dd_{s,-s}(b,w) =\Cbw_{b\approx v}\cdot(b-v)^s+\Cbw^\star_{b\approx v}\cdot\eta_b^2\cdot(\bar{b}-\bar{v})^{-s} \\ +O\left(\frac{|b-v|^{1+|s|}(\Im v)^{|s|}|w-\bar{v}|^{|s|}}{|w-v|^{1-|s|}}\left(\frac{1}{|w-v|}+\frac{1}{\Im v}\right)\right),
\end{multline}
(assuming all marked points are at distance at least $\frac12$ from each other), where
\[
\Cbw_{b\approx v}=-\frac{1}{2\pi}\frac{(w-\bar{v})^s}{(v-\bar{v})^s(w-v)^{s+1}},\qquad
\Cbw^\star_{b\approx v}=-\frac{\eta^2_w}{2\pi}\frac{(\bar{v}-v)^s(\bar{w}-\bar{v})^{s-1}}{(\bar{w}-v)^{s}}.
\]
Note that the above estimate holds even when $w$ is close to $\bar{v}$. A similar estimate holds as $b$ is close to $\bar{v}$, with a substitution $s\mapsto -s$ and $v\leftrightarrow \bar{v}$.

We now construct an approximate inverse $S_{s,-s}$. The general idea is to put $S_{s,-s}(b,w)=\Kinv[s,-s](b,w)$ when the marked points $b,w,v,\bar{v}$ are away from each other, and whenever $b$ (respectively, $w$; both $b,w$) is close to one of the other marked points, use the best available fit out of our toolbox of functions that are discrete holomorphic in $b$ (respectively, in $w$; in both $b,w$). This amounts to replacing $\Dd_0(b,w;v)$, $(b-v)^s$, $(w-v)^s$,  $\eta_b^2(\bar{b}-\bar{v})^{-s}$, etc. in \eqref{eq:def_of_S_C_rho}, \eqref{eq:S_C_b_close_to_v}  by $K^{-1}(b,w)$, $g_s(b-v)$, $g_s(w-b)$, $f_s(b-v)$, etc., respectively. Recall the definition of $\chi_{s,l}(w,b)$ given in~\eqref{eq:def_of_chi_extended}.

We start with more notation. Denote
\[
\Cbw_{b\approx w\approx v}=\frac{s i}{4\pi\Im v},\qquad \Cbw^\star_{b\approx w\approx v}=\frac{s  i}{4\pi\Im v}.
\]
\[
\Cbw_{b\approx v,w\approx \bar{v}}=\frac{e^{-i\pi s}}{2\pi }(2\Im v)^{-2s-1},\qquad
\Cbw^\star_{b\approx v,w\approx \bar{v}}=-\frac{e^{\pi i s}}{2\pi}(2\Im v)^{2s-1}.
\] We put
\begin{align}
      \Kinv[b\approx v](b,w)&:=\Cbw_{b\approx v}\cdot g_s(b-v)+\Cbw^\star_{b\approx v}\cdot f_s(b-v);\\
      \Kinv[b\approx w](b,w)&:=\chi_{s,l}(w,b)^{-1}\left(K^{-1}(b,w)+\Cbw_{b\approx w}+\eta^2_b\Cbw^\star_{b\approx w}\right);\\
      \Kinv[b\approx v,w\approx \bar{v}](b,w)&:=\Cbw_{b\approx v,w\approx \bar{v}}\cdot g_{s}(w-\bar{v})g_s(b-v)+\Cbw^\star_{b\approx v,w\approx \bar{v}}\cdot f_{s}(w-\bar{v})f_s (b-v);\\
      \Kinv[b\approx w\approx v](b,w)&:=K_s^{-1}(b,w;v)
      +\Cbw_{b\approx w\approx v}\cdot g_{-s}(w-v)g_s(b-v)+\Cbw^\star_{b\approx w\approx v}\cdot f_{-s}(w-v)f_s (b-v).
\end{align}
Expressions for other regimes are obtained from the above using the symmetries
\[
\Kinv[s](b,w,v,\bar{v})=\Kinv[-s](b,w;\bar{v},v),\qquad \Kinv[s](b,w,v,\bar{v})=-\Kinv[-s](w,b,v,\bar{v}).
\]
Thus, we define $\Kinv[b\approx \bar{v}](b,w)$ and $\Kinv[b\approx w\approx \bar{v}](b,w)$ by substituting $s\mapsto -s$ and $v\leftrightarrow \bar{v}$ into the expressions for $\Kinv[b\approx v](b,w)$ and $\Kinv[b\approx w\approx v](b,w)$ respectively. Similarly, we define $\Kinv[w\approx v]$, $\Kinv[w\approx \bar{v}]$ and $\Kinv[b\approx \bar{v},w\approx v]$ by substituting $s\mapsto -s$, $b\leftrightarrow w$, and changing the sign in the expressions for $\Kinv[b\approx v]$, $\Kinv[b\approx \bar{v}]$ and $\Kinv[b\approx v,w\approx \bar{v}]$, respectively.

We are ready to define the parametrix $S_{s,-s}(b,w)$. Let $a>0$ be a small parameter (say, $a<\frac{1}{100}$) and define $r_1<r_2<r_3$ by
\begin{equation}
  \label{eq:def_of_r123}
  r_1 = |\Im v|^{1-10a},\quad r_2 = |\Im v|^{1-3a},\quad r_3 = |\Im v|^{1-a}.
\end{equation}
To define $S_{s,-s}(\cdot, w)$ we modify $\Dd_{s,-s}(\cdot, w)$ in accordance with the position of $w$.

\begin{defin}
If $|w-v|> r_2$ and $|w-\bar{v}|>r_2$, we put
\begin{equation*}
S_{s,-s}(b,w)=\begin{cases}\Kinv[b\approx v](b,w),& |b-v|\leq r_1;\\
\Kinv[b\approx \bar{v}](b,w),& |b-\bar{v}|\leq r_1;\\
\Kinv[b\approx w](b,w),& |b-w|\leq r_1;\\
\Kinv[s](b,w,v,\bar{v}),& \text{else}.\end{cases}
\end{equation*}
If $|w-v|\leq r_2$ or $|w-\bar{v}|\leq r_2$, we put respectively
\begin{equation*}
S_{s,-s}(b,w)=\begin{cases}
\Kinv[b\approx w\approx v](b,w),& |b-v|\leq r_3;\\
\Kinv[b\approx \bar{v},w\approx v](b,w),&  |b-\bar{v}|\leq r_3;\\
\Kinv[w\approx v](b,w),& \text{else,}\end{cases}\quad\text{or}\quad S_{s,-s}(b,w)=\begin{cases}
\Kinv[b\approx w\approx \bar{v}](b,w),& \quad |b-\bar{v}|\leq r_3;\\
\Kinv[b\approx v,w\approx \bar{v}](b,w),& \quad |b-v|\leq r_3;\\
\Kinv[w\approx \bar{v}](b,w),& \quad \text{else.}\end{cases}
\end{equation*}
\end{defin}

Depending on the position of $w$, define three annular rings $\Aring_{v}$, $\Aring_{\bar{v}}$, $\Aring_w$ as follows: if $|w-v|\wedge |w-\bar{v}|\geq r_2 $, put
  \[
\Aring_{z}=\{u:r_1-1\leq |u-z|\leq r_1+1\},\quad z\in\{v,\bar{v},w\},
  \]
while if $|w-v|<r_2$ (respectively,  $|w-\bar{v}|<r_2$), we put $\Aring_w=\Aring_v$ (respectively, $\Aring_w=\Aring_{\bar{v}}$) and
\[
\Aring_{z}=\{u:r_3-1\leq |u-z|\leq r_3+1\},\quad z\in\{v,\bar{v}\}.
\]
We denote by $\Aring$ the unbounded component of the complement to these three rings.

In the estimates below, by $\sss$ we mean any expression of the form $\alpha s+\beta a s$, with absolute constants $\alpha,\beta$. As we are interested in the regime $s\to 0$, such expressions should be thought of as negligible; eventually, we will pick $a$ small enough and then choose $s_0>0$ so that all $p(s)$ are smaller than $a/2$ for $|s|<s_0$. Thus, for example, we will write $r_i^s,$ for $i=1,2,3,$ as $(\Im v)^{\sss}$.
\begin{lemma}
For all $b,w$, we have a uniform estimate
\begin{equation}
\label{eq: unif_estimate_on_S}
|S_{s,-s}(b,w)|\lesssim |b-w|^{-1}(\Im v)^{\sss}.
\end{equation}
\end{lemma}
\begin{proof}
It is enough to prove the result for $S_{s,-s}(b,w)$ replaced with $\Kinv[s](b,w,v,\bar{v})$, as the two are uniformly comparable to each other. For the latter, note that $\left|\frac{b-v}{b-\bar{v}}\right|^s$ is uniformly bounded when $|b|>2\Im v$ and dominated by $(3\sqrt{2} \Im v)^s$ else, and similarly for $\left|\frac{\bar{b}-v}{\bar{b}-\bar{v}}\right|^s, \left|\frac{w-\bar{v}}{w-v}\right|^s, \left|\frac{\bar{w}
-\bar{v}}{\bar{w}-v}\right|^s$.
\end{proof}

\begin{lemma}
\label{lem:bounds_on_T}
Let $T = K_{s,-s} S_{s,-s} - \Id$. Then, we have
\begin{equation}
\label{eq:T_bound_w_far}
|T(u,w)|\lesssim \begin{cases}
\left(|u-w|^{-3} +|u-v|^{-3} + |u-\bar{v}|^{-3}\right)|u-w|^{-1}(\Im v)^{\sss},& u\in \Aring;\\
(\Im v)^{-1-a+\sss},&u\in \Aring_w\cup \Aring_v\cup \Aring_{\bar{v}};\end{cases}
\end{equation}
with the implied constant uniform over $w$ and over $s$ small enough. Also,  \[T(u,w)= 0,\qquad u\notin(\Aring\cup\Aring_v\cup\Aring_{\bar{v}}\cup\Aring_w).\]
\end{lemma}
\begin{proof}
The last claim is immediate from definitions; for the rest, we proceed as in the one-puncture case. In the case $u\in \Aring$, we have that $S_{s,-s}$ is equal to $\Kinv[s](b,w,v,\bar{v}),$ $\Kinv[w\approx v](b,w)$, or $\Kinv[w\approx \bar{v}](b,w)$. All three of them have the form $f(b,w)+\eta^2_b\tilde{f}(b,w)$, with $f$ holomorphic and $\tilde{f}$ anti-holomorphic; we apply \eqref{eq: estimate_K_f}, \eqref{eq: estimate_K_f_bar} respectively. This yields a bound of the form $|T(u,w)|\lesssim P|f|+\tilde{P}|\tilde{f}|$, where $P,\tilde{P}$ are homogeneous third degree polynomials in $|u-v|^{-1},|u-\bar{v}|^{-1}$, and possibly $|u-w|^{-1}$, with coefficients bounded as $s\to 0$. We conclude by noticing that $f,\tilde{f}$ satisfy the same bound as in \eqref{eq: unif_estimate_on_S}.

For $u\in \Aring_{v},\Aring_{\bar{v}},\Aring_{w},$ we use the bound $|T(u,w)|\leq 4 \sup_{b\sim u}|\Kinv[in](b,w)-\Kinv[out](b,w)|$, where $\Kinv[in](\cdot,w),\Kinv[out](\cdot,w)$ are the expressions for $S_{s,-s}$ inside and outside of the corresponding ring. The latter estimate involves two contributions: the error term in the continuous asymptotic expansions such as \eqref{eq:S_C_b_close_to_v}, \eqref{eq:S_C_b_close_to_w}, and ones coming from the error terms in the asymptotics \eqref{eq:K-1_asymp}, \eqref{eq:asymp_of_fs}, \eqref{eq:big_b_asymptotics_revisited} of discrete holomorphic functions. We proceed case by case:
\begin{itemize}
\item When $|w-v|\wedge |w-\bar{v}|\geq r_2$ and $u\in \Aring_w$, it follows from \eqref{eq:S_C_b_close_to_w} and \eqref{eq:K-1_asymp} that
\[
|T(u,w)|\lesssim r^{-2}_1+r_1r^{-2}_2\lesssim(\Im v)^{-1-4a}.
\]
\item When $|w-v|\wedge |w-\bar{v}|\geq r_2$ and $u\in \Aring_v$, we apply \eqref{eq:S_C_b_close_to_v} and \eqref{eq:asymp_of_fs}; observing that the $O(\cdot)$ term in \eqref{eq:S_C_b_close_to_v} is the largest when $|w-v|\approx r_2$, we get
\[
|T(u,w)|\lesssim r_2^{s-1}r_1^{s-1}+r^{1+s}_1r_2^{-2+s}(\Im v)^{2s}\lesssim (\Im v)^{-1-4a+\sss}.
\]
\item When $|w-\bar{v}|< r_2$ and $u\in \Aring_{v}$, we have by \eqref{eq:asymp_of_fs}
\begin{multline*}
\Kinv[in](b,w)=\Cbw_{b\approx v,w\approx \bar{v}}\cdot g_{s}(w-\bar{v})g_s(b-v)+\Cbw^\star_{b\approx v,w\approx \bar{w}}\cdot f_{s}(w-\bar{v})f_s (b-v)\\=\Cbw_{b\approx v,w\approx \bar{v}}\cdot g_{s}(w-\bar{v})(b-v)^s+\Cbw^\star_{b\approx v,w\approx \bar{w}}\cdot f_{s}(w-\bar{v})\eta^2_b(\bar{b}-\bar{v})^{-s}+O\left((\Im v)^{-1+\sss}r^{s-1}_3\right)
\end{multline*}
while by definition of $\Kinv[out]=\Kinv[w\approx \bar{v}]$ and Taylor expansion, we have
\begin{multline*}
\Kinv[out](b,w)=\frac{1}{2\pi}\frac{(b-v)^s}{(\bar{v}-v)^s(b-\bar{v})^{s+1}}g_s(w-\bar{v})+\frac{1}{2\pi}\frac{\eta^2_b(v-\bar{v})^s(\bar{b}-\bar{v})^{-s}}{(\bar{b}-v)^{1-s}}f_{s}(w-\bar{v})
\\= \Cbw_{b\approx v,w\approx \bar{v}}\cdot g_{s}(w-\bar{v})(b-v)^s+\Cbw^\star_{b\approx v,w\approx \bar{w}}\cdot f_{s}(w-\bar{v})\eta^2_b(\bar{b}-\bar{v})^{-s}+O\left(r_3(\Im v)^{-2+\sss}\right),
\end{multline*}
so that we conclude
\[
|T(u,w)|\lesssim r_3^{-1}(\Im v)^{-1+\sss}+r_3 (\Im v)^{-2+\sss}\lesssim(\Im v)^{-1-a+\sss}.
\]
\item When $|w-v|< r_2$ and $u\in \Aring_v$, we use that by \eqref{eq:big_b_asymptotics_revisited},
\begin{multline*}
\Kinv[in](b,v)=K_{s}(b,w;v)+\Cbw_{b\approx w\approx v}\cdot g_{-s}(w-v)g_s(b-v)+\Cbw^\star_{b\approx w\approx v}\cdot f_{-s}(w-v)f_s (b-v)\\=\left(\frac{1}{2\pi}(b-v)^{s-1}+\Cbw_{b\approx w\approx v}\cdot g_s(b-v)\right)g_{-s}(w-v)\\+\left(\frac{\eta_b^2}{2\pi}(\bar{b}-\bar{v})^{-s-1}+\Cbw^\star_{b\approx w\approx v}\cdot f_s (b-v)\right)f_{-s}(w-v)+O\left(r_2^{s+1}r_3^{s-2}+r^{-\frac{5}{4}+s}_3\right).
\end{multline*}
and, by definition of $\Kinv[out]=\Kinv[w\approx v]$ and Taylor approximation again,
\begin{multline*}
\Kinv[out](b,v)=\frac{1}{2\pi}\frac{(b-\bar{v})^{-s}}{(v-\bar{v})^{-s}(b-v)^{-s+1}}g_{-s}(w-v)+\frac{\eta_b^2}{2\pi}\frac{(\bar{v}-v)^{-s}(\bar{b}-\bar{v})^{-s-1}}{(\bar{b}-v)^{-s}}f_{-s}(w-v)
\\= \frac{1}{2\pi}\left((b-v)^{s-1}+\frac{si}{2\Im v}(b-v)^s\right)g_{-s}(w-v)\\ +\frac{\eta_b^2}{2\pi}\left((\bar{b}-\bar{v})^{-s-1}+\frac{si}{2\Im v}(\bar{b}-\bar{v})^{-s}\right)f_{-s}(w-v)+O\left(r^{s+1}_3r^s_2(\Im v)^{-2}\right).
\end{multline*}
By applying \eqref{eq:asymp_of_fs}, we conclude, when $a$ is small enough, \[|T(u,v)|\lesssim r_2^{s+1}r_3^{s-2}+r^{-\frac{5}{4}+s}_3+r^{s+1}_3r^s_2(\Im v)^{-2}+r_2^sr_3^{s-1}(\Im v)^{-1}\lesssim (\Im v)^{-1-a+\sss}\]
\item The remaining cases $|w-v|< r_2$ and $u\in \Aring_{\bar{v}}$; $|w-\bar{v}|< r_2$ and $u\in \Aring_{\bar{v}}$; $|w-v|\wedge |w-\bar{v}|\geq r_2$ and $u\in \Aring_{\bar{v}}$, are identical to those already considered up to substitutions $v\leftrightarrow \bar{v}$ and $s\mapsto -s$.
\end{itemize}
\end{proof}
\begin{lemma}
  \label{lemma:Krhoinv}
  There exist absolute constants $\cst>0$, $s_0>0$, $a>0$ and $\beta>0$, such that if $s\in(0,s_0)$ and $\Im v\geq \cst$, there exists a unique right inverse $K_{s,-s}^{-1}$ of $K_{s,-s}$ satisfying, for a fixed $w$,
\begin{equation}
\label{eq:K_bound_b_large}
|K_{s,-s}^{-1}(b,w)|\leq C(w,v)|b|^{-1},\quad |b|>8(\Im v\vee |w|).
\end{equation}
Moreover, $K_{s,-s}^{-1}$ is also a left inverse of $K_{s,-s}$, and  we have
\begin{equation}
K_{s,-s}^{-1}(b,w)=S_{s,-s}(b,w)+O\left((\Im v)^{-1-\beta}\right),
\label{eq:error_bound_everywhere}
\end{equation}
uniformly in $b,w,s$, and
\begin{equation}
\label{eq:error_bound_b_large}
K_{s,-s}^{-1}(b,w)=S_{s,-s}(b,w)+O\left((\Im v)^{-\beta}|b|^{-1}\right),
\end{equation}
uniformly in $|b|>8(\Im v\vee |w|)$ and in $s$.
\end{lemma}
\begin{proof}
  Define $T = K_{s,-s} S_{s,-s} - \Id$. We view $T$ as a linear operator acting on functions on the right, i.e. $f(\cdot)\mapsto \sum_uf(u)T(u,\cdot).$ Our goal is to construct $K_{s,-s}^{-1}$ as
  \begin{equation}
  \label{eq:expression for Kss}
  K_{s,-s}^{-1}=S_{s,-s}(\Id+T)^{-1}=S_{s,-s}-S_{s,-s}T(\Id+T)^{-1}.
  \end{equation}
  To this end, we observe that $\|T\|_{l^\infty\to l^\infty}\to 0$ as $\Im v\to\infty$; more precisely, we will prove that $\|T\|_{l^\infty\to l^\infty}\lesssim(\Im v)^{-\beta}$ for a suitable choice of the parameter $a$. We proceed to estimating $\|T\|_{l^\infty\to l^\infty}=\sup_w\sum_u|T(u,w)|$ using Lemma \ref{lem:bounds_on_T}. The contributions from $u\in\Aring_w,\Aring_{v},\Aring_{\bar{v}}$ are estimated by multiplying the bound in \eqref{eq:T_bound_w_far} by $r_1$ or $r_3$, after plugging in \eqref{eq:def_of_r123}, this yields a bound $ \sum_{u\in\Aring_z}|T(u,w)|\lesssim (\Im v)^{-2a+p(s)}$. For the contribution from $u\in \Aring$, we can rescale the arguments by $\delta=(\Im v)^{-1}$ and compare sums to integrals. This yields

\[
\sum_{u\in\Aring}|T(u,w)|\lesssim \delta^2(\Im v)^{\sss}\int_{\delta\Aring} \left(|u-\delta w|^{-3} +|u-\delta v|^{-3} + |u-\delta \bar{v}|^{-3}\right)|u-\delta w|^{-1}\,du.
\]
When $|w-v|\wedge|w-\bar{v}|\geq r_2$, the balls of radii $\delta r_1$ around $\delta v,\delta \bar{v},\delta w$ are excluded from $\delta \Aring$, and using $|\delta v-\delta w|\geq \delta r_2$ and $|\delta \bar{v}-\delta w|\geq \delta r_2$, we get \[\lesssim\delta^2(\Im v)^{\sss}\left((\delta r_1)^{-2}+(\delta r_2)^{-1}(\delta r_1)^{-1}\right)\lesssim(\Im v)^{-2+20a+\sss}.\]  Similarly, if $|w-v|<r_2$ or $|w-\bar{v}|<r_2$, we get the bound $\lesssim \delta^2(\Im v)^{\sss}(\delta r_3)^{-2}=(\Im v)^{-2+2a+\sss}$.
 Therefore, for $s$ small enough, we can indeed choose $a$ such that all the exponents are smaller than $-\beta<0.$

  Note that by \eqref{eq: unif_estimate_on_S}, for any $b$, we have $\|S_{s,-s}(b,\cdot)\|_{l^\infty}<\infty.$ Therefore, once $\|T\|_{l^\infty\to l^\infty}<1$, the right-hand side of \eqref{eq:expression for Kss} makes sense. Moreover, since for each $b$, $K_{s,-s}(b,w)=0$ for all but four $w$, we can apply associativity to get \[K_{s,-s}  K_{s,-s}^{-1}=K_{s,-s}(S_{s,-s}(\Id+T)^{-1})=(K_{s,-s}S_{s,-s})(\Id+T)^{-1}=(\Id+T)(\Id+T)^{-1}=\Id,\]
so that $K_{s,-s}^{-1}$ is indeed a right inverse of $K_{s,-s}$.

We proceed by estimating $(S_{s,-s}T)(b,w)=\sum_u S_{s,-s}(b,u)T(u,w).$
Using \eqref{eq:T_bound_w_far} and then \eqref{eq: unif_estimate_on_S}, we get, for $z=w,v,\bar{v}$ and all $b,w$,
\begin{equation}
\sum_{u\in \Aring_z} |S_{s,-s}(b,u)T(u,w)|\lesssim (\Im v)^{-1-a+\sss}\sum_{u\in \Aring_z}|S_{s,-s}(b,u)|\lesssim \log (\Im v) (\Im v)^{-1-a+\sss}. \\
\end{equation}
To estimate the contribution from $u\in\Aring$, we can simply write by \eqref{eq: unif_estimate_on_S} and the above bound,
\[
\sum_{u\in\Aring}|S_{s,-s}(b,u)T(u,w)|\leq \sup_{u}|S_{s,-s}(b,u)|\sum_{u\in\Aring}|T(u,w)|\lesssim (\Im v)^{-2+20a+\sss}.
\]
Putting everything together, we get by \eqref{eq:expression for Kss} for all $b, w$,
\[
|K_{s,-s}^{-1}(b,w)-S_{s,-s}(b,w)|\leq \|(S_{s,-s}T)(b,\cdot)\|_{l^\infty}\|(1+T)^{-1}\|_{l^\infty \to l^\infty}\leq (\Im v)^{-1-\beta}.
\]
for $\beta=a/2$, provided that $a$ is small enough and $s$ is so small that all $p(s)$ are smaller than $a/4$. This completes the proof of \eqref{eq:error_bound_everywhere}.

Let us improve the above estimate for $|b|>8(|w|\vee \Im v)$. We then have
\begin{multline}
|(S_{s,-s}T^k)(b,w)|\leq \left|\sum_{u:|u|\leq \frac{|b|}{2}}S_{s,s}(b,u)T^k(u,w)\right|+\left|\sum_{u:|u|> \frac{|b|}{2}}S_{s,s}(b,u)T^k(u,w)\right|\\
\lesssim |b|^{-1}\|T\|^k_{l^\infty \to l^\infty}+\sum_{u:|u|>\frac{|b|}{2}}|T^k(u,w)|.
\end{multline}
Write $T^k(u,w)=\sum_{u_1,\dots,u_k}T(u_0,u_1)T(u_1,u_2)\dots T(u_k,u_{k+1}),$ with the convention $u_0=u$ and $u_{k+1}=w$. Note that for each tuple $u=u_0,u_1,\dots ,u_{k+1}=w$, there is an index $i$ such that $|u_i-u_{i+1}|>|b|/(4(k+1))$ and $|u_i|>2\Im v$; denote the smallest such index by $i_\text{min}.$ Then, by \eqref{eq:T_bound_w_far}, we have  \[|T(u_{i_\text{min}},u_{i_\text{min}+1})|\lesssim |u_{i_\text{min}}-u_{i_\text{min}+1}|^{-4}(\Im v)^{p(s)}.\]  Breaking the sum depending on $i_{\text{min}}$, and plugging in the above estimate, we obtain
\[
\sum_{u>\frac{|b|}{2}}|T^k(u,w)|\lesssim k(\Im v)^{p(s)}\|T\|^{k-1}_{l^\infty \to l^\infty}\left(\sum_{|u|>|b/(4(k+1))|}|u|^{-4}\right)\lesssim k^3(\Im v)^{p(s)}\|T\|^{k-1}_{l^\infty \to l^\infty}|b|^{-2}.
\]
Hence
\[
|(S_{s,-s}T^k)(b,w)|\leq |b|^{-1}\|T\|^{k-1}_{l^\infty \to l^\infty}(\|T\|_{l^\infty \to l^\infty}+k^3(\Im v)^{p(s)}|b|^{-1})
\]
Since $K_{s,-s}^{-1}(b,w)-S_{s,-s}(b,w)=\sum^\infty_{k=1}(-1)^k(S_{s,-s}T^k)(b,w)$ and $\|T\|^{k-1}_{l^\infty \to l^\infty}\lesssim (\Im v)^{-\beta}$, summing the above bounds yields
\[
|K_{s,-s}^{-1}(b,w)-S_{s,-s}(b,w)|\lesssim |b|^{-1}((\Im v)^{-\beta}+(\Im v)^{p(s)} |b|^{-1})\lesssim |b|^{-1}(\Im v)^{-\beta}.
\]
This completes the proof of \eqref{eq:error_bound_b_large}, and \eqref{eq:K_bound_b_large} follows from that and the asymptotics of $S_{s,-s}$.

To prove the uniqueness and the ``left inverse" claim, we construct separately a left inverse, denoted by $\widetilde{K}_{\rho}^{-1}$, following the exact same procedure as above, but reversing the roles of $b$ and $w$. To show that these two operators coincide, fix black and white $b_0,w_0$ and let $l$ be a large counterclockwise oriented simple loop on the dual lattice encircling these two points. Consider the quantity
  \[
    X = \sum_{(bw)^*\in l} K_{s,-s}^{-1}(b,w_0)\widetilde{K}_{s,-s}^{-1}(b_0,w) d(bw)^\ast.
  \]
  Deforming the contour inside we find out that
  \[
    X = \widetilde{K}_{s,-s}^{-1}(b_0,w_0) - K_{s,-s}^{-1}(b_0,w_0).
  \]
  On the other hand, deforming the contour outside all the way to the infinity and using the asymptotics~\eqref{eq:K_bound_b_large} for both $K_{s,-s}^{-1}$ and $\widetilde{K}_{s,-s}^{-1}$ we conclude that $X = 0$. It follows that $K_{s,-s}^{-1} = \widetilde{K}_{s,-s}^{-1}$. Since the argument only used that $K_{s,-s}^{-1}$ is a right inverse of $K_{s,-s}$ satisfying \eqref{eq:K_bound_b_large}, it also establishes the uniqueness claim.
\end{proof}
What we will use is the following corollary, which will be actually applied to neighboring $b,w$.
\begin{cor}
  \label{cor:nea-diag_estimate_on_Krhoinv}
  Let $K_{s,-s}^{-1}$ be as in Lemma~\ref{lemma:Krhoinv}, $\Im w>0$, and $|b-w|\leq (|w-v|\wedge \Im v)^\frac34$. Then, provided that $a$ is chosen small enough, we have
  \begin{multline}
    K_{s,-s}^{-1}(b,w) = \chi_{s,l}(w,b)^{-1} \left( K^{-1}(b,w) + \frac{s}{2\pi} \left[ \frac{1}{w-v} - \frac{(\eta_b\eta_w)^2}{\bar{w}-\bar{v}} \right] \right) \\+O\left(\frac{\log|w-v|+1}{|w-v|^{5/4}} + (\Im v)^{-1} \right)
  \end{multline}
\end{cor}
\begin{proof}
  We have by Lemma~\ref{lemma:Krhoinv}, $K_{s,-s}^{-1}(b,w) = S_{s,-s}(b,w)+O((\Im v)^{-1-\beta}).$ When $|w-v|>r_2,$ we will have $|b-w|<r_1$ if $a$ is chosen small enough, so, by definition, \[S_{s,-s}(b,w)=\Kinv[b\approx w]=\chi_{s,l}(w,b)^{-1}\left(K^{-1}(b,w)+\Cbw_{b\approx w}+\eta^2_b\Cbw^\star_{b\approx w}\right),\] and we conclude by noticing that when $\Im w>0$ we have
  \[
  \Cbw_{b\approx w}+\eta^2_b\Cbw^\star_{b\approx w}=\frac{s}{2\pi} \left[ \frac{1}{w-v} - \frac{(\eta_b\eta_w)^2}{\bar{w}-\bar{v}} \right]+O\left(|w-v|^{-2}\right)=\frac{s}{2\pi} \left[ \frac{1}{w-v} - \frac{(\eta_b\eta_w)^2}{\bar{w}-\bar{v}} \right]+O\left((\Im v)^{-2+6a}\right).
  \]
  When $|w-v|<r_2$, we will have $|b-v|<r_3$, so
  \begin{multline*}
  S_{s,-s}(b,w)=\Kinv[b\approx w\approx v]=K_s^{-1}(b,w;v)+\Cbw_{b\approx w\approx v}\cdot g_{-s}(w-v)g_s(b-v)+\Cbw^\star_{b\approx w\approx v}\cdot f_{-s}(w-v)f_s (b-v)\\
  =K_s^{-1}(b,w;v)+O\left((\Im v)^{-1}\left(|w-v|^s|b-v|^{-s}+|w-v|^{-s}|b-v|^{s}\right)\right)
  \\=\chi_{s,l}(w,b)^{-1} \left( K^{-1}(b,w) + \frac{s}{2\pi} \left[ \frac{1}{w-v} - \frac{(\eta_b\eta_w)^2}{\bar{w}-\bar{v}} \right] \right)+O\left((\Im v)^{-1}+\frac{\log|w-v|+1}{|w-v|^\frac54}\right),
  \end{multline*}
  where we used the bound $|b-w|\leq |w-v|^{3/4}$, and, in the last inequality, Lemma~\ref{lemma:near-diagonal-expansion}.
\end{proof}

\subsection{Proof of Theorem \ref{thma:main1}}
\label{subsec:Proof_main1}

Given $s\in \RR$ we define the representation $\rho_s : \pi_1(\Cpd\smm\{v\})\to \SL(2,\CC)$ by declaring
\begin{equation}
\rho_s(\gamma) = \begin{pmatrix}
    e^{2\pi i s} & 0 \\ 0 & e^{-2\pi i s}
\end{pmatrix}
\end{equation}
for a simple loop $\gamma$ encircling $v$ in the counterclockwise direction. Recall the definition of the operator $K_{\Cpd, \rho_s}$ given by~\eqref{eq:def_of_KOmega_delta_rho}. Note that since $\rho_s$ is diagonal, we have $K_{\Cpd,\rho_s}=K_{\Cpd,s}\oplus K_{\Cpd,-s}$, where $K_{\Cpd,s}$ is defined as in \eqref{eq:def_of_Ks} but restricted to $\Cpd$. The inverse $K^{-1}_{\Cpd,s}$ of $K_{\Cpd,s}$ can be constructed as
\begin{equation}
      \label{eq:KHsinv_via_reflection}
      K_{\Cpd, s}^{-1}(b,w) = \delta^{-1}K_{s,-s}^{-1}(\delta^{-1}b,\delta^{-1}w) - \eta_b^2\delta^{-1} K_{s,-s}^{-1}(\delta^{-1}\bar{b},\delta^{-1}w),
\end{equation}
    where $K_{s,-s}^{-1}$ be the inverse operator constructed by Lemma~\ref{lemma:Krhoinv} with $v$ scaled to $\delta^{-1}v$. It is clear that this is a right inverse; an argument similar to the end of the proof of Lemma \eqref{lemma:Krhoinv} shows that it is also a left inverse, i.e., $K_{\Cpd, s}^{-1}K_{\Cpd, s} f=f$ for any $f$ compactly supported in $\Cpd$.

Using Lemma~\ref{lemma:det_Krho} we obtain
\begin{equation}
    \label{eq:EcosN_via_monodromy}
    \EE \left(\cos (2\pi s)\right)^{\NCd(v)}=\det(K_{\Cpd, \rho_s} K_{\Cpd, \Id_{2\times 2}}^{-1}) =\det(K_{\Cpd, s} K_{\Cpd}^{-1})\det(K_{\Cpd, -s} K_{\Cpd}^{-1}).
\end{equation}
We will use the following variational identity (cf.~\cite[Lemma~3]{DubedatFamiliesOfCR}):
\begin{lemma}
We have, for every $s$ small enough and $\delta$ small enough,
\begin{equation}
    \label{eq: log_det_variational}
        \frac{d}{ds}\log\det(K_{\Cpd, s}K_{\Cpd}^{-1}) = \Tr\left[ \left(\frac{d}{ds}K_{\Cpd, s}\right)K_{\Cpd, s}^{-1}\right].
    \end{equation}
\end{lemma}
\begin{proof}
    We start with the following identity of matrices
    \begin{equation}
    \label{eq:K's associativity}
    K_{\Cpd}^{-1}\left(K_{\Cpd}K_{\Cpd, s}^{-1}\right)=K_{\Cpd, s}^{-1},
    \end{equation}
    which needs justification since the associativity does not hold in general for infinite matrices. Recall that $K_{\Cpd}^{-1}$ is a left inverse of $K_{\Cpd}$, hence we have $K_{\Cpd}^{-1}(K_{\Cpd} f)=f$ for every $f$ with finite support. Observe also that $(K_{\Cpd}K_{\Cpd, s}^{-1})(u,w)=0$ unless $u=w$ or $u$ is adjacent to the cut $l$.  Fixing $w$ and taking $f_R(b)=K_{\Cpd, s}^{-1}(b,w)\mathbb{I}_{b\leq R}$ with $R$ larger than $2|w|,2\Im v$, this implies, when $2|b|<R,$
    \begin{multline}
    K_{\Cpd, s}^{-1}(b,w)= \left(K_{\Cpd}^{-1}\left(K_{\Cpd}f_R\right)\right)(b,w)\\
    = \left(K_{\Cpd}^{-1}\left(K_{\Cpd}K_{\Cpd, s}^{-1}\right)\right)(b,w)+\sum_{|u|\in [R-1,R+1]}K_{\Cpd}^{-1}(b,u) (K_{\Cpd}f_R)(u).
    \end{multline}
    By \eqref{eq:K_bound_b_large} and \eqref{eq:K-1_asymp}, we have  $|(K_{\Cpd}f_R)(u)|=O(R^{-1})$ and $|K_{\Cpd}^{-1}(b,u)|=O(R^{-1})$ when $|u|\in [R-1,R+1]$ and $b$ is fixed, hence the last sum vanishes as $R\to\infty$, and \eqref{eq:K's associativity} follows.

    Since for a given $w$, $K_{\Cpd, s}(w,b)$ is non-zero for only four $b$, we deduce that
    $$
    (K_{\Cpd, s}K_{\Cpd}^{-1})(K_{\Cpd}K_{\Cpd, s}^{-1})=K_{\Cpd, s}\left(K_{\Cpd}^{-1}\left(K_{\Cpd}K_{\Cpd, s}^{-1}\right)\right)=\Id.$$
    Using the variational formula \cite[eq. 1.14]{gohberg1978introduction} for $\log \det$, the fact that for a given $u$, $\frac{d}{ds}K_{\Cpd, s}(u,b)$ is non-zero for an most one $b$,  and finally \eqref{eq:K's associativity}, we obtain
    \begin{multline}
    \frac{d}{ds}\log\det(K_{\Cpd, s}K_{\Cpd}^{-1})=\Tr\left[ \left(\frac{d}{ds}\left(K_{\Cpd, s}K_{\Cpd}^{-1}\right)\right)\left(K_{\Cpd}K_{\Cpd, s}^{-1}\right)\right]
    \\=\Tr\left[ \frac{d}{ds} K_{\Cpd, s}\left(K_{\Cpd}^{-1}\left(K_{\Cpd}K_{\Cpd, s}^{-1}\right)\right)\right]=\Tr\left[ \frac{d}{ds} K_{\Cpd, s}K_{\Cpd, s}^{-1}\right].
    \end{multline}
(Recall that here $\det$ and $\Tr$ are computed for restrictions to a finite-dimensional invariant subspace, as in Remark~\ref{ref: infinite_determinants}; note that taking the inverse commutes with restricting to that subspace.)
\end{proof}

We have the following:
\begin{prop}
  \label{prop: log_der_Laplace}
  Let $X_\delta=\frac{\NCd(v)-\mu_\delta}{\sigma_\delta}.$ Then, we have
  \[
    \frac{d}{d\lambda} \log \mathbb{E}e^{-\lambda X_\delta} = \lambda+O\left(\lambda^{-\frac12}|\log \delta|^{-\frac14}\right)
  \]
  uniformly over $\lambda\in(0,R]$ for any $R>0$.
\end{prop}
\begin{proof}
We have
\[
\EE e^{-\lambda X_\delta}=e^{\lambda\frac{\mu_\delta}{\sigma_\delta}}\EE\left(\cos2\pi s_\delta\right)^{N_\delta(v)},
\]
where $s_\delta= s_\delta(\lambda)=(2\pi)^{-1}\arccos e^{-\frac{\lambda}{\sigma_\delta}}.$
Using \eqref{eq:EcosN_via_monodromy} and then  \eqref{eq: log_det_variational}, we get
\begin{multline}
\label{eq: pa_lambda_log_lap}
    \frac{d}{d\lambda} \log \EE e^{-\lambda X_\delta}
    \\
    = \frac{\mu_\delta}{\sigma_\delta} + \Tr\left[ \left(\frac{d}{d\lambda}K_{\Cpd,s_\delta(\lambda)}\right)K_{\Cpd,s_\delta(\lambda)}^{-1} \right]+\Tr\left[ \left(\frac{d}{d\lambda}K_{\Cpd,-s_\delta(\lambda)}\right)K_{\Cpd,-s_\delta(\lambda)}^{-1} \right] =\\
    = \frac{\mu_\delta}{\sigma_\delta} + s_\delta'(\lambda)\sum_{b,w} \left(\left(\frac{d}{ds}K_{\Cpd,s_\delta}(w,b)\right)K^{-1}_{\Cpd,s_\delta}(b,w)+\left(\frac{d}{ds}K_{\Cpd,-s_\delta}(w,b)\right)K^{-1}_{\Cpd,-s_\delta}(b,w)\right).
    \end{multline}
Using the definition \eqref{eq:def_of_Ks} and \eqref{eq:def_of_dedge}, we have
\[
\frac{d}{ds}K_{\Cpd,s}(w,b)=\begin{cases}\pm 2\pi i s e^{\pm 2\pi i s}K_{\Cpd}(w,b),&(bw)^*\in l \\0&\text{else.}\end{cases}=\begin{cases} 2\pi i s \chi_{s,l}(w,b)d(bw)^\star,&(bw)^*\in l \\0&\text{else.}\end{cases}
\]
where the sign $\pm$ is positive iff $b$ is on the left of $l$. In particular, the sum in \eqref{eq: pa_lambda_log_lap} is a finite sum. Recalling \eqref{eq:KHsinv_via_reflection}, we can therefore write
\begin{multline}
\label{eq: dlogExpectation_interim}
 \frac{d}{d\lambda} \log \EE e^{-\lambda X_\delta} = \frac{\mu_\delta}{\sigma_\delta}
    \\
    +2\pi i  s_\delta'\sum_{(bw)^\star\in l} \left(\chi_{s,l}(w,b)\delta^{-1}K^{-1}_{s,-s}(\delta^{-1}b,\delta^{-1}w)-\chi_{-s,l}(w,b)\delta^{-1}K^{-1}_{-s,s}(\delta^{-1}b,\delta^{-1}w)\right) d(bw)^*
\\
   +2\pi i  s_\delta'\sum_{(bw)^\star\in l} \eta^2_b\left(\chi_{-s,l}(w,b)\delta^{-1}K^{-1}_{-s,s}(\delta^{-1}\bar{b},\delta^{-1}w)-\chi_{s,l}(w,b)\delta^{-1}K^{-1}_{s,-s}(\delta^{-1}\bar{b},\delta^{-1}w)\right) d(bw)^*.
\end{multline}
 Denote the sums in the above expression by $\Sigma_1$ and $\Sigma_2$ and estimate them separately.

To estimate $\Sigma_1$, we use the approximation for $K_{s,-s}^{-1}$ and $K_{-s,s}^{-1}$ provided by Corollary \ref{cor:nea-diag_estimate_on_Krhoinv}, noting that the error term sums up to $O(1)$, and we have the cancellation $\chi_{s,l}(w,b)\chi_{s,l}(w,b)^{-1}  K^{-1}(b,w)-\chi_{-s,l}(w,b)\chi_{-s,l}(w,b)^{-1}  K^{-1}(b,w)=0.$ Therefore, taking into account $(\eta_b\eta_w)^2=1$ for $(bw)^\star\in l,$
\begin{multline*}
\Sigma_1=\frac{s_\delta}{\pi}\sum_{(bw)^\star\in l}\left(\frac{1}{w-v}-\frac{(\eta_b\eta_w)^2}{\bar{w}-\bar{v}}\right)d(bw)^\star+O(1)\\
=\frac{2is_\delta}{\pi}\sum_{(bw)^\star\in l}\frac{(bw)^\star}{\Im v-\Im w}+O(1)=\frac{2is_\delta}{\pi}\log \delta^{-1}+O(1).
\end{multline*}

To show that $\Sigma_2=O(1)$, we split $
l=\{(bw)^\star \in l: \Im w> \Im v/4\}\cup \{(bw)^\star \in l: \Im w\leq \Im v/4\}=:l_1\cup l_2.$ On $l_1$, we use Lemma \ref{lemma:Krhoinv} to replace $K_{s,-s}$ and $K_{-s,s}$ with $S_{s, -s}$ and $S_{-s,s}$; the error term in that lemma sums to $O(1)$, the sum behaves as a Riemann sum for $b\mapsto \delta^{-1}\Kinv[\pm s](\delta^{-1}\bar{b},\delta^{-1}b,v,\bar{v})$ which has an integrable (uniformly in $s$) singularity at $\delta^{-1}b=v$; specifically $\delta^{-1}S_{s,-s}(\delta^{-1}\bar{b},\delta^{-1} w)=O(|w-v_\delta|^{-2s})$. Hence the contribution of $l_1$ to $\Sigma_2$ is $O(1)$. Now, fix $(bw)^\star \in l_2$, and note that the function
\[
   U_w(x):=\eta^2_x\left(\chi_{-s,l}(w,x)\delta^{-1}K^{-1}_{-s,s}(\delta^{-1}\bar{x},\delta^{-1}w)-\chi_{s,l}(w,x)\delta^{-1}K^{-1}_{s,-s}(\delta^{-1}\bar{x},\delta^{-1}w)\right)
\]
is discrete holomorphic (i.e., satisfies $K_\delta U_w\equiv 0$) in the ball $|\delta^{-1}x-\Re{v}|\leq \Im v/2$ (One can view this function as obtained by representing the multi-valued functions $x\mapsto K_{-s,s}(\delta^{-1}\bar{x},\delta^{-1}w)$ and $x\mapsto K_{s,-s}(\delta^{-1}\bar{x},\delta^{-1}w)$ by their branches w.r.t. a cut disjoint with that ball.) Therefore, by maximum principle, \[|U_w(b)|\leq \max_{x:|\delta^{-1}x-\Re v|\in [\Im v/2-1,\Im v/2]}|U_w(x)|=O(1),\] uniformly in $(bw)^\star \in l_2$, using  We conclude that $\Sigma_2 = O(1).$

To conclude the proof, we note that $s'_\delta=O(\lambda^{-\frac12}\sigma_\delta^{-\frac12})$ and
 \begin{equation}
 \label{eq: ssprime}
 s_\delta'(\lambda)s_\delta(\lambda)=\frac{1}{4\pi^2 \sigma_{\delta}}-\frac{\lambda}{6\pi^2\sigma_\delta^2}+O\left(\lambda^2\sigma_\delta^{-3}\right).
 \end{equation}
Plugging this and \eqref{eq: ssprime} and the estimates of $\Sigma_1$ and $\Sigma_2$ into \eqref{eq: dlogExpectation_interim}, and using \eqref{eq:asymp_of_mu}, we get
\begin{multline*}
 \frac{d}{d\lambda} \log \EE e^{-\lambda X_\delta}=\frac{\mu_\delta}{\sigma_\delta}-\frac{\log \delta^{-1}}{\pi^2\sigma_\delta}+\frac{2\lambda\log \delta^{-1}}{3\pi^2\sigma_\delta^2}+O\left(\sigma_\delta^{-1}+\frac{\lambda}{\sigma_\delta^2}+\frac{\lambda^2\log \delta^{-1}}{\sigma_\delta^3}+\lambda^{-\frac12}\sigma_\delta^{-\frac12}\right)\\=\lambda+O\left(\lambda^{-\frac12}(\log \delta^{-1})^{-\frac14}\right),
\end{multline*}
uniformly in $\lambda\in(0,R]$ for any $R$.
\end{proof}

We will use the following elementary analog of L\'evy continuity theorem.
\begin{lemma}
\label{lem:tightness}
Let $\{Y_\delta\}_{\delta>0},$ be a family of scalar random variables such that for any $\lambda\geq 0$, one has
\[
\mathbb{E}e^{-\lambda Y_\delta}\stackrel{\delta\to 0}{\longrightarrow}h(\lambda)
\] where  $h:[0,\infty)\to \RR $ is continuous. Then, $Y_\delta\stackrel{\delta\to 0}{\longrightarrow}Y$ in distribution, where the distribution of $Y$ is uniquely characterized by the condition $\mathbb{E}e^{-\lambda Y}= h(\lambda)$ for all $\lambda\geq 0$.
\end{lemma}
\begin{proof}
Since $\mathbb{P}(Y_\delta\leq -R)\leq e^{-R}\mathbb{E}e^{-Y_\delta}$, the family $Y_\delta$ is tight in $\mathbb{R}\cup \{+\infty\}.$  From any sequence $\delta_k\to 0$, one can extract a subsequence $\delta_{k(m)}$ such that $Y_{\delta_{k(m)}}$, converges in distribution, say to $Y$. For any $\lambda>0$, the function $x\mapsto e^{-\lambda x}$, extended by $0$ at $+\infty$, is continuous on $\mathbb{R}\cup \{+\infty\},$  moreover, the family $\{e^{-\lambda Y_\delta}\}_{\delta>0}$ is bounded in $L^2$, in particular, uniformly integrable. It follows that $\mathbb{E}e^{-\lambda Y}=h(\lambda)$. By continuity of $h$ and dominated convergence, \[1=h(0)=\lim_{\lambda\searrow 0}\mathbb{E}e^{-\lambda Y}=\mathbb{P}(Y<+\infty),\]that is, $Y$ in fact is $\mathbb{R}$-valued. Since the condition $\mathbb{E}e^{-\lambda Y}=h(\lambda)$ for all $\lambda\geq 0$ determines the law of $Y$ uniquely (by analytic continuation to $\Re \lambda>0$ and Laplace inversion), the claim follows.
\end{proof}
\begin{proof}[Proof of Theorem \ref{thma:main1}.]
Parts 1 and 2 are already proven in Lemma \ref{lemma:asymp_of_mu_delta_sigma_delta}. Integrating the bound in Proposition \ref{prop: log_der_Laplace}, one readily checks that $X_\delta$ satisfy the conditions of Lemma \ref{lem:tightness} with $h(\lambda)=\exp(\lambda^2/2)$, which is a Laplace transform of a standard Gaussian.
\end{proof}

\begin{rem}
Note that tightness of $X_\delta$ also follows from \eqref{eq:asymp_of_mu}. On the other hand, the above proof of Part 3 of Theorem \ref{thma:main1}(3) \emph{does not}, in fact, rely on \eqref{eq:asymp_of_mu} or other results in Section \ref{sec:Height_function_loop_statistics_and_monodromy}, insofar as one is willing to replace $\mu_\delta$ and $\sigma_\delta$ in the the statement by $-\frac{1}{\pi^2}\log \delta$ and $\sqrt{-\frac{2}{3\pi^2}\log \delta}$ respectively.
\end{rem}

\section{The double dimer nesting field and the \texorpdfstring{${\CLE_4}$}{CLE(4)} nesting field}
\label{sec:CLE_nesting}

Recall that $\vphi_\delta(x) = N_\delta(x) - \EE N_\delta(x)$ denotes the nesting field associated with the double-dimer model in $\CC_\delta^+$. The main goal of this section is to prove Theorem~\ref{thma:main2} asserting that $\vphi_\delta$ converges to the nesting field $\vphi$ of $\CLE_4$ introduced by Miller, Watson and Wilson~\cite{miller2015conformal}. Let us briefly describe our strategy before we go into details. To control nesting fields we consider their regularized versions $\varphi^\eps(x)=N(x,x+\eps) - \EE N(x,x+\eps)$ and $\varphi_\delta^\eps(x)=N_\delta(x,x+\eps) - \EE N_\delta(x,x+\eps)$, as defined in Section \ref{subsec:combinatorial_corresp_betwen_nesting_and_height}. The proof of the theorem comes from comparing $\vphi$ to $\vphi^\eps$, $\vphi^\eps$ to $\vphi^\eps_\delta$, and finally $\vphi^\eps_\delta$ to $\vphi_\delta$. We implement it as follows:

\begin{enumerate}
    \item In Section~\ref{subsec:def_of_mL_loc}, we review a toolbox for convergence of random fields in Sobolev spaces of negative regularity.
    \item In Section~\ref{subsec:nesting_fields_for_CLE(4)}, we estimate the error between $\vphi$ and $\vphi^\eps$. This is done by repeating Miller--Watson--Wilson arguments almost verbatim; still we include this part for completeness.
    \item In Section~\ref{sec:FieldConv}, we proceed by proving that $\vphi^\eps_\delta$ converge in distribution to $\vphi^\eps$. This is done using results of Kenyon, Dub\'edat, Basok--Chelkak and Bai--Wan on topological properties of double-dimer loops.
    \item In Section~\ref{subsec:double-dimer_nesting_fields}, we estimate the error between $\vphi^\eps_\delta$ and $\varphi_\delta$, and in Section~\ref{subsec:proof_main2}, we put everything together to complete the proof.
\end{enumerate}

\subsection{
Random fields belonging to a Sobolev class}
\label{subsec:def_of_mL_loc}

Given an open set $\Omega\Subset \CC$, denote by $\mC_c^\infty(\Omega)$ the set of smooth functions with compact support lying in $\Omega$. Denote by $\Distr[\Omega]$ the set of Schwartz distributions in $\Omega$. Recall that, given $\phi\in \mC^\infty_c(\CC)$ and $s\in \RR$, the ($\mL^2$) Sobolev norm $\|\phi\|_{\mH^s(\CC)}$ is defined as
\[
    \|\phi\|^2_{\mH^s(\CC)} = \int_\CC (|1 + |z|^2|)^s |\widehat{\phi}(z)|^2\,dz
\]
where $\widehat{\phi}$ is the Fourier transform of $\phi$. The Sobolev norm of those functions $\phi$ whose support is contained in a square $x_0 + (R,R)^2$ can be conveniently measured using the Fourier transform on a torus instead of the plane.
\begin{lemma}
\label{lemma:Sobolev_on_torus}
    Fix $R>0, z_0\in \CC$, put $\TT^2 = \CC/(2\pi\ZZ)^2$, and let $\alpha: \TT^2\to z_0 + (-R,R]^2$ be defined by $\alpha(z) = z_0 + R\pi^{-1} z$, where $\TT^2$ is identified with $(-\pi, \pi]^2$. Given $k\in \ZZ^2$ and a smooth function $\phi$ with $\supp \phi\subset z_0 + (-R,R)^2$ put
    \[
        \widehat{\alpha^\ast\phi}(k) = (2\pi)^{-2}\int_{\TT^2}e^{-iz\cdot k}(\phi\circ \alpha)(z)\,dz.
    \]
    There exists a constant $C$ depending on $R$ and $s$ continuously and independent on everything else such that
    \[
        \|\phi\|^2_{\mH^s(\CC)}\asymp_C \sum_{k\in \ZZ^2}(1 + |k|^2)^s |\widehat{\alpha^\ast\phi}(k)|^2,
    \]
    where $A\asymp_C B$ means $C^{-1}A \leq B \leq CA$.
\end{lemma}
\begin{proof}
    Follows from the discussion in~\cite[Chapter~4.3]{taylor1996partial}.
\end{proof}

Given $f\in \Distr[\Omega]$ and $K\Subset\Omega$ we define
\[
    \|f\|_{\mH^s(K)} = \sup\{f(\phi)\ \mid\ \phi\in \mC_c(K),\ \|\phi\|_{\mH^{-s}(K)} = 1\}
\]
and put
\[
    \mH^s_\loc(\Omega) = \{f\in \Distr[\Omega]\ \mid\ \forall K\Subset\Omega,\ \|f\|_{\mH^s(K)}<\infty\}.
\]
Let now $(X,\PP)$ be a probability space. A map $\vphi: X\to \Distr[\Omega]$ that is measurable with respect to the Borel sigma algebra on the right-hand side is called a \emph{random field}. We consider two random fields to be the same if they coincide almost surely.
Given a random field $\vphi$, consider the linear operator
\[
    \Adj_\vphi: \mC_c^\infty(\Omega)\to \{F: X\to \RR\}/_\sim,\qquad \Adj_\vphi(\phi)(x)\coloneqq \vphi(x)(\phi),
\]
where $F_1\sim F_2$ if they coincide almost surely. Note that $\Adj_\vphi(\phi)$ is Borel measurable for every $\phi\in \mC_c(\Omega)$.
Given $K\Subset \Omega$ put
\begin{equation}
    \label{eq:def_of_mA(K)}
    \begin{split}
        &\mA(K) \coloneqq\{A: \mC_c(K)\to \mL^2(X,\PP)\ \mid\ A \text{ is linear, }\|A\|_{\mL^\infty(K)\to \mL^2(X)} < \infty\},\\
        &\|A\|_{\mL^\infty(K)\to \mL^2(X)}^2 = \sup_{\|\phi\|_\infty = 1} \EE A(\phi)^2.
    \end{split}
\end{equation}
Also, put
\begin{equation}
    \label{eq:def_of_mAloc}
    \mA_\loc(\Omega) \coloneqq \{ A: \mC_c(\Omega)\to \mL^2(X,\PP)\ \mid\ \forall K\Subset \Omega,\ A\vert_{\mC_c(K)}\in \mA(K) \}.
\end{equation}
Using the operator norm on each $A(K)$ as a semi-norm on $\mA_\loc(\Omega)$ we can make $\mA_\loc(\Omega)$ to be a Fr\'echet space. It is easy to see that this space is metrizable and the metric is complete.

The next proposition asserts that each $A\in \mA_\loc(\Omega)$ has the form $\Adj_\vphi$ for some random field and, moreover, a suitable Sobolev norm of $\vphi$ can be controlled via the operator norm of $A$.

\begin{prop}
    \label{prop:representable_operators}
    For each $A\in \mA_\loc(\Omega)$ there is a unique (defined almost everywhere) $\vphi:X\to \Distr[\Omega]$ such that $A = \Adj_\vphi$. Moreover, for each open relatively compact $K\Subset \widetilde{K}\Subset \Omega$ and $\nu>0$, there exists a constant $C_{K, \widetilde{K},\nu}>0$ independent of $\vphi$ such that
        \[
            \EE\|\vphi\|_{\mH^{-1-\nu}(K)}^2 \leq C_{K,\widetilde{K},\nu}\|\Adj_\vphi\|_{\mL^\infty(\widetilde{K})\to \mL^2(X)}^2.
        \]
\end{prop}
\begin{proof}
Note that the uniqueness statement follows trivially from the inequality in the proposition and the fact that $\Adj_\vphi$ is linear in $\vphi$.

Fix $K,\widetilde{K}$ and pick a finite cover $\{B_m\}$ of $K$ by balls $B_m=B(z_m,\eps_m)$ such that $B(z_m,2\eps_m)\subset \widetilde{K}$. Let $\psi_m\in \mC_c^\infty(\Omega)$ be a smooth partition of unity subordinate to $B_m$, i.e., $\psi_m\in \mC^\infty$, $0\leq \psi_m\leq 1$, $\sum_m \psi_m\equiv 1$ on $K$ and $\supp \psi_m\subset B(z_m,2\eps_m)$. For a function $\phi \in \mC_c^\infty(K)$, we consider the Fourier expansion
\[
\left(\psi_m\right)^\frac12\phi=\sum_{k\in\ZZ^2} \alpha_{m,k} e_k,\]
where $e_k$ denotes the pullback of the function $e^{ix\cdot k}$ to $z_m+[-2\eps_m,2\eps_m]^2$ as in Lemma \ref{lemma:Sobolev_on_torus}. For each $m,k$ fix a function on $X$ representing $A(\psi_m^{\frac{1}{2}}e_k)$, and define $\vphi_{K, \widetilde{K}}:X\to \mathcal{D}'(K)$ by
\[
\vphi_{K, \widetilde{K}}(x)(\phi):=\sum_m\sum_{k\in\ZZ^2} \alpha_{m,k}A\left(\left(\psi_m\right)^\frac12 e_k\right)(x).
\]
By Cauchy--Schwarz, we have for every $x$
\begin{multline}
|\vphi_{K, \widetilde{K}}(x)(\phi)|\leq \sum_m \left(\sum_{k\in\ZZ^2} (\alpha_{m,k})^2 (1 + |k|^2)^{1+\nu}\right)^\frac12\left(\sum_{k\in\ZZ^2} \left|A\left(\left(\psi_m\right)^\frac12 e_k\right)(x)\right|^2(1 + |k|^2)^{-1-\nu}\right)^\frac12\\
\leq \sum_m\|\psi^\frac12_m\phi\|_{\mH^{1+\nu}(\TT^2)}\left(\sum_{k\in\ZZ^2} \left|A\left(\left(\psi_m\right)^\frac12 e_k\right)(x)\right|^2(1 + |k|^2)^{-1-\nu}\right)^\frac12,
\end{multline}
and the last factor is finite almost surely since it has a finite second moment:
\[
\EE\left(\sum_{k\in\ZZ^2} \left|A\left(\left(\psi_m\right)^\frac12 e_k\right)\right|^2(1 + |k|^2)^{-1-\nu}\right)\leq \|A\|_{\mL^\infty(\widetilde{K})\to \mL^2(X)}^2\sum_{k\in \ZZ^2}(1 + |k|^2)^{-1-\nu},
\]
see~\eqref{eq:def_of_mA(K)}. Since $\|\psi^\frac12_m\phi\|_{\mH^{1+\nu}(\TT^2)}\leq C_{\psi_m}\|\phi\|_{\mH^{1+\nu}(K)}$ (cf. Lemma~\ref{lemma:Sobolev_on_torus}), we conclude that $\|\vphi_{K, \widetilde{K}}\|_{\mH^{-1-\nu}(K)}$ is finite almost surely and satisfies $\EE\|\vphi_{K, \widetilde{K}}\|^2_{\mH^{-1-\nu}(K)}\leq C_{K,\widetilde{K},\nu}\|A\|_{\mL^\infty(\widetilde{K})\to \mL^2(X)}$.

Because of the above estimate, for a fixed $\phi$, we have almost surely
\[
A\left(\sum_m\left(\psi_m\right)^\frac12\sum_{|k|\leq R}\alpha_{m,k}e_k\right)(x)= \sum_m\sum_{|k|\leq R} \alpha_{m,k}A\left(\left(\psi_m\right)^\frac12 e_k\right)(x)\to \vphi_{K, \widetilde{K}}(x)(\phi).
\]
On the other hand, since $\sum_{|k|\leq R}\alpha_{m,k}e_k\to \left(\psi_m\right)^\frac12\phi$ in $L^\infty$, the left-hand side converges to $A(\phi)$ in $L^2$. We conclude that for a fixed $\phi$, we have $\vphi_{K, \widetilde{K}}(x)(\phi)=A(\phi)(x)$ almost surely.

Assume now that $K_1\Subset \widetilde{K}_1$ and $K_2\Subset \widetilde{K}_2$ are given and $K_1\subset K_2$. Picking a countable dense subset of $\mC_c^\infty(K_1)$, we infer that almost surely the restriction of $\vphi_{K_2, \widetilde{K}_2}$ to $\mC_c^\infty(K_1)$ coincides with $\vphi_{K_1,\widetilde{K}_1}$. It follows that the collection $\{\vphi_{K,\widetilde{K}}\}$ consistently defines a random field $\vphi: X\to \mH^{-1-\nu}_\loc(\Omega)$. This random field satisfies all the desired properties by construction.

\end{proof}

Proposition~\ref{prop:representable_operators} provides a correspondence between the space $\mA_\loc(\Omega)$ and a corresponding space of random fields whose `adjoint operators' have locally finite norm. More precisely, let us define
\begin{align}
    \label{eq:def_of_mA(A)_norm}
    &\|\vphi\|_{\Ff(K)}^2\coloneqq \sup_{\phi\in \mC_c^\infty(K),\ \|\phi\|_\infty = 1} \EE \vphi(\phi)^2 = \|\Adj_\vphi\|_{\mL^\infty(K)\to \mL^2(X)}^2,\\
    \label{eq:def_of_Ffloc}
    &\Ff_\loc(\Omega) = \{\vphi: X\to \Distr[\Omega]\ \mid\ \forall K\Subset \Omega,\ \|\vphi\|_{\Ff(K)} < \infty\}
\end{align}
and endow $\Ff_\loc(\Omega)$ with the topology generated by the semi-norms $\|\cdot\|_{\Ff(K)}$. We have
\begin{cor}
    \label{cor:Ff_is_complete}
    The topology on the space $\Ff_\loc(\Omega)$ is metrizable and the underlying metric space is complete. For each $\nu>0$ and $K\Subset \widetilde{K}\Subset \Omega$ there exists a constant $C_{K,\widetilde{K},\nu}>0$ such that
    \begin{equation}
        \label{eq:Sobolev_via_Ff}
        \EE \|\vphi\|^2_{\mH^{-1-\nu}(K)} \leq C_{K,\widetilde{K},\nu} \|\vphi\|^2_{\Ff(\widetilde{K})}\qquad \forall \vphi\in \Ff_\loc(\Omega).
    \end{equation}
    In particular, each $\vphi\in \Ff_\loc(\Omega)$ has its values in $\mH^{-1-\nu}_\loc(\Omega)$ almost surely.
\end{cor}
\begin{proof}
    Proposition~\ref{prop:representable_operators} implies that the map $\vphi\mapsto \Adj_\vphi$ is bijection between $\Ff_\loc(\Omega)$ and $\mA_\loc(\Omega)$. By the definition, this map preserves the norms corresponding to each $K\Subset \Omega$, thus, it preserves the topology. Note that the metric on $\Ff_\loc(\Omega)$ can be explicitly written as follows:
    \begin{equation}
        \label{eq:def_of_dFf}
        d_\Ff(\vphi_1,\vphi_2) = \sum_{n = 1}^{+\infty} 2^{-n}\min(\|\vphi_1 - \vphi_2\|_{\Ff(K_n)} , 1)
    \end{equation}
    where $K_1\subset K_2\subset\ldots$ is an arbitrary sequence of open relatively compact subsets exhausting $\Omega$. Finally,~\eqref{eq:Sobolev_via_Ff} follows from Proposition~\ref{prop:representable_operators}.
\end{proof}

\begin{lemma}
    \label{lemma:mL-norm_via_kernel}
    Assume that $\vphi\in \Ff_\loc(\Omega)$ is a random field and there exists a function $G\in \mL^1_\loc(\Omega\times \Omega)$ such that for each $\phi_1,\phi_2\in \mC_c^\infty(\Omega)$ we have
    \[
        \EE (\vphi(\phi_1)\vphi(\phi_2)) = \int_{\Omega\times \Omega}\phi_1(x)G(x,y)\phi_2(y)\,dxdy.
    \]
    Then for any $K\subset \Omega$ we have $\|\vphi\|_{\Ff(K)}^2 \leq \|G\|_{\mL^1(K\times K)}$.
\end{lemma}
\begin{proof}
Indeed, we have for any $\phi\in \mC_c^\infty(K)$
\[
\EE \vphi(\phi)^2= \int_{\Omega\times \Omega}\phi(x)G(x,y)\phi(y)\,dxdy\leq \|\phi\|_{\mL^\infty(K)}^2\cdot \|G\|_{L_1(K\times K)}.
\]
\end{proof}
We conclude this section by observing that convergence in the topology of $\Ff_\loc$ implies convergence in distribution in $H^{-1-\nu}_\loc.$

\begin{lemma}
\label{lemma:conv_in_mean_implies_conv_in_law}
    Given relatively compact open $K\Subset \widetilde{K}\Subset \Omega$, a bounded continuous function $F:H^{-1-\nu}(K)\to\RR$, and $L,\eps>0$, one can find $\rho>0$ such that if $\varphi_1$ and $\varphi_2$ are two random fields satisfying $\|\vphi_1\|_{\Ff(\widetilde{K})}<L$, $\|\vphi_2\|_{\Ff(\widetilde{K})}<L$, and $\|\varphi_1- \varphi_2\|_{\Ff(\widetilde{K})}<\rho$, then \[|\EE F(\varphi_1)-\EE F(\varphi_2)|<\eps.\]
\end{lemma}
\begin{proof}
Note that $\vphi_{1,2}\in \mH^{-1-\nu}(K)$ almost surely due to Corollary~\ref{prop:representable_operators}, hence we can define $F(\vphi_{1,2})$ using the restrictions of $\vphi_{1,2}$ to $\mC_c^\infty(K)$. Let $M=\sup F$. We can write, for any $R>0$,
\begin{multline*}
 |\EE F(\varphi_1)-\EE F(\varphi_2)|\leq \EE \left(|F(\varphi_1)-F(\varphi_2)|\indic[\varphi_1,\varphi_2\in B_{H^{-1-\nu/2}}(0,R)]\right)\\+2M \PP[\|\varphi_1\|_{H^{-1-\nu/2}(K)}>R]+2M\PP[|\varphi_2\|_{H^{-1-\nu/2}(K)}>R].
\end{multline*}
By Chebyshev inequality and Proposition \ref{prop:representable_operators}, we have for any $R,\hat{\rho}>0$,
\[
\PP[\|\varphi_{1,2}\|_{H^{-1-\nu/2}(K)}>R]\leq \frac{2MC_{K,\widetilde{K},\nu/2}L}{R}\quad \text{and} \quad \PP[\|\vphi_1-\vphi_2\|_{H^{-1-\nu}(K)}> \hat{\rho}]\leq \frac{C_{K,\widetilde{K},\nu}\|\vphi_1-\vphi_2\|_{L^\infty(K)}}{\hat{\rho}},
\]
hence we can write, for any $\hat{\rho}>0$ and $R>0$,
\begin{multline*}
 |\EE F(\varphi_1)-\EE F(\varphi_2)|\leq \sup\{|F(\psi_1)-F(\psi_2)|:\psi_1,\psi_2\in B_{H^{-1-\nu/2}}(0,R),\|\psi_1-\psi_2\|_{H^{-1-\nu}(K)}\leq \hat{\rho}\}\\
 +\frac{2M C_{K,\widetilde{K},\nu}\|\vphi_1-\vphi_2\|_{L^\infty(K)}}{\hat{\rho}}+\frac{4MC_{K,\widetilde{K},\nu/2}L}{R}.
\end{multline*}
Since $H^{-1-\nu/2}$ embeds compactly into $H^{-1-\nu}$, $F$ is equicontinuous on $B_{H^{-1-\nu/2}}(0,R)$. This means that, given $\eps>0$, we can first choose $R>0$ such that the third term above is smaller than $\eps/3$, then choose $\hat{\rho}>0$ such that the supremum is smaller than $\eps/3$, and finally ensure that the second term is smaller than $\eps/3$ by the choice of $\rho$.
\end{proof}

\subsection{Two-point approximation of the \texorpdfstring{${\CLE_\kappa}$}{CLE(kappa)}-nesting field}
\label{subsec:nesting_fields_for_CLE(4)}

Nesting fields for ${\CLE_\kappa}$ with $\kappa\in (8/3,8)$ were introduced in~\cite{miller2015conformal} via the following procedure. Let $\Ll = \Ll_{{\CLE_\kappa}}$ be the ${\CLE_\kappa}$ sampled in a proper simply-connected domain $D$. Fix an $\eps>0$ and define
\begin{equation}
    \label{eq:approx_nesting_via_balls}
    \begin{split}
        &N^{B\eps}(x) = \#\text{ of loops in $\Ll$ surrounding $B(x,\eps)$},\\
        &\vphi^{B\eps}(x) = N^{B\eps}(x) - \EE N^{B\eps}(x).
    \end{split}
\end{equation}
The main result of~\cite{miller2015conformal} asserts that for each $\nu>0$ the family $\vphi^{B\eps}$ converges almost surely to a conformally invariant random field $\vphi$ with values in $\mH^{-2-\nu}_\loc(D)$, which the authors called the~\emph{nesting field} of $\CLE_\kappa$. Besides the almost sure convergence, the arguments of~\cite{miller2015conformal} imply the following theorem. Recall the space $\Ff_\loc(D)$ introduced in~\eqref{eq:def_of_Ffloc} and Corollary~\ref{cor:Ff_is_complete}.

\begin{thmas}[Miller--Watson--Wilson]
    \label{thmas:miller_nesting_approximation}
    The family $\{\vphi^{B\eps}\}_{\eps > 0}\subset \Ff_\loc(\Omega)$ has a limit $\vphi\in \Ff_\loc(\Omega)$ as $\eps\to 0$. In particular, for each $K\Subset D$ and $\nu>0$ the family $\{\vphi^{B\eps}\}_{\eps>0}$ converges in distribution in~$\mH^{-1-\nu}(K)$.
\end{thmas}
\begin{proof}
    By~\cite[Theorem~4.1]{miller2015conformal} and Lemma~\ref{lemma:mL-norm_via_kernel} the sequence $\{\vphi^{B\eps}\}_{\eps>0}$ is Cauchy. By Corollary~\ref{cor:Ff_is_complete} the space $\Ff_\loc(\Omega)$ is a complete metric space, thus $\lim_{\eps\to 0+}\vphi^{B\eps}$ exists. By Lemma~\ref{lemma:conv_in_mean_implies_conv_in_law} this implies the weak convergence with respect to the $\mH^{-1-\nu}(K)$ norm for each each $K\Subset D$ and $\nu>0$. In particular, the limit coincides with the nesting field $\vphi$.
\end{proof}

For our purposes we need to consider a slightly different approximation of the nesting field $\vphi$. Similarly to the discrete definition in Section~\ref{subsec:combinatorial_corresp_betwen_nesting_and_height}, let $\vphi^\eps=N(x,x+\eps)-\EE(x,x+\eps)$, where $N(x,x+\eps)$ is the number of $\CLE_\kappa$ loops surrounding both $x$ and $x+\eps$; note that
it has exponentially decaying tails and thus finite expectation.

\begin{prop}
    \label{prop:two-point_approximation}
    For each $\eps>0$ we have $\vphi^\eps\in \Ff_\loc(D)$. For any relatively compact open $K\subset D$ we have
    \[
        \lim\limits_{\eps\to 0+} \|\vphi - \vphi^\eps\|_{\Ff(K)} = 0.
    \]
\end{prop}

Our strategy of proving Proposition~\ref{prop:two-point_approximation} is to adapt the proof of the key lemma~\cite[Lemma~4.8]{miller2015conformal}, replacing $\vphi^{B\eps_1} - \vphi^{B\eps_2}$ with $\vphi^{B\eps} - \vphi^\eps$. This requires reproving some of the preparatory lemmas from~\cite{miller2015conformal} as well. We begin by recalling another technical result from~\cite{miller2015conformal}. Given a proper simply-connected domain $D\subset \CC$, denote by $G_D(x,y)$ the Green's function for this domain. According to~\cite[Theorem~1.3]{miller2015conformal}, for any $j>0$ we have
\begin{equation}
    \label{eq:jth_moment_of_N(x,y)}
    \EE N(x,y)^j = \left(\nu_{\mathrm{typical}} G_D(x,y)\right)^j + O((G_D(x,y) + 1)^{j-1}),
\end{equation}
where $\nu_{\mathrm{typical}}>0$ is certain absolute constant determined by $\kappa$ and the normalization of $G_D$.

Given $x\in D$ we enumerate all the loops in $\Ll_{{\CLE_\kappa}}$ encircling $x$ consequently starting from the outermost one. Let $U_{x,j}$ denote the connected component of the complement of the $j$-th loop which contains $x$, with $U_{x,0}=D$. Given $y\in D$ not equal to $x$ denote by $S_{x,y}$ the index of the first loop encircling $x$ but not $y$. Denote by $\CR(x;D)$ the conformal radius of $x$ in the domain $D$.

\begin{lemma}
    \label{lemma:CR_of_separating_loop}
    For any $\kappa\in (8/3,8)$ there exist constant $C,\alpha>0$ such that the following holds. Let $D$ be a simply-connected domain and $x,y\in D$ be two distinct points. Then for any $\eps>0$ we have
    \[
        \PP[\CR(x; U_{x,S(x,y)}) \leq \eps]\leq C\left(\frac{\eps}{\min(|x-y|, \CR(x; D))}\right)^\alpha.
    \]
\end{lemma}
\begin{proof}
    Denote $r = \min(|x-y|, \CR(x; D))$, and assume without loss of generality that $\eps\leq r$ and $\CR(x;D) = 1$. Given $t>0$ define $\tau(t) = \min\{ j\geq 0\ \mid\ \CR(x; U_{x,j}) \leq t \}$. Note that $\xi_j=\log \CR(x; U_{x,j}) - \log \CR(x; U_{x,j+1})$ are non-negative i.i.d. variables with exponentially decaying right tails. Let $\mathcal{T}=\log r-\log \eps;$ note that $\left(\frac{\eps}{r}\right)^\alpha=e^{-\alpha \mathcal{T}}$. We claim that for a suitable $\beta>0$ and $\alpha>0$, we have
\begin{equation}
        \label{eq:CRs2}
        \PP\left[\tau(\eps) - \tau(r)\geq \beta\mathcal{T} \right] = 1 - O\left( \frac{\eps}{r} \right)^\alpha.
    \end{equation}
Indeed, by union bound, we have
\[
\PP\left[\tau(\eps) - \tau(r)< \beta\mathcal{T} \right]\leq \PP\left[-\log\CR(x; U_{x,\tau(r)})+\log r \geq \mathcal{T}/2 \right]+\PP\left[\sum_{\tau(r)\leq j\leq \tau(r)+\beta\mathcal{T}}\xi_j\geq\mathcal{T}/2\right].
\]
Provided that $\beta\EE\xi_k<\frac12,$ both terms are  exponentially small in $\mathcal{T}$, the first one by the ``overshoot estimate" ~\cite[Lemma~2.8]{miller2016extreme}, and the second one by conditioning on $\tau(r)$ and applying Chernoff bound. This gives \eqref{eq:CRs2}.

    Let $f: U_{x,\tau(r)}\to \DD$ be a conformal mapping such that $f(x) = 0$, and let $g=f^{-1}$. We have $|g'(0)|=\CR(x; U_{x,\tau(r)})\leq r$, therefore, by Koebe distortion theorem, $|g(z)-x|\leq |z|r(1-|z|)^{-2}$. If $y\in U(x; \tau(r))$, then plugging in  $z=f(y)$ and taking into account that $r\leq|x-y|$ yields  $|f(y)|\geq \frac{3-\sqrt{5}}{2}$. Note that if $\tau(\eps)-\tau(r)\geq \beta \mathcal{T}$ and $S(x,y)\geq \tau(\eps)$, then $y\in U_{x,\tau(r)}$ and $N_{U_{x,\tau(r)}}(x,y)\geq \beta \mathcal{T}$. Hence
    \begin{multline*}
        \PP[\CR(x; U_{x,S(x,y)}) \leq \eps] = \PP[S(x,y)\geq \tau(\eps)] \leq \\
        \leq \EE\left[\PP[y\in U_{x,\tau(r)},\,N_{U_{x,\tau(r)}}(x,y)\geq \beta \mathcal{T}|U_{x,\tau(r)}]\right]+\PP[\tau(\eps) - \tau(r)< \beta\mathcal{T} \}]  \leq \\
        \leq \sup_{|z|\geq (3-\sqrt{5})/2}\PP\left[N_\DD(0,z) \geq \beta \mathcal{T} \right] + O\left( \frac{\eps}{r} \right)^\alpha = O\left( \frac{\eps}{r} \right)^\alpha,
    \end{multline*}
using \eqref{eq:CRs2} and the fact that $N_\DD(0,z)$ has exponential right tail with parameters depending only on $\kappa$ and the lower bound on $|z|$.
\end{proof}

\begin{lemma}
    \label{lemma:bound_on_loops_separating_pair_from_the_third}
    For any $j\geq 1$ and $\kappa\in (8/3,8)$ there exist constants $C,\alpha>0$ such that the following holds. Let $D\subset \CC$ be a proper simply-connected domain, $x,y\in D$ be distinct points such that $|x-y|\leq \min(\CR(x;D), \CR(y; D))$. Let $N^{B\eps}(x,y)$ denote the number of loops in a ${\CLE_\kappa}$ sample in $D$ encircling $x$ and $y$, but not intersecting their $\eps$-neighborhoods $B_\eps(x)$ and $B_\eps(y)$. Then for any $\eps > 0$ we have
    \[
        \EE \left( N(x,y) - N^{B\eps}(x, y) \right)^j \leq C\left(\frac{\eps}{|x-y|}\right)^\alpha \cdot \left(\log\left(\frac{\eps}{|x-y|}\vee 1\right) + 1\right)^j.
    \]
\end{lemma}
\begin{proof}
    Applying Lemma~\ref{lemma:CR_of_separating_loop} we get
    \begin{equation}
        \label{eq:ccv}
        \PP[N(x,y) > N^{B\eps}(x, y)]  = O\left(\frac{\eps}{|x-y|}\right)^\alpha.
    \end{equation}
    Let us condition on this event. Let $\gamma$ denote the first loop encircling $x,y$ and intersecting their $\eps$-neighborhood, and let $U$ be the Jordan domain encircled by $\gamma$. For definiteness, assume that $\gamma\cap B_\eps(x)\neq \emptyset$ and hence $\CR(x,U)\leq 4\eps$. Let $f:U\to \DD$ be a conformal mapping such that $f(x) = 0$. Applying Koebe's distortion theorem as in the proof of Lemma \ref{lemma:CR_of_separating_loop}, we get
    \[
        |f(y)|\geq c\cdot \left(|x-y|\eps^{-1}\wedge 1\right)
    \]
    where $c>0$ is some absolute constant.
    It follows that, on the event $N(x,y) > N^{B\eps}(x, y)$, we have by conformal invariance and \eqref{eq:jth_moment_of_N(x,y)}
    \begin{multline*}
    \EE[(N(x,y) - N^{B\eps}(x, y))^j|U]\leq \sup_{|z|\geq c\cdot \left(|x-y|\eps^{-1}\wedge 1\right)}\EE (N_{{\CLE_\kappa}\text{ in }\DD}(0, z)+1)^j\\
    =O\left(G_{\DD}\left(0,c\cdot \left(|x-y|\eps^{-1}\wedge 1\right)\right)+1\right)^j = O\left(\log\left(\frac{\eps}{|x-y|}\vee 1\right) + 1\right)^j.
    \end{multline*}
    Combining this with \eqref{eq:ccv} yields the claim.
\end{proof}

\begin{lemma}
    \label{lemma:conformal_change_of_vphieps}
    For any $j\geq 1$ and $\kappa\in (8/3,8)$ there exist constants $C,\alpha>0$ such that the following holds. Let $D\subset \CC$ be a proper simply-connected domain, $x\in D$ be such that $\CR(x;D) = 1$ and $\eps>0$ be given. Let $f$ be a conformal mapping from $D$ to another domain $f(D)$ such that $f'(x) = 1$, and assume that $f(x)+\eps\in f(D)$. Let $\Ll$ denote the ${\CLE_\kappa}$ sampled in $D$. Then we have
    \[
        \EE \left| N_\Ll(x,x+\eps) - N_{f(\Ll)}(f(x), f(x) + \eps) \right|^j \leq C\eps^\alpha.
    \]
\end{lemma}
\begin{proof}
    We first fix some small numerical constant $\eps_0$ and assume that $\eps<\eps_0$; in that case $f(x)+\eps\in f(D)$ by Koebe $\frac14$ theorem. Put $y = x + \eps$ and $z = f^{-1}(f(x)+ \eps)$. Note that Koebe $\frac14$ and distortion theorems guarantee that $f(u)=f(x)+u-x+O(|u-x|^2)$ and $f^{-1}(u)=x+u-f(x)+O(|u-f(x)|^2)$ when $|u-x|<\frac{1}{10}$ (respectively, $|u-f(x)|<\frac{1}{10}$), with an absolute constant in the $O(\cdot)$. Hence, $|y-z|\leq A\eps^2$, with an absolute constant $A$. By the conformal invariance it is enough to bound
    \[
        \EE \left| N_\Ll(x,y) - N_\Ll(x,y,z) \right|^j + \EE \left| N_\Ll(x,z) - N_\Ll(x,y,z) \right|^j.
    \]
    Koebe's $\frac14$ theorem ensures that if $\eps_0$ was chosen so that $\eps_0+A\eps_0^2<\frac{1}{100}$ and $A\eps^2_0 <\eps_0$, then $\max(|x-y|, |x-z|)< \min(\CR(x;D), \CR(y; D), \CR(z; D))$. Thus, we can apply Lemma~\ref{lemma:bound_on_loops_separating_pair_from_the_third} to get the desired bound, since any loop encircling $x,y$, but not $z$ must pass at distance at most $ A\eps^2<\eps$ from $y$.

    If $\eps\geq\eps_0$, we can use that $N_\Ll(x,y)$ and $N_{f(\Ll)}(f(x), f(x) + \eps)$ both are stochastically dominated by some geometric variable whose parameter depends on $\kappa$ and $\eps_0$ only.
\end{proof}

Given $x,y\in D$, let $\Sigma_{x,y}$ denote the sigma-algebra generated by all loops in ${\CLE_\kappa}$ surrounding $x$ or $y$ and having the index at most $S(x,y)$ (note that $S(x,y)=S(y,x)=N(x,y)+1$).

\begin{lemma}
    \label{lemma:nesting_is_locally_constant}
    Let $D\subset \CC$ be proper and simply connected. For any $\kappa\in (8/3,8)$ there exist $C,\alpha>0$ such that for any two distinct points $x,y\in D$ and $0<\eps\leq r \coloneqq \min(|x-y|, \CR(x;D))$
    \[
       \EE\left[ \EE\left( \varphi^{B\eps}(x) - \varphi^\eps(x) \vert
       \Sigma_{x,y}
       \right)^2\right] \leq C\left(\frac{\eps}{r}\right)^\alpha.
    \]
\end{lemma}
\begin{proof}
    The lemma is analogous to~\cite[Lemma~4.7]{miller2015conformal}, and our proof will consist of adapting the arguments therein. According to~\cite[proof~of~Lemma~4.7]{miller2015conformal} there exists a coupling between two ${\CLE_\kappa}$ in $D$, denoted by $\Ll$ and $\widetilde{\Ll}$, such that the following holds. Let $U_{x,S(x,y)}$ and $\widetilde{U}_{x,\widetilde{S}(x,y)}$ be the domains defined as above for $\Ll$ and $\widetilde{\Ll}$ respectively. Then we have
    \begin{enumerate}
        \item The random domains $U_{x,S(x,y)}$ and $\widetilde{U}_{x,\widetilde{S}(x,y)}$ and the parts of $\Ll$, $\widetilde{\Ll}$ outside of them are independent.
        \item The two random walks
        \[
            \{ X_k = -\log \CR(x; U_{x,k}) \}_{k\geq 0}\qquad \text{and} \qquad \{ \widetilde{X}_k = -\log \CR(x; \widetilde{U}_{x,k}) \}_{k\geq 0}
        \]
        are coupled in such a way that there exist random indices $K,\tilde{K}$ such that $X_{K+j}=X_{\tilde{K}+j}$ for all $j\geq 0$, and moreover for some constants $C,c>0$ depending only on $\kappa$ we have

        \begin{equation}
            \label{eq:nilc1}
            \PP[X_K\leq \max(X_{S(x,y)}, \widetilde{X}_{\widetilde{S}(x,y)}) + M]\geq 1 - Ce^{-cM},
        \end{equation}

        \item Let $f: U_{x,K}\to \widetilde{U}_{x,\widetilde{K}}$ be the conformal map normalized by $f(x) = x$, $f'(x) = 1$. Then the set of loops of $\widetilde{\Ll}$ inside $\widetilde{U}_{x,\widetilde{K}}$ is the image of the set of loops of $\Ll$ inside $U_{x,K}$ under $f$.
    \end{enumerate}
    Following~\cite[proof~of~Lemma~4.7]{miller2015conformal} we put
    \[
        \Delta = \EE[N^{B\eps}(x) - N(x,x+\eps) \vert
        \Sigma_{x,y}
        ] - \EE[\widetilde{N}^{B\eps}(x) - \widetilde{N}(x,x+\eps) \vert
        \widetilde{\Sigma}_{x,y}
        ]
    \]
    and observe that it is enough to bound $\EE\Delta^2$, which in turn can be achieved by estimating
    \[
        \begin{split}
            &\EE\left[\EE( N^{B\eps}(x) - K - \widetilde{N}^{B\eps}(x) + \widetilde{K}\vert
            \Sigma_{x,y},\widetilde{\Sigma}_{x,y})^2
            \right],\\
            &\EE\left[\EE( N(x,x+\eps) - K - \widetilde{N}(x,x+\eps) + \widetilde{K}\vert
            \Sigma_{x,y},\widetilde{\Sigma}_{x,y})^2\right]
        \end{split}
    \]
    separately. The first expectation was already bounded in~\cite[proof~of~Lemma~4.7]{miller2015conformal}. We deal with the second one by repeating the same arguments with~\cite[Lemma~2.7]{miller2015conformal} replaced by Lemma~\ref{lemma:conformal_change_of_vphieps}. Let us sketch these arguments for the sake of completeness.

    -- On the event $A = \{\CR(x; U_{x,K})> \sqrt{r\eps}\}$ we use that $N(x,x+\eps) - K$ (respectively, $\widetilde{N}(x,x+\eps) - \widetilde{K}$) is the number of loops surrounding both $x,x+\eps$ in $\CLE_\kappa$ in the random domain $U_{x,K}$ (respectively, in its image under $f$). Hence, Lemma~\ref{lemma:conformal_change_of_vphieps} readily gives a bound of the form $O\left(\frac{\eps}{r}\right)^\alpha$.

    -- The properties of the coupling give also a bound of the same form on $1- \PP(A)$. Indeed, applying Lemma~\ref{lemma:CR_of_separating_loop} and using~\eqref{eq:nilc1} for $M = -\frac{1}{4}\log\frac{\eps}{r}$ we get
    \[
        1-\PP(A) \leq 2\PP[\CR(x; U_{x,S(x,y)}) \leq r^{3/4}\eps^{1/4}] + \PP[X_K\geq \max(X_{S(x,y)}, \widetilde{X}_{\widetilde{S}(x,y)}) + M] = O\left(\frac{\eps}{r}\right)^\alpha.
    \]
    for $\alpha>0$ small enough.

    -- We will conclude by
    Cauchy--Schwarz
    if we can bound
    \[
        \EE\left[\EE( N(x,x+\eps) - K - \widetilde{N}(x,x+\eps) + \widetilde{K}\vert
        \Sigma_{x,y},\widetilde{\Sigma}_{x,y}
        )^4
        \right]
    \]
    by a constant depending on $\kappa$ only. To this end, consider again two cases. First, on the event $B=\{\CR(x,U_{x,K})> 10\eps\}$, we have by Koebe's theorem $x+\eps \in U_{x,K}$ and $x+\eps \in \widetilde{U}_{x,\widetilde{K}}$, and applying Lemma~\ref{lemma:conformal_change_of_vphieps} as above readily gives the desired bound. On $B^c$, we write, using Jensen inequality,
\begin{multline*}
\EE\left[\EE( N(x,x+\eps) - K - \widetilde{N}(x,x+\eps) + \widetilde{K}\vert
        \Sigma_{x,y},\widetilde{\Sigma}_{x,y}
        )^4  \indic_{B^c}
        \right]
    \\
        \leq 8\EE\left[( N(x,x+\eps) - K
        )^4 \indic_{B^c}
        \right]+8\EE\left[( \widetilde{N}(x,x+\eps) - \widetilde{K}
        )^4\indic_{B^c}
        \right]
\end{multline*}
    We will estimate the first term; the second one is similar. It is enough to see that the random variable $(N(x,x+\eps) - K)\indic_{B^c}$ has exponentially decaying tails with the rate depending on $\kappa$ only. Indeed,
    \[
        \PP[N(x,x+\eps) - K \geq M,\ B^c ] \leq\PP[N_{U_{x,K}}(x,x+\eps)\geq M\ \vert B^c,\ x+\eps\in U_{x,K}]
    \]
    decays exponentially as $M\to +\infty$ because if $f:U_{x,K}\to \DD$ is conformal and $f(x) = 0$, then $|f(x+\eps)|$ is bounded from below by an absolute constant. On the other hand,
    \[
        \PP[K-N(x,x+\eps) \geq M] \leq \PP[K - S(x,y)\geq M/2] + \PP[S(x,y)-N(x,x+\eps)\geq M/2].
    \]
    The first term decays exponentially as $M\to +\infty$ by the construction of the coupling. To estimate the second one, consider the uniformizing map $f: U_{x, N(x,x+\eps)+1}\to \DD$ such that $f(x) = 0$. Then, on the event $y\in U_{x, N(x,x+\eps)+1}$, we have $|f(y)|$ bounded from below by some absolute constant due to Koebe's distortion theorem. It follows that
    \[
       \PP[S(x,y)-N(x,x+\eps)\geq M/2] = \PP[N_{U_{x, N(x,x+\eps)+1}}(x,y) \geq M/2-1]
    \]
    decays exponentially.
\end{proof}

\begin{proof}[Proof of Proposition~\ref{prop:two-point_approximation}]
    It is enough to prove that
    \begin{equation}
        \label{eq:tpa1}
        \lim\limits_{\eps\to 0+} \|\vphi^{B\eps} - \vphi^\eps\|_{\Ff(K)} = 0
    \end{equation}
    for any fixed relatively compact open $K\subset \CC^+$. In order to do this we will estimate the kernel
    \[
        G_\eps(x,y) = \EE \bigl[(\vphi^{B\eps}(x) - \vphi^\eps(x))(\vphi^{B\eps}(y) - \vphi^\eps(y))\bigr],
    \]
for $x,y\in K$; we will assume $2\eps<\dist(K,\partial \Omega)\wedge 1$. The estimate is similar to~\cite[proof~of~lemma~4.8]{miller2015conformal} but easier. First, note that by~\eqref{eq:jth_moment_of_N(x,y)}, we have
    $
    \EE (\varphi^\eps(z))^2=O(|\log \eps|)
    $ for $z\in K$. Similarly, by~\cite[Corollary 3.2]{miller2015conformal}, we have the same estimate for $\varphi^{B\eps}(z)$. Therefore, Cauchy--Schwarz implies
    \begin{equation}
        \label{eq:tpa2}
        |G_\eps(x,y)|\lesssim |\log \eps|,\qquad x,y\in K.
    \end{equation}

    Let us now assume that $|x-y|\geq \eps$. The random variables $\vphi^{B\eps}(x) - \vphi^\eps(x)$ and $\vphi^{B\eps}(y) - \vphi^\eps(y)$ are conditionally independent given $\Sigma_{x,y}$, therefore we have
    \begin{multline}
        \label{eq:tpa3}
        |G_\eps(x,y)| \leq \EE\bigl[\EE[(\vphi^{B\eps}(x) - \vphi^\eps(x))\vert \Sigma_{x,y}]^2\bigr]^{1/2}\cdot \, \EE\bigl[\EE[(\vphi^{B\eps}(y) - \vphi^\eps(y))\vert \Sigma_{x,y}]^2\bigr]^{1/2} \\
        \leq C\left(\frac{\eps}{\min(|x-y|, \CR(x;D))}\right)^\alpha,\qquad |x-y| > \eps.
    \end{multline}
    where $C$ and $\alpha$ are as in Lemma~\ref{lemma:nesting_is_locally_constant}. We conclude from~\eqref{eq:tpa2} and~\eqref{eq:tpa3} that
    \[
        \|G_\eps\|_{\mL^1(K\times K)}\lesssim -\eps\log \eps
    \]
    with some constant depending on $K$ and $\kappa$ only. Eq.~\eqref{eq:tpa1} now follows from Lemma~\ref{lemma:mL-norm_via_kernel}.
\end{proof}

\subsection{Cylindrical events and the convergence of \texorpdfstring{$\varphi_\delta^\eps$}{mdelta} to \texorpdfstring{$\varphi^\eps$}{mdelta}.}
\label{sec:FieldConv}

From now on let us put $\kappa = 4$, fix $D = \CC^+$ and denote by $\vphi$ and $\vphi^\eps$ the nesting field of $\CLE_4$ and its two point approximation respectively, see Section~\ref{subsec:nesting_fields_for_CLE(4)} for details. We also denote by $N(x,x+\eps)$ the number of $\CLE_4$ loops surrounding $x$ and $x+\eps$. Recall that $\Ll_\delta$ denotes the double-dimer loop ensemble in $\CC_\delta^+$ and $\vphi_\delta^\eps = N_\delta(x,x+\eps) - \EE N_\delta(x,x+\eps)$, where $N_\delta(x,x+\eps)$ is the number of loops in $\Ll_\delta$ surrounding $x,x+\eps$. In this section we will show that for a fixed $\eps>0$ the field $\vphi^\eps_\delta$ converges to $\vphi^\eps$ as $\delta\to 0+$.

To this end we will use the following result, which was established in the series of papers~\cite{kenyon2014conformal, DubedatDoubleDimers, BasokChelkak, bai2023crossing} of Kenyon, Dub\'edat, Basok--Chelkak and Bai--Wan. Given a set of distinct points $\lambda_1,\dots, \lambda_n\in \CC^+$ a macroscopic lamination $\Gamma$ in $\CC^+\smm\{\lambda_1,\dots, \lambda_n\}$ is the free homotopy class of any finite collection of simple pairwise non-intersecting loops such that each of them encircles at least two of the points $\lambda_1,\dots, \lambda_n$. Given $\lambda_1,\dots, \lambda_n$ let $l_1,\dots, l_n$ be simple disjoint paths connecting these points with the real line inside $\CC^+$ chosen as follows: if $\Re \lambda_i$ are all different, then $l_i$'s are chosen to be straight vertical segments; if $\Re \lambda_i = \Re\lambda_j, i\neq j$, then each of $l_i$ and $l_j$ can be chosen to be a concatination of horizontal and vertical segments. Given paths $l_1,\dots, l_n$, the complexity of any loop $\gamma$ in $\CC^+\smm\{\lambda_1,\dots, \lambda_n\}$ is defined to be the minimal number of intersections between $l_1\cup\dots\cup l_n$ and $\gamma'$ among all $\gamma'$ freely homotopic to $\gamma$. The complexity of a macroscopic lamination is, by definition, the sum of complexities of the underlying loops. Note that the notion of a complexity introduced in~\cite{BasokChelkak} is slightly different, but one can easily show that these two notions provide comparable answers.

Given a loop ensemble $\Ll$ in $\CC^+$ and a macroscopic lamination $\Gamma$ in $\CC^+\smm\{\lambda_1,\dots, \lambda_n\}$ we say that $\Ll\sim \Gamma$ if $\Gamma$ is obtained from $\Ll$ after removing all the loops that encircle at most one of $\lambda_1,\dots, \lambda_n$.

\begin{thmas}
    \label{thmas:on_convergence_of_cyl_prob}
    Let $\CC^+_\delta = \CC^+\cap \delta \ZZ^2$ and $\lambda_1,\dots,\lambda_n\in \CC^+$ be distinct points. Denote by $\Ll_\delta$ the double-dimer loop ensemble in $\CC^+_\delta$ and by $\Ll_{\CLE_4}$ the $\CLE_4$ in $\CC^+$. The following holds:
    \begin{enumerate}
        \item Then for each $R>0$ there is a constant $C>0$ such that for each macroscopic lamination $\Gamma$ in $\CC^+\smm\{\lambda_1,\dots, \lambda_n\}$ we have
        \begin{equation}
            \label{eq:superexp_estimate_CLE4}
            R^{|\Gamma|} \PP[\Ll_{\CLE_4}\sim \Gamma] \leq C.
        \end{equation}
        Moreover, this estimate holds uniformly in $(\lambda_1,\dots, \lambda_n)$ staying in any compact in $(\CC^+)^n\smm\mathrm{diags}$.
        \item For each $R>0$ we have
        \begin{equation}
            \label{eq:convergence_of_cyl_prob}
            \lim\limits_{\delta \to 0}R^{|\Gamma|} \PP[\Ll_\delta\sim \Gamma] = R^{|\Gamma|} \PP[\Ll_{\CLE_4}\sim \Gamma]
        \end{equation}
    uniformly in $\Gamma$ and $(\lambda_1,\dots, \lambda_n)$ staying in any compact in $(\CC^+)^n\smm\mathrm{diags}$.
    \end{enumerate}
\end{thmas}
\begin{proof}
    The estimate~\eqref{eq:superexp_estimate_CLE4} is the central result of~\cite{bai2023crossing}. Note that it is proven for all $\kappa\leq 4$. The convergence~\eqref{eq:convergence_of_cyl_prob} is the main output of the series of works~\cite{kenyon2014conformal, DubedatDoubleDimers, BasokChelkak}, in particular, in~\cite{BasokChelkak} it is proven based on~\eqref{eq:superexp_estimate_CLE4}, see~\cite[Corollary~1.7]{BasokChelkak}.
\end{proof}

We have the following corollary:
\begin{lemma}
    \label{lemma:vphi_delta_eps_to_vphi_eps}
    For any relatively compact open $K\subset \CC^+$, any $\nu>0$, and any $\eps>0$ small enough, the fields $\vphi_\delta^\eps$ converge as $\delta\to 0$ to $\vphi^\eps$ in distribution with respect to the topology of $\mH^{-1-\nu}(K)$.
\end{lemma}
\begin{proof}
    Fix $\nu>0$, $\eps>0$, and $K\Subset \Omega$. For $n>0$ put
    \[
    \begin{split}
        &\vphi_\delta^{\eps, n} = N_\delta(x,x+\eps) \indic_{N_\delta(x,x+\eps) \leq n} - \EE \left[N_\delta(x,x+\eps) \indic_{N_\delta(x,x+\eps) \leq n}\right],\\
        &\vphi^{\eps, n} = N(x,x+\eps) \indic_{N(x,x+\eps) \leq n} - \EE\left[ N(x,x+\eps) \indic_{N(x,x+\eps) \leq n}\right].
    \end{split}
    \]
Observe that when specialized to $n=2$, $\lambda_1=x$ and $\lambda_2=x+\eps$, Theorem~\ref{thmas:on_convergence_of_cyl_prob} gives, for any $R>0$, the bounds $\PP[N(x,x+\eps)\geq n] \leq C\cdot R^{-n}$ and $\PP[N_\delta(x,x+\eps)\geq n] \leq C\cdot R^{-n}$ with $C$ uniform in $x\in K$ and in $\delta>0$. Therefore, the same estimates hold also for $\EE [N^2_\delta(x,x+\eps)\indic_{N_\delta(x,x+\eps)\geq n}]$ and $\EE [N^2(x,x+\eps)\indic_{N(x,x+\eps)\geq n}].$ Using Lemma~\ref{lemma:mL-norm_via_kernel} and Cauchy--Schwarz inequality, we deduce that
    \[
        \begin{split}
            &\lim\limits_{n\to +\infty} \|\vphi^{\eps,n} - \vphi^\eps \|_{\Ff(K)} = 0,\\
            &\lim\limits_{n\to +\infty} \|\vphi_\delta^{\eps,n} - \vphi_\delta^\eps \|_{\Ff(K)} = 0\qquad \text{uniformly in }\delta>0.
        \end{split}
    \]
    Thus, by Lemma~\ref{lemma:conv_in_mean_implies_conv_in_law} it is enough to prove that for each $n>0$ the field $\vphi_\delta^{\eps,n}$ converges to $\vphi^{\eps,n}$ weakly with respect to the topology of $\mH^{-1-\nu}(K)$.

    Put
    \[
        \Kk_n = \mathrm{Cl}_{\mH^{-1-\nu}(K)}\,\{\vphi\in \mH^{-1-\nu}(K)\cap \mL^\infty(K)\ \mid\ \|\vphi\|_{\mL^\infty(K)} \leq n\}
    \]
    where $\mathrm{Cl}_{\mH^{-1-\nu}(K)}$ denotes the closure in the norm topology in $\mH^{-1-\nu}(K)$. Note that for each $n>0$ the set $\Kk_n$ is a compact subset of $\mH^{-1-\nu}$, and $\vphi_\delta^{\eps,n}\in \Kk_n$. For each $\phi\in \mC_c^\infty(K)$ define $F_\phi: \Kk_n\to \RR$ by
    \[
        F_\phi(\vphi) = \vphi(\phi) = \int_K \vphi(x)\phi(x)\,dx.
    \]
    Then $F_\phi$ is continuous in the topology of $\mH^{-1-\nu}(K)$ and, by Stone--Weierstrass theorem, the algebra with unit finitely generated by these functions is dense in $\mC(\Kk_n)$. Thus, in order to prove that $\vphi_\delta^{\eps,n}$ converge to $\vphi^{\eps,n}$ in the weak topology, it is enough to prove that for each $\phi_1,\dots, \phi_k\in \mC_c^\infty(K)$ we have
    \[
        \lim\limits_{\delta\to 0} \EE \prod_{i = 1}^k F_{\phi_i}(\vphi_\delta^{\eps,n}) = \EE \prod_{i = 1}^k F_{\phi_i}(\vphi^{\eps,n}).
    \]
    But we have, by dominated convergence and Theorem~\ref{thmas:on_convergence_of_cyl_prob},
    \begin{multline*}
        \EE \prod_{i = 1}^k F_{\phi_i}(\vphi_\delta^{\eps,n}) = \int_K\ldots \int_K \left( \EE \prod_{i = 1}^k\vphi_\delta^{\eps, n}(x_i) \right) \prod_{i = 1}^k\phi_i(x_i)\,dx_i \xrightarrow[\delta\to 0]{} \\
        \xrightarrow[\delta\to 0]{}\int_K\ldots \int_K \left( \EE \prod_{i = 1}^k\vphi^{\eps, n}(x_i) \right) \prod_{i = 1}^k\phi_i(x_i)\,dx_i = \EE \prod_{i = 1}^k F_{\phi_i}(\vphi^{\eps,n}).
    \end{multline*}

\end{proof}

\subsection{Two-point approximation of the double-dimer nesting field}
\label{subsec:double-dimer_nesting_fields}

Let $\psi_\delta,\psi_\delta^\eps$ be the fields $\psi,\psi^\eps$ sampled with respect to $\Ll_\delta$ as defined in Section~\ref{subsec:combinatorial_corresp_betwen_nesting_and_height}, that is,
    \[
        \psi_\delta(x) = h_\delta(x)^2 - \EE h_\delta(x)^2,\qquad \psi^\eps_\delta(x) = h_\delta(x)h_\delta(x+\eps) - \EE h_\delta(x)h_\delta(x+\eps),
    \]
where $h_\delta$ is the double-dimer height function. The goal of this section is to estimate $\|\vphi^\eps_\delta - \vphi_\delta\|_{\Ff(K)}$ uniformly in $\delta>0$. This is achieved in Proposition~\ref{prop:phi_delta_eps_to_phi_delta}. In order to obtain this estimate we compare $\vphi_\delta$ and $\vphi_\delta^\eps$ with $\psi_\delta$ and $\psi_\delta^\eps$ respectively. The control on $\|\psi_\delta - \psi_\delta^\eps\|_{\Ff(K)}$ is obtained by using Lemma~\ref{lemma:mL-norm_via_kernel} and the results of Section~\ref{sec:Height_function_loop_statistics_and_monodromy}.

We begin by proving the following:

\begin{lemma}
    \label{lem:psi_delta_eps_to_psi_delta}
    For any fixed relatively compact open $K\subset \CC^+$, there are functions $\alpha_0(\eps)$ and $\beta_0(\delta)$ such that $\lim_{\delta\to 0}\beta_0(\delta)=\lim_{\eps\to 0}\alpha_0(\eps)=0,$ and for any $\eps,\delta\in (0,1/2)$ one has
    \[
        \|\psi_\delta - \psi_\delta^\eps\|^2_{\Ff(K)} \leq \alpha_0(\eps)+\beta_0(\delta).
    \]
\end{lemma}
\begin{proof}
    Given two vertices $x,y$ of the dual lattice $(\CC^+_\delta)^\ast$ define
    \begin{equation}
        \label{eq:ppd1}
        \begin{split}
            &\Psi_\delta(x,y) = \EE(\psi_\delta(x)\psi_\delta(y)),\\
            &\Psi_{\delta,\eps}(x,y) = \EE(\psi_\delta(x)\psi_\delta^\eps(y)),\\
            &\Psi_{\delta,\eps,\eps}(x,y) = \EE(\psi_\delta^\eps(x)\psi_\delta^\eps(y)).
        \end{split}
    \end{equation}
    and $\Psi_{\delta,\eps}^{\leftrightarrow}(x,y)=\Psi_{\delta,\eps}(y,x)$. According to Lemma~\ref{lemma:mL-norm_via_kernel}, we have
    \begin{equation}
        \label{eq:ppd2}
        \|\psi_\delta - \psi_\delta^\eps\|^2_{\Ff(K)}\leq\|\Psi_\delta - \Psi_{\delta, \eps} - \Psi^{\leftrightarrow}_{\delta,\eps} + \Psi_{\delta,\eps,\eps}\|_{L^1(K\times K)}.
    \end{equation}
    Denote $G_\delta(x,y)=\EE \left[h_\delta(x)h_\delta(y)\right].$ By Lemma~\ref{lemma:Wick_rule_dimer_height}, we have as $\delta\to 0$
    \begin{equation}
        \label{eq:ppd3}
        \begin{split}
            &\Psi_\delta(x,y) = 2G^2_\delta(x,y) + o(1),\\
            &\Psi_{\delta,\eps}(x,y) = 2G_\delta(x,y)G_\delta(x,y+\eps) + o(1),\\
            &\Psi_{\delta,\eps,\eps}(x,y) =
        G_\delta(x,y)G_\delta(x+\eps,y+\eps)+G_\delta(x+\eps,y)G_\delta(x,y+\eps)+o(1).
        \end{split}
    \end{equation}
   with $o(1)$ uniform over $x,y\in K$ and independent of $\eps$. By the result of \cite{KenyonGFF} (or simply keeping track of the constant terms in the proof of Lemma \ref{lemma:asymptotics_of_h(0)2}), we have for distinct $x,y\in K$, \[\EE G_\delta(x,y)\stackrel{\delta\to 0}{\longrightarrow} G(x,y):=-\frac{1}{\pi^2}\log\left|\frac{x-y}{x-\bar{y}}\right|.\] The bound of Lemma \ref{lemma:asymptotics_of_h(0)2}, together with the dominated convergence theorem, ensures that the convergence also holds in $L^2(K\times K)$. Therefore, using Cauchy-Schwarz, we see that
   \[
\|\Psi_\delta - \Psi_{\delta, \eps} - \Psi_{\delta,\eps}^{\leftrightarrow} + \Psi_{\delta,\eps,\eps}\|_{L^1(K\times K)}=\|\Psi - \Psi_{\eps} - \Psi^{\leftrightarrow}_{\eps} + \Psi_{\eps,\eps}\|_{L^1(K\times K)}+o(1),\quad \delta\to 0,
   \]
   where $o(1)$ is independent of $\eps$, and we denote $\Psi(x,y)=2G^2(x,y),$ $\Psi_\eps(x,y)=2G(x,y)G(x,y+\eps)$, $\Psi^{\leftrightarrow}_\eps(x,y)=\Psi_\eps(y,x)$, and $\Psi_{\eps,\eps}(x,y) =
        G(x,y)G(x+\eps,y+\eps)+G(x+\eps,y)G(x,y+\eps)$. Since it is also elementary to check that each of $G(x,y+\eps)$, $G(x+\eps,y)$ and $G(x+\eps,y+\eps)$ converges to $G(x,y)$ as $\eps\to 0$ in $L^2(K\times K)$,
 the result follows.
\end{proof}

Denote
\[
        \begin{split}
            \Pi_\eps(x,y) &= \EE\left[P(x,y)  - P_\eps(x,y) - P_\eps(y,x)+ P_{\eps,\eps}(x,y)\right],\\
            \Pi_{\delta,\eps}(x,y) &=\EE\left[P_\delta(x,y)  - P_{\delta, \eps}(x,y) - P_{\delta, \eps}(y,x)+ P_{\delta,\eps,\eps}(x,y)\right]
        \end{split}
    \]
where $P,P_\eps,P_{\eps,\eps}$ (respectively, $P_\delta,P_{\delta,\eps},P_{\delta,\eps,\eps}$) are as in Lemma~\ref{lemma:nesting_via_height} with respect to $\CLE{}_4$ in $\CC^+$ (respectively, with respect to the double dimer loop model in $\CC^+_\delta$).
\begin{lemma}
    \label{lemma:norm_of_Q_eps}
    For any relatively compact open $K\subset \Omega$, we have
    \begin{equation}
    \label{eq:Pi_norm_goes_to_zero}
    \|\Pi_\eps\|_{L^1(K\times K)}\stackrel{\eps\to 0}{\longrightarrow}0,
    \end{equation}
    and there exist functions $\alpha_1(\eps)$ and $\beta_1(\eps,\delta)$ such that
    \begin{equation}
    \label{eq:Pi_delta_norm}
    \|\Pi_{\delta,\eps}\|_{L^1(K\times K)}\leq \alpha_1(\eps)+\beta_1(\eps,\delta),
    \end{equation}
and  $\lim_{\eps\to 0}\alpha_1(\eps)=0$ and $\lim_{\delta\to 0} \beta_1(\eps,\delta)=0$ for every $\eps>0$.
\end{lemma}
\begin{proof}
    Recall from Remark \ref{rem:P_well_defined_for_CLE} the properties of $P,P_\eps,P_{\eps,\eps}$. We can expand $\Pi_{\eps}(x,y)$ into a linear combination of $\EE[N(A)N(B)]$ and $\EE N(A),$ where $A,B\subset\{x,y,x+\eps,y+\eps\}$, apply Cauchy-Schwarz on each term, and note that each $N(A)$ can be upper bounded by one of $N(x,y)$, $N(x,y+\eps)$, $N(x+\eps,y)$, $N(x+\eps,y+\eps)$.  Applying~\eqref{eq:jth_moment_of_N(x,y)}, we get
    \begin{equation}
        \label{eq:nQe1}
        |\Pi_\eps(x,y)|\lesssim (|\log |x-y||+1)^2 + (|\log |x-y-\eps||+1)^2 + (|\log |x-y+\eps||+1)^2.
    \end{equation}
    On the other hand, when $x,y$ are far apart, say $|x-y|\geq 4\eps$, we observe that for every subset $A$ which contains at least one of $x,x+\eps$ and at least one of $y,y+\eps$, we have, in the notation of Lemma~\ref{lemma:bound_on_loops_separating_pair_from_the_third}, $|N(A)-N(x,y)|\leq |N(x,y)-N^{B_\eps}(x,y)|.$ Therefore, applying Cauchy-Schwarz, Lemma~\ref{lemma:bound_on_loops_separating_pair_from_the_third} and~\eqref{eq:jth_moment_of_N(x,y)}, we get
    \[
    \left|\EE [N(A)N(B)]-\EE N^2(x,y)\right|\lesssim \left(\frac{\eps}{|x-y|}\right)^\alpha\cdot (|\log |x-y||+1),
    \]
    and similarly for $\left|\EE [N(A)]-\EE N(x,y)\right|.$ Since replacing each $N(\cdot)$ with $N(x,y)$ in the definition of $\Pi_\eps$ yields identically zero, we have
    \begin{equation}
        \label{eq:nQe2}
        |\Pi_\eps(x,y)| \lesssim \left(\frac{\eps}{|x-y|}\right)^\alpha\cdot  (|\log |x-y||+1),\qquad |x-y|>4\eps.
    \end{equation}
    Integrating~\eqref{eq:nQe1} for $|x-y|<4\eps$ and~\eqref{eq:nQe2} for $|x-y|\geq 4\eps$, we get \eqref{eq:Pi_norm_goes_to_zero}.

    To prove \eqref{eq:Pi_delta_norm}, we write
    \begin{multline*}
         \|\Pi_{\delta,\eps}\|_{L^1(K\times K)}\leq \|\Pi_{\eps}\|_{L^1(K\times K)}\\
        + \int\limits_{|x-y|<4\eps,\ (x,y)\in K}|\Pi_{\delta,\eps}(x,y)|\,dxdy + \int\limits_{|x-y|\geq 4\eps,\ (x,y)\in K}|\Pi_{\delta,\eps}(x,y) - \Pi_\eps(x,y)|\,dxdy.
    \end{multline*}
       The second integral goes to zero as $\delta\to 0$, since given the set $K$ and $\eps>0$, we have
    \begin{equation}
        \label{eq:Qde1}
        \lim\limits_{\delta\to 0+} G_{\delta,\eps}(x,y) = G_\eps(x,y)
    \end{equation}
    uniformly in $x,y\in K,\ |x-y|\geq 4\eps,$ due to Theorem~\ref{thmas:on_convergence_of_cyl_prob}. To estimate the first integral, observe that we have, independently of $\delta$, the same bound for $\Pi_{\delta,\eps}$ as in \eqref{eq:nQe1}, with the same proof except that we use Lemma \ref{lemma:moments_of_Nxy} instead of \eqref{eq:jth_moment_of_N(x,y)}. Integrating this bound concludes the proof.
\end{proof}

\begin{prop}
    \label{prop:phi_delta_eps_to_phi_delta}
    For any relatively compact open $K\subset \CC^+$ there exist functions $\alpha(\eps),\beta(\eps,\delta)$ such that $\lim\limits_{\eps\to 0+}\alpha(\eps) = 0$, for any fixed $\eps\in (0,1/2)$ we have $\lim\limits_{\delta\to 0+}\beta(\eps,\delta) = 0$ and
    \[
        \|\vphi_\delta - \vphi_\delta^\eps\|^2_{\Ff(K)} \leq \alpha(\eps) + \beta(\eps,\delta)
    \]
    for all $\eps\in (0,1/2)$ and $\delta>0$.
\end{prop}
\begin{proof}
    Note that, due to Lemma~\ref{lemma:nesting_via_height}, we have
    \[
       \EE[(\vphi_\delta(x) - \vphi_\delta^\eps(x))(\vphi_\delta(y) - \vphi_\delta^\eps(y))] = \EE[(\psi_\delta(x) - \psi_\delta^\eps(x))(\psi_\delta(y) - \psi_\delta^\eps(y))]-\Pi_{\delta,\eps}(x,y).
    \]
    The result now follow directly from Lemmas \ref{lemma:mL-norm_via_kernel}, \ref{lem:psi_delta_eps_to_psi_delta} and \ref{lemma:norm_of_Q_eps}.
\end{proof}

\subsection{Proof of Theorem~\ref{thma:main2}}
\label{subsec:proof_main2}
Given a relatively compact $K\subset \CC^+$, a bounded, continuous function $F:H^{-1-\nu}(K)\to \RR,$ and $\eps>0$, we can write
\[
|\EE F(\varphi_\delta)-\EE F(\varphi)|\leq |\EE F(\varphi_\delta)-\EE F(\varphi_\delta^\eps)|+|\EE F(\varphi^\eps_\delta)-\EE F(\varphi^\eps)|+|\EE F(\varphi^\eps)-\EE F(\varphi)|.
\]
Let $\epsilon >0$ be given. By Proposition \ref{prop:two-point_approximation} and Lemma \ref{lemma:conv_in_mean_implies_conv_in_law}, we can ensure that $|\EE F(\varphi^\eps)-\EE F(\varphi)|<\epsilon/3$ by choosing $\eps$ small enough. By Proposition \ref{prop:phi_delta_eps_to_phi_delta} and Lemma \ref{lemma:conv_in_mean_implies_conv_in_law}, by taking $\eps>0$ small enough, we can ensure that $|\EE F(\varphi_\delta)-\EE F(\varphi_\delta^\eps)|<\epsilon/3$ for all small enough $\delta$. Fixing this $\eps$, by Lemma \ref{lemma:vphi_delta_eps_to_vphi_eps}, we have that $|\EE F(\varphi^\eps_\delta)-\EE F(\varphi^\eps)|<\epsilon/3$ provided that $\delta$ is small enough. Since $\epsilon$ and $F$ are arbitrary, this shows that $\varphi_\delta\to\varphi$ in distribution in $H^{-1-\nu}(K)$ for each $\nu>0$ and $K$.

It is now standard to derive convergence in distribution in $H^{-1-\nu}_\loc(\CC^+)$. Let $\mathcal{U}\subset H^{-1-\nu}_\loc(\CC^+)$ be open; by regularity of measures, given $\eps>0$, we can choose $\mathcal{K}\subset \mathcal{U}$ compact such that $\PP[\varphi\in \mathcal{U}\setminus \mathcal{K}]<\eps$. Every point $x$ of $\mathcal{K}$ has a neighborhood of the form $\{\psi:\psi|_{H^{1+\nu}(K)}\in U\}$ contained in $\mathcal{U},$ where $K\subset \CC^+$ is open and relatively compact, and $U\subset H^{-1-\nu}(K)$ is open. Picking a finite sub-cover, we can find an open set $\mathcal{U}'$ of the above form such that $\mathcal{K}\subset \mathcal{U}'\subset \mathcal{U}$. Since $\varphi_\delta\to\varphi$ in distribution in $H^{-1-\nu}(K)$, Portmanteau theorem gives
\[
\lim\inf \PP(\varphi_\delta\in \mathcal{U})\geq \lim\inf \PP(\varphi_\delta\in \mathcal{U}')\geq \PP(\varphi\in \mathcal{U}')\geq \PP(\varphi \in \mathcal{U})-\eps.
\]
Since $\eps>0$ is arbitrary, by Portmanteau theorem we have $\varphi_\delta\to\varphi$ in distribution in $H^{-1-\nu}_\loc(\CC^+)$.

\printbibliography

\end{document}